\documentclass[10pt,a4paper]{article}
\usepackage[utf8]{inputenc}
\usepackage{geometry}
\usepackage{graphicx}
\usepackage{amssymb}
\usepackage{epstopdf}
\usepackage{mathrsfs}
\usepackage{amsmath}
\usepackage{amsthm}
\begin{document}
\newtheorem{theoreme}{Theorem}
\newtheorem{ex}{Example}
\newtheorem{definition}{Definition}
\newtheorem{lemme}{Lemma}
\newtheorem{remarque}{Remark}
\newtheorem{exemple}{Example}
\newtheorem{proposition}{Proposition}
\newtheorem{corolaire}{Corollary}
\newtheorem{hyp}{Hypothesis}
\newtheorem*{rec}{Recurrence Hypothesis}
\newcommand\Bound{\partial\overline{X}}
\newcommand\bdf{b_\xi}
\newcommand\Euc{\mathbb{R}^d}
\newcommand\COIn{\tilde{\chi}}
\newcommand\Comp{\overline{X}}
\newcommand\Dir{\varphi}
\newcommand\Lag{\mathcal{L}_0}
\newcommand\ins{\gamma_{uns}}
\newcommand\guil{`}
\newcommand\ext{\mathcal{W}_{0}}
\newcommand\brut{A_7}
\newcommand\inter{\mathcal{W}_2}
\newcommand\close{\mathcal{W}_1}
\newcommand\tins{N_{uns}}
\newcommand\din{\delta_{in}}
\newcommand\inco{\mathcal{DE}_-}
\newcommand\se{\epsilon_{sec}}
\newcommand\Nlag{N_{lag}}
\newcommand\cmin{c_{min}}
\newcommand\ind{J_a^{h,r}}
\newcommand\sor{\mathcal{T}}
\newcommand\secur{\varepsilon_0}
\newcommand\diam{\varepsilon_6}
\newcommand\tim{T_0}
\newcommand\zone{\overline{a}}
\newcommand\exit{N_\epsilon}
\newcommand\wait{N_1}
\newcommand\cli{\varepsilon_7}
\newcommand\nuag{\chi_{7}}
\newcommand\rain{\Pi_7}
\newcommand\sui{A_1^{\leq N}}
\newcommand\symp{\mathcal{C}_0}
\newcommand\zeit{t_2}
\newcommand\petit{\varepsilon_2}
\newcommand\temps{t_5}
\newcommand\set{\mathcal{W}_3}
\newcommand\taille{\varepsilon_2}
\newcommand\bel{>}
\newcommand\R{\mathbb{R}}
\title{Distorted plane waves in chaotic scattering}
\author{Maxime Ingremeau}

\maketitle

\begin{abstract}
In this paper we provide a precise description of distorted plane waves for semiclassical Schr\"odinger
operators under the assumption that the classical trapped set is hyperbolic and that a certain
topological pressure (a quantity defined using thermodynamical formalism) is negative.
 Distorted plane waves are generalised eigenfunctions of the Schr\"odinger operator which differ
from free plane waves, $ e^{i \langle x, \xi \rangle/h} $, by an outgoing term. Under our assumptions
we show that they can be written as a convergent sum of Lagrangian states. That provides a description of their semiclassical defect measures in the spirit of quantum ergodicity and extends results of Guillarmou--Naud obtained for hyperbolic quotients to our setting.
\end{abstract}

\section{Introduction}

In this paper, we will consider on $\mathbb{R}^d$ a semiclassical Hamiltonian of the form
\begin{equation*}P_h=-h^2\Delta + V(x), ~~V\in C_c^\infty (\mathbb{R}^d).
\end{equation*}

We will study the “distorted plane waves”, or “scattering states”  associated to $P_h$. They are a family of functions $E^{\xi}_h\in C^\infty(\mathbb{R}^d)$ with parameter
$\xi\in \mathbb{S}^d$ (the
direction of propagation of the incoming wave) which are generalized eigenfunctions of $P_h$, that is to say, they satisfy the differential equation
\begin{equation}\label{boite}
(P_h - 1)E^{\xi}_h=0,
\end{equation}
but which are not in $L^2(\mathbb{R}^d)$ (since $P_h$ has no embedded eigenvalues in $\mathbb{R}^+$).

These distorted plane waves resemble the actual plane waves, in the sense that we may write
\begin{equation}\label{chti}
E_h^{\xi}(x) = e^{\frac{i}{h}x\cdot\xi} + E^\xi_{out},
\end{equation}
where, $E_{out}$ is outgoing in the sense that it satisfies the Sommerfeld radiation condition:
\begin{equation}\label{sommerfeld}
\lim \limits_{|x|\rightarrow \infty} |x|^{(d-1)/2} \Big{(} \frac{\partial}{\partial |x|} - \frac{i}{h} \Big{)} E^\xi_{out}(x) = 0.
\end{equation}

One can show (see for instance \cite[\S 2]{Mel} or \cite[\S 4]{dyatlov2016mathematical})  that for any $\xi\in \mathbb{S}^{d-1}$ and $h>0$, there exists a unique function $E_h^\xi$ satisfying conditions (\ref{boite}), (\ref{chti}) and (\ref{sommerfeld}).

Condition (\ref{sommerfeld}) may be equivalently stated by asking that $E^\xi_{out}$ is the image of a function in
$C_c^\infty(\mathbb{R}^d)$ by the outgoing resolvent
$(P_h-(1+i0)^2)^{-1}$, or by asking that $E_{out}^\xi$ may be put in the form
\begin{equation*}
E^\xi_{out}(x)= e^{i|x|/h} |x|^{-\frac{1}{2}(d-1)} \Big{(} a^\xi_h(\omega) + O\Big{(}\frac{1}{|x|}\Big{)} \Big{)},
\end{equation*}
where $\omega=x/|x|$. The function $a_h(\xi,\omega):=a_h^\xi(\omega)$ is called the \emph{scattering amplitude}, and is the integral kernel of the scattering matrix minus identity. The scattering amplitude, and hence the distorted plane waves, are central objects in scattering theory.

The aim of this paper is to discuss the behaviour of distorted plane waves in the semiclassical limit $h\rightarrow 0$.
Distorted plane waves can be seen as an analogue, on manifolds of infinite volume, of the eigenfunctions of a Schrödinger operator on a compact manifold. It is therefore natural to ask questions similar to those in the compact case: what can be said about the semiclassical measures of distorted plane waves ? About the behaviour of their $L^p$ norms as $h\rightarrow 0$ ? About their nodal sets and nodal domains ? 

The answer to these questions will depend in a drastic way on the properties of the underlying \emph{classical dynamics}. Let us define the classical Hamiltonian by
\begin{equation*}p(x,\xi)= |\xi|^2+V(x),
\end{equation*}
and the layer of energy $1$ as
$$\mathcal{E}=\{\rho\in T^*\mathbb{R}^d; ~ p(\rho)=1\}.$$ Note that this is a non-compact set, but its intersection with
any fibre $T^*_xX$ is compact.

We also denote, for each $t\in \mathbb{R}$, the Hamiltonian flow generated
by $p$ by
$\Phi^t : T^*\mathbb{R}^d\longrightarrow T^*\mathbb{R}^d$.
For $\rho\in \mathcal{E}$, we will say that $\rho\in \Gamma^\pm$ if $\{\Phi^t(\rho), \pm t\leq 0\}$
is a bounded subset of $T^*\R^d$; that is to say, $\rho$ does not “go to
infinity”, respectively in the past or
in the future. The sets $\Gamma^\pm$ are called respectively the
\textit{outgoing} and \textit{incoming} tails (at energy $1$).

The \textit{trapped set} is defined as
\begin{equation}\label{ensemblecapte}
K:=\Gamma^+\cap \Gamma^-.
\end{equation}
It is a flow invariant set, and it is compact, because $V$ is compactly
supported.

If the trapped set is empty, then we can easily describe the distorted plane waves in the semiclassical limit. Namely, one can show (cf. \cite[\S 5.1]{DG}) that $E_h^\xi$ is a \emph{Lagrangian (WKB) state}. Furthermore, for any $\chi\in C_c^\infty(\mathbb{R}^d)$, the norm $\|\chi E_h^\xi\|_{L^2}$ is bounded independently of $h$.

However, if the trapped set is non-empty, the distorted plane waves may not be bounded uniformly in $L^2_{loc}$ as $h\rightarrow 0$. Actually, $\|\chi E_h^\xi\|_{L^2}$ could grow exponentially fast as $h\rightarrow 0$. If we want this quantity to remain bounded uniformly in $h$, we must therefore make some additional assumptions on the classical dynamics. Let us now detail these assumptions.

\subsubsection*{Hypotheses on the classical dynamics}
\begin{itemize}
\item \emph{Hyperbolicity assumption}:
In the sequel, we will suppose that the potential $V$ is such that the trapped set contains no fixed point, and is a \emph{hyperbolic set}. We refer to section \ref{averse} for the definition of a hyperbolic set. The potential in figure \ref{exhyp} is an example of such a potential.

\begin{figure}\label{exhyp}
    \center
   \includegraphics[scale=0.35]{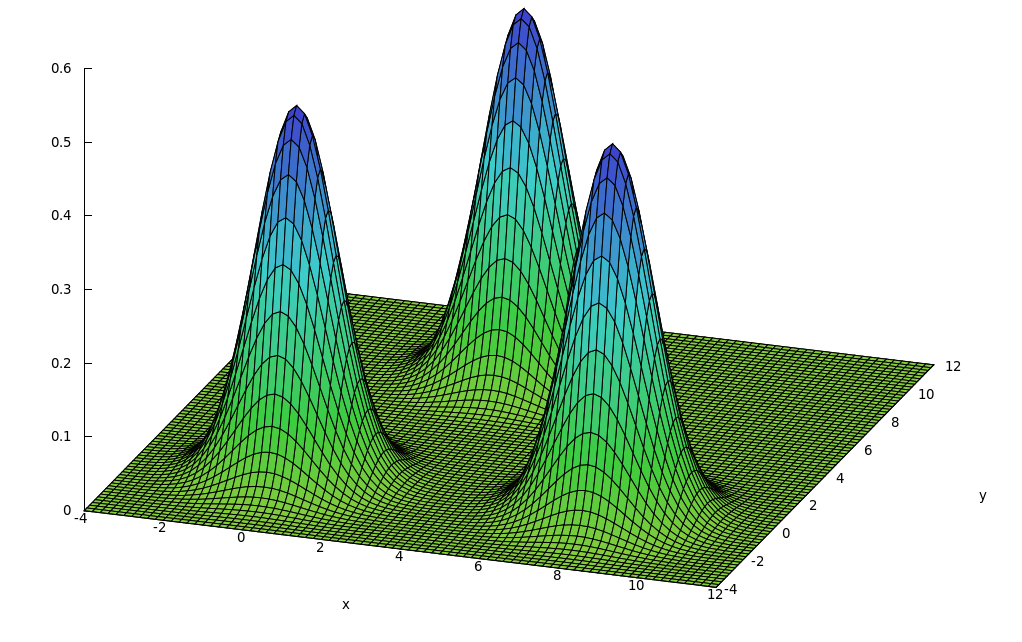}
    \caption{An example of a potential on $(\mathbb{R}^2,g_{flat})$ such that the dynamics is hyperbolic on the trapped set in some energy range. (See \cite[Appendix C]{Sj90} for details.)}
\end{figure}

\item \emph{Topological pressure assumption}: For our result on distorted plane waves to hold, we must also make the
assumption (Hypothesis \ref{Husserl}) that the
topological pressure associated to half the logarithm of the unstable
Jacobian of the flow on $K$ is negative.
The definition of the topological pressure will be recalled in section
\ref{press}. Hypothesis \ref{Husserl} roughly says that the system is “very open”. One should note that in dimension 2, this
condition is
equivalent to the fact that the Hausdorff dimension of $K$ is strictly
smaller than 2. In the three-bumps potential of figure \ref{exhyp},  this condition is satisfied if the three bumps are far enough from each other, but it is not satisfied if the bumps are close to each other.

\item \emph{Transversality assumption}: Our last assumption does not concern directly the classical dynamics, but the Lagrangian manifold\footnote{By a Lagrangian manifold, we mean a $d$-dimensional submanifold of a $2d$-dimensional symplectic manifold, on which the symplectic form vanishes. We will allow Lagrangian manifolds to have boundaries, and to be disconnected.}
\begin{equation}\label{swing}
\Lambda_\xi:=\{(x,\xi), x\in \mathbb{R}^d\}.
\end{equation}
Note that the plane wave $e^{\frac{i}{h} x\cdot \xi}$ is a Lagrangian state associated with the Lagrangian manifold $\Lambda_\xi$.

We need to make a \emph{transversality assumption} on $\Lambda_\xi$.
This assumption roughly says that the direction $\xi$ defining $\Lambda_\xi$ is such that the incoming tail $\Gamma^-$ and
$\Lambda_\xi$ intersect transversally. We postpone
the precise statement of this assumption to Hypothesis \ref{Happy} in section \ref{scotiabank}.
This assumption is probably generic in $\xi$, although we don't know how to prove it. In \cite{Ing2}, we will show that it is always satisfied for every $\xi$, when we consider geometric scattering on a manifold of non-positive curvature.
\end{itemize}

\subsubsection*{Statement of the results}

In Theorem \ref{ibrahim}, we will give a precise description of $E_h^{\xi}$ as a sum of WKB states, under the assumptions above. Since the precise statement of the theorem is a bit technical, we postpone it to section \ref{statementMain}, and only state two important consequences of this result.

The first one is a bound analogous to what we would get in the non-trapping case.
\begin{theoreme}\label{davies}
 Suppose that the Hypothesis \ref{sieste} on hyperbolicity  holds, that the topological pressure Hypothesis \ref{Husserl} is satisfied, and that $\xi\in
\mathbb{S}^{d-1}$ is such that $\Lambda_\xi$ satisfies Hypothesis \ref{Happy} of transversality. 

Let $\chi\in C_c^\infty (X)$. Then there exists a constant $C_{\xi,\chi}$ independent of $h$ such that for any $h\bel 0$, we have
\begin{equation}\label{heatwave}
\|\chi E_h^{\xi}\|_{L^2}\leq C_{\xi,\chi}.
\end{equation}
\end{theoreme}

\begin{remarque}
The bound (\ref{heatwave}) could not be obtained directly from resolvent estimates. Indeed, as we will see in section \ref{defdistorted}, the term $E_{out}$ in (\ref{chti}) can be written as the outgoing resolvent $(P_h- (1+i0)^2)^{-1}$ applied to a term which is compactly supported, and whose $L^2$ norm is $O(h)$. Therefore, we have a priori that $\|\chi E_h^\xi\|_{L^2}\leq O(h) \|\chi (P_h- (1+i0)^2)^{-1}\chi\|_{L^2\rightarrow L^2}$, as least if the support of $\chi$ is large enough. But under hypotheses \ref{sieste} and \ref{Husserl}, it is known since \cite{NZ} (see Theorem \ref{weasel}) that $\|\chi (P_h- (1+i0)^2)^{-1}\chi\|_{L^2\rightarrow L^2}\leq C \frac{|\log h|}{h}$, and such estimates are sharp in the presence of trapping (see \cite{bonyburqramondminoration}.) Such a priori estimates would therefore only give $\|\chi E_h^\xi\|_{L^2}\leq C |\log h |$.
\end{remarque}

Our next result concerns the \emph{semiclassical measure} of $E_h^\xi$. Consider on $T^*\mathbb{R}^d$ the measure $\mu^\xi_0$ given by
\begin{equation*}
\mathrm{d}\mu_0^\xi(x,v)= \mathrm{d}x \delta_{v= \xi}.
\end{equation*}
The measure $\mu^\xi_0$ is the semiclassical measure associated to $e^{\frac{i}{h} x\cdot \xi}$, in the sense that for any $\psi\in C_c^\infty(T^*\mathbb{R}^d)$ and any $\chi\in C_c^\infty(\mathbb{R}^d)$, we have
\[\lim\limits_{h\rightarrow 0} \langle Op_h(\psi) \chi e^{\frac{i}{h} x\cdot \xi}, \chi e^{\frac{i}{h} x\cdot \xi}\rangle = \int_{T^*\mathbb{R}^d} \chi^2(x) \psi(x,v) \mathrm{d}\mu_0^\xi(x,v) \]
For the definition and properties of the Weyl quantization $Op_h$, we refer the reader to section \ref{greve}.

We then define a measure $\mu^\xi$ on $T^*\R^d$ by 
\begin{equation*}\int_{T^*\R^d} a \mathrm{d}\mu^\xi := \lim \limits_{t\rightarrow \infty} \int_{T^*\R^d} a\circ\Phi^{t}
\mathrm{d}\mu^\xi_0,
\end{equation*}
for any $a\in C_c^0(T^*\R^d)$.

We will show in section \ref{appendice} that this limit exists under our above assumptions. Actually, the proof will not use the hypothesis \ref{Husserl} that the topological pressure of half the unstable jacobian is negative, but the much weaker assumption that the topological pressure of the unstable jacobian is negative.

The following theorem tells us that, under our hypotheses, $\mu^\xi$ is the semiclassical measure associated to $E_h^\xi$, and it gives us a precise description of $\mu^\xi$ close to the trapped set.

 \begin{theoreme} \label{blacksabbath20}
 Suppose that the Hypothesis \ref{sieste} on hyperbolicity  holds, that the topological pressure Hypothesis \ref{Husserl} is satisfied, and that $\xi\in
\mathbb{S}^{d-1}$ is such that $\Lambda_\xi$ satisfies Hypothesis \ref{Happy} of transversality.

Then for any $\psi\in C_c^\infty(T^*\mathbb{R}^d)$ and any $\chi\in C_c^\infty(\mathbb{R}^d)$, we have

\begin{equation*}\langle Op_h(\psi) \chi  E_h^{\xi}, \chi E_h^{\xi}\rangle = \int_{T^*\R^d} \psi(x,v)
\mathrm{d}\mu^\xi(x,v) +
O(h^c),
\end{equation*}

Furthermore, for any $\rho\in K$, there exists a small neighbourhood $U_\rho\subset T^*\mathbb{R}^d$ of $\rho$, 
  and a local change of symplectic coordinates $\kappa_\rho : U_\rho \rightarrow T^*\mathbb{R}^d$ with $\kappa_\rho(\rho)=0$ such that the following holds. There exists a constant $c>0$ and two sequences of functions $f_n, \phi_n \in C_c^{\infty}(\mathbb{R}^d)$ for $n\in \mathbb{N}$, such that for any $(y,\eta)\in \kappa_\rho (U_\rho)$, we have
\begin{equation*}\mathrm{d}\mu^\xi(\kappa_\rho^{-1}(y,\eta)) = \sum_{n=0}^{\infty}  f_{n}(y) \delta_{\{\eta=\partial
\phi_{n}(y)\}} d y,
\end{equation*}
and where the functions $f_{n}$ satisfy 
\begin{equation}\label{julienisback}
\sum_{n=0}^\infty \|f_n\|_{C^0} < \infty.
\end{equation}
\end{theoreme}
\begin{remarque}
Theorem \ref{blacksabbath20} tells us that the distorted plane waves $E_h^\xi$ have a unique semiclassical measure. This result is therefore analogous to the Quantum Unique Ergodicity conjecture for eigenfunctions of the Laplace-Beltrami operator on manifolds of negative curvature. However, on compact manifolds of negative curvature, the semiclassical measure we expect is the Liouville measure. Here, the semiclassical measure given by Theorem \ref{blacksabbath20} is very different from the Liouville measure, since, close to the trapped set, it is concentrated on a countable union of Lagrangian sub-manifolds of $T^*X$. There is therefore a deep difference between compact and non-compat manifolds concerning the semiclassical measure of eigenfunctions, a fact which was already noted in \cite{GN}.
\end{remarque}

\subsubsection*{Idea of proof}\label{boulgakov}

Theorems \ref{blacksabbath20} and \ref{davies} will be deduced from a precise description of the distorted plane waves $E_h^{\xi}$ microlocally near the trapped set. In Theorem \ref{ibrahim}, we will show that, microlocally near the trapped set, $E_h^{\xi}$ can be written as a convergent sum of WKB states. Let us now explain how this result is obtained.

By definition, the distorted plane waves $E_h^{\xi}$ are generalized eigenfunctions of the operator $P_h$. Therefore, if we write $U(t)=e^{-\frac{i}{h} P_h}$ for the Schrödinger propagator associated to $P_h$, we would like to write formally that $U(t) E_h^{\xi} = e^{-i t/h} E_h^{\xi}$. Of course, this expression can only be formal, since $E_h^\xi\notin L^2$, but we will give it a precise meaning by truncating it by some cut-off functions.

By equation (\ref{chti}), $E_h^{\xi}$ may decomposed into two terms, which we will write $E_h^0$ and $E_h^1$ in the sequel. $E_h^0$ is a Lagrangian state associated to the Lagrangian manifold $\Lambda_\xi$, while $E_h^1$ is the image of a smooth compactly supported function by the resolvent $(P_h-(1+i0)^2)^{-1}$.

Using some resolvent estimates and hyperbolic dispersion estimates, we will show in the sequel that, for any compactly supported function $\chi$, we have $\lim\limits_{t\rightarrow\infty} \|\chi U(t) E_h^1\|=0$.

Therefore, in order to describe $E_h^{\xi}$, we only have to study $U(t) E_h^0$ for some very long times. Since $E_h^0$ is a Lagrangian state, its evolution can be described using the WKB method. To do this, we will have to understand the classical evolution of the Lagrangian manifold $\Lambda_\xi$ for large times. We will
show that for any
$t>0$, the restriction of $\Phi^t(\Lambda_\xi)$ to a region close to the
trapped set consists of finitely many Lagrangian manifolds,
most of which are very close to the “outgoing tail” of the trapped set (see Theorem
\ref{Cyril} for more details).

\subsubsection*{Relation to other works}
The study of the high frequency behaviour of eigenfunctions of Schrödinger operators, and of their semiclassical measures, in the case where the associated classical dynamics has a chaotic behaviour, has a long story. It goes back to the classical works \cite{Shn},\cite{Zel} and \cite{CdV} dealing with Quantum Ergodicity on compact manifolds.

Analogous results on manifolds of infinite volume are much more recent.
In \cite{DG}, the authors studied the semiclassical measures associated to distorted plane waves in a very general framework, with very mild assumptions on the classical dynamics. The counterpart of this generality is that the authors have to average on directions $\xi$ and on an energy interval of size $h$ to be able to define the semiclassical measure of distorted plane waves. Their result can be seen as a form of Quantum Ergodicity result on non-compact manifolds, although no “ergodicity” assumption is made.

In \cite{GN}, the authors considered the case where $X=\Gamma\backslash \mathbb{H}^{d}$ is a manifold of infinite volume, with sectional curvature constant equal to $-1$ (convex co-compact hyperbolic manifold), and with the assumption that the Hausdorff dimension of the limit set of $\Gamma$ is smaller that $(d-1)/2$. In this setting, distorted plane waves are often called \emph{Eisenstein series}. The authors prove that there is a unique semiclassical measure for the Eisenstein series with a given incoming direction, and they give a very explicit formula for it. This result can hence be seen as a Quantum Unique Ergodicity result in infinite volume.

Our result is a generalization of those of \cite{GN}. Indeed, we also obtain a unique semiclassical measure for the distorted plane waves with a given incoming direction. Our assumption on the topological pressure is a natural generalization of the assumption on the Hausdorff dimension of the limit set of $\Gamma$ to the case of nonconstant curvature. As in \cite{GN}, the main ingredient of the proof is a decomposition of the distorted plane waves as a sum of WKB states. Although our description of the distorted plane waves and of their semiclassical measure is slightly less explicit than that of \cite{GN}, our methods are much more versatile, since they rely on the properties of the Hamiltonian flow close to the trapped set, instead of relying on the global quotient structure.

In \cite{Dcusp}, the author was able to obtain semiclassical convergence of distorted plane waves on manifolds of finite volume (with cusps), by working at complex energies; see also \cite{Yannick} for more precise results. The main argument of \cite{Dcusp}, \cite{Yannick} and of \cite{DG}, which is to describe the distorted plane waves as plane waves propagated during a long time by the Schrödinger flow, is the starting point of our proof. However, the reason of the convergence in the long-time limit is very different in \cite{Dcusp} and \cite{Yannick}, \cite{DG} and in the present paper.

Many of the tools used in this paper were inspired by \cite{NZ}. We will use the notations and methods of this paper a lot.

Let us notice that most of the results of the present paper can be made more precise if we suppose that we work on a manifold of non-positive sectional curvature, without a potential. This has been studied in \cite{Ing2}, where the author is able to show, by using the methods developed in the present paper, that distorted plane waves are bounded in $L^\infty_{loc}$ independently of $h$, and to give sharp bounds on the Hausdorff measure of nodal sets of the real part of distorted plane waves restricted to a compact set.

\subsubsection*{Organisation of the paper}

In section \ref{houx}, we will state and prove a result concerning the propagation by the Hamiltonian flow of Lagrangian manifolds similar to $\Lambda_\xi$ near the trapped set, under general assumptions. In part \ref{GE}, we will state Theorem \ref{ibrahim}, which is our main theorem, giving a description of distorted plane waves as a sum of WKB states. We will deduce Theorem \ref{davies} as an easy corollary. In section \ref{tools}, we will recall various tools which were introduced in \cite{NZ}, and which will play a role in the proof of Theorem \ref{ibrahim}. We shall then prove Theorem \ref{ibrahim} in section \ref{preuve}.
Section \ref{SC} will be devoted to the proof of the Theorem \ref{blacksabbath20}.

The main reason why we want to state Theorem \ref{ibrahim} for generalized eigenfunctions that are more general than distorted plane waves on $\mathbb{R}^d$ is that our results do also apply if the manifold is hyperbolic near infinity (which allows us to recover some of the results of \cite{GN}), as is shown in \cite[Appendix B]{Ing2}. Our results do probably also apply if the manifold is asymptotically hyperbolic; this shall be pursued elsewhere. 

\paragraph{Acknowledgements}
The author is partially supported by the Agence Nationale de la Recherche project GeRaSic (ANR-13-BS01-0007-01).

The author would like to thank Stéphane Nonnenmacher for suggesting this project, as well as for his advice during the redaction of this paper.
He also thanks the anonymous referee for suggesting several clarifications and improvements in the paper.

\section{Propagation of Lagrangian manifolds}\label{houx}
\subsection{General assumptions for propagation of Lagrangian manifolds}
\label{troll}

Let $(X,g)$ be a noncompact complete Riemannian manifold of dimension $d$,
and let $V : X \longrightarrow \mathbb{R}$ be a smooth compactly supported
potential.

We denote by $p(x,\xi)= p(\rho) : T^*X\longrightarrow \mathbb{R}$,
$p(x,\xi) = \|\xi\|^2+V(x)$ the
classical Hamiltonian.

For each $t\in \mathbb{R}$, we denote by $\Phi^t:T^*X\longrightarrow T^*X$
the Hamiltonian flow at time $t$ for the Hamiltonian $p$.

Given any smooth function $f : X \longrightarrow \mathbb{R}$, it may be
lifted to a function $f : T^*X \longrightarrow \mathbb{R}$, which we
denote by the same letter. We may then define $\dot{f}, \ddot{f}\in
C^\infty (T^*X)$ to be the derivatives of $f$ with respect to the
Hamiltonian flow.
\begin{equation*}\dot{f}(x,\xi):= \frac{ d f(\Phi^t(x,\xi))}{d
t}\big{|}_{t=0},~~ \ddot{f}(x,\xi):= \frac{ d^2
f(\Phi^t(x,\xi))}{d t^2}\big{|}_{t=0}.
\end{equation*}

\subsubsection{Hypotheses near infinity} \label{Hector}

We suppose the following conditions are fulfilled.

\begin{hyp}[Structure of $X$ near infinity] \label{Guepard}
We suppose that the manifold $(X,g)$ is such that the following holds:

(1) There exists a compactification $\Comp$ of $X$, that is, a compact
manifold with boundaries $\Comp$ such that $X$ is diffeomorphic to the
interior of $\Comp$. The boundary $\Bound$ is called the boundary at
infinity.

(2) There exists a boundary defining function $b$ on $X$, that is, a
smooth function $b : \Comp \longrightarrow [0,\infty)$ such that $b>0$ on
$X$, and $b$ vanishes to first order on $\Bound$.

(3) There exists a constant $\epsilon_0>0$ such that for any point
$(x,\xi)\in \mathcal{E}$,
\begin{equation*}\text{if } b(x,\xi)\leq \epsilon_0 \text{ and } \dot{b}(x,\xi)=0 \text{
then } \ddot{b}(x,\xi)<0.
\end{equation*}
\end{hyp}

Note that, although part $(3)$ of the hypothesis makes reference to the Hamiltonian flow, it is only an assumption on the manifold $(X,g)$ and not on the potential $V$, because $V$ is assumed to be compactly supported. 

\begin{ex} \label{nokia}
$\mathbb{R}^d$ fulfils the Hypothesis \ref{Guepard}, by taking the
boundary defining function $b(x)=(1+|x|^2)^{-1/2}$. We then have $\overline{X}\equiv B(0,1)$.
\end{ex}
\begin{ex} \label{samsung}
The Poincaré space $\mathbb{H}^{d}$
also fulfils the Hypothesis \ref{Guepard}. Indeed, in the ball model
$B_0(1)=\{x\in \mathbb{R}^d; |x|<1\}$, where $|\cdot|$ denotes the
Euclidean norm, then $\mathbb{H}^{d}$ compactifies to the close unit ball,
and the boundary defining function $b(x)=2\frac{1-|x|}{1+|x|}$ fulfills
conditions (2) and (3).
\end{ex}
We will write $X_0:=\{x\in X; b(x)\geq\epsilon_0/2\}$
By possibly taking $\epsilon_0$ smaller, we can assume that $supp (V)
\subset \{x\in X; b(x)\bel \epsilon_0\}$.
We will call $X_0$ the \emph{interaction region}. We will also write
\begin{equation}\label{frite}
W_0:=T^*(X\backslash X_0) = \{\rho\in T^*X; b(\rho) < \epsilon_0/2\},~~~~ \ext = W_0\cap \mathcal{E}.
\end{equation}

By possibly taking $\epsilon_0$ even smaller, we may ask that 
\begin{equation}\label{descendance}
\forall \rho \in \ext, b(\Phi^1(\rho)) < \epsilon_0.
\end{equation}

\begin{definition}
If $\rho=(x,\xi)\in \mathcal{E}$, we say that $\rho$ escapes directly in the forward
direction, denoted $\rho\in\mathcal{D}\mathcal{E}_+$, if $b(x) <
\epsilon_0/2$ and
$\dot{b}(x,\xi)\leq 0$.

If $\rho=(x,\xi)\in \mathcal{E}$, we say that $\rho$ escapes directly in the backward
direction, denoted $\rho\in \mathcal{D}\mathcal{E}_-$, if $b(x)<
\epsilon_0/2$ and
$\dot{b}(x,\xi)\geq 0$.
\end{definition}
Note that we have
\begin{equation*}
\mathcal{W}_0=\mathcal{DE}_-\cup \mathcal{DE}_+.
\end{equation*}

Part (3) of Hypothesis \ref{Guepard} implies the following \emph{geodesic convexity} result, which reflects the fact that once a trajectory has left the interaction region, it cannot come back to it.

\begin{lemme}\label{saladin}
For any $t\geq 0$, we have
\begin{equation*}
\Phi^t(\mathcal{E}\cap T^*X_0) \cap \mathcal{DE}_-=\emptyset.
\end{equation*}
\end{lemme}
\begin{proof}
Suppose that there exists a $\rho\in \Phi^t(\mathcal{E}\cap T^*X_0) \cap
\inco$ for some $t\geq 0$. Then there exists $\rho'\in \mathcal{E}\cap T^* X_0$ such that
$\rho = \Phi^t(\rho')$. Let us consider $f(s):= b(\Phi^s(\rho'))$. We
have $f(0)>\epsilon_0/2$, $f(t)<\epsilon_0/2$ and $f'(t)\geq 0$ by hypothesis.
This is impossible, because by Hypothesis \ref{Guepard}, point (3),
whenever $f(s)\leq \epsilon_0$ and $f'(s)=0$, we  have $f''(s)<0$.
\end{proof}

\subsubsection{Hyperbolicity}\label{averse}
Recall that the \emph{trapped set} was defined in (\ref{ensemblecapte}). In the sequel, we will always suppose that the trapped set is a \emph{hyperbloc set}, as follows.
\begin{hyp}[Hyperbolicity of the trapped set] \label{sieste}
We assume that $K$ is a hyperbolic set for the flow
$\Phi^t_{|\mathcal{E}}$. That is to say,
there exists a metric $g_{ad}$ on a neighbourhood of $K$ included in
$\mathcal{E}$, and $\lambda>0$, such that the following holds. For each
$\rho\in K$, there is a decomposition \begin{equation*}T_\rho
\mathcal{E}=\mathbb{R}H_p(\rho) \oplus E_\rho^+\oplus
E_\rho^-
\end{equation*} such that
\begin{equation*}\|d\Phi_\rho^t(v)\|_{g_{ad}}\leq  e^{-\lambda|t|}\|v\|_{g_{ad}}
\text{   for all } v\in E_\rho^\mp, \pm t\geq 0.
\end{equation*}
\end{hyp}

We will call $E^\pm$ the \emph{unstable} (resp. \emph{stable}) subspaces at the point $\rho$.

We may extend $g_{ad}$ to a metric
on the whole energy layer, so that outside of the interaction region, it coincides with the metric on $T^*X$ induced from the Riemannian metric on $X$. From now on, $d$ will denote the Riemannian
distance associated to this metric on $\mathcal{E}$.

\begin{figure}
    \center
   \includegraphics[scale=0.4]{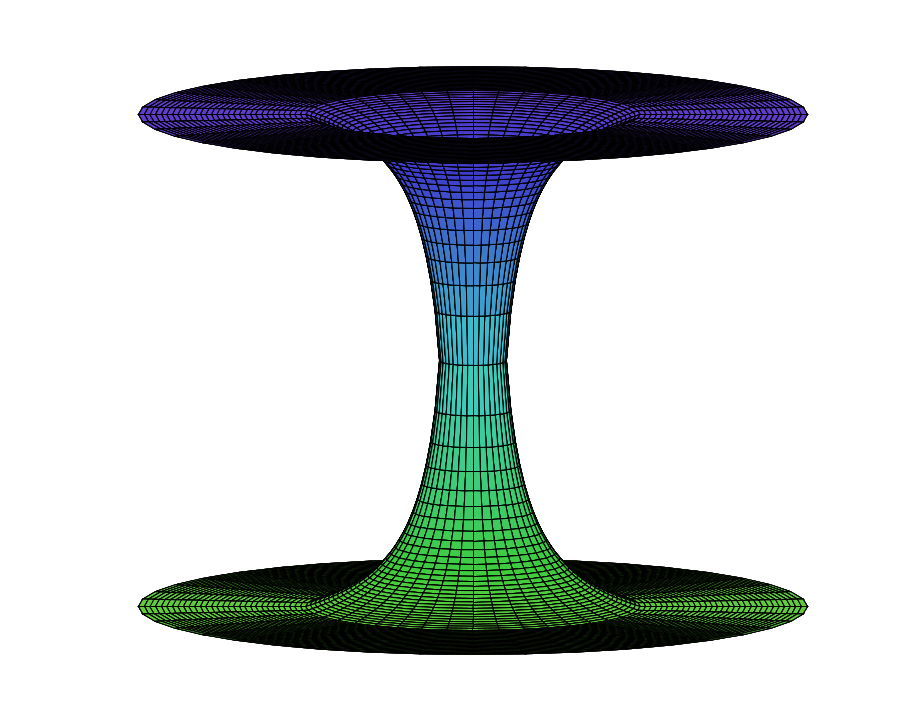}
    \caption{A surface which has negative curvature close to the trapped
set of the geodesic flow, and which is isometric to two copies of
$\mathbb{R}^2\backslash B(0,R_0)$ outside of a compact set. It satisfies Hypothesis \ref{sieste} near the trapped set and Hypothesis \ref{Guepard} at infinity.} \label{example}
\end{figure}

Let us recall a few properties of hyperbolic
dynamics (see \cite[Chapter 6]{KH} for the proofs of the statements).

i) The hyperbolic set is \textit{structurally stable}, in the following sense. For $\mathbf{E}>0$, define the layer of energy $\mathbf{E}$
\begin{equation}\label{deflayer}
\mathcal{E}_\mathbf{E}:= \{\rho\in T^*X; p(\rho)=\mathbf{E}\},
\end{equation}
and the trapped set at energy $\mathbf{E}$ as
\begin{equation}\label{deftrap}
K_\mathbf{E}:=\{\rho\in \mathcal{E}_\mathbf{E} \text { and } \Phi^t(\rho) \text{ remains in a compact set for all } t\in \mathbb{R}\}.
\end{equation}
If $K$ is a hyperbolic set for $\Phi^t_{|\mathcal{\mathbf{E}}}$, then
\begin{equation}\label{struc}
\exists \delta>0, \forall \mathbf{E}\in (1-\delta,1+\delta), ~~ K_{\mathbf{E}}
\text{ is a hyperbolic set for } \Phi^t_{|\mathcal{E}_\mathbf{E}}.
\end{equation}

ii) $d\Phi_\rho^t(E_\rho^\pm)=E_{\Phi^t(\rho)}^\pm$

iii) $K\ni \rho \mapsto E_\rho^\pm \subset T_\rho(\mathcal{E}) $
is Hölder-continuous

iv) Any $\rho\in K$  admits local strongly (un)stable manifolds $
W_{loc}^\pm (\rho)$  tangent to $E^\pm_\rho$,  defined by
\begin{equation*}W_{loc}^\pm (\rho)=\{\rho'\in \mathcal{E};
d(\Phi^t(\rho),\Phi^t(\rho'))<\epsilon \text{ for all } \pm t\leq 0 \text{
and} \lim\limits_{t\rightarrow \mp
\infty} d(\Phi^t(\rho'),\Phi^t(\rho))=0\}, 
\end{equation*}
where $\epsilon>0$ is some small number.

We call
\begin{equation*}E_\rho^{+0}:= E_\rho^+ \oplus \mathbb{R} H_p(\rho),~~~~ E_\rho^{-0}:=
E_\rho^- \oplus \mathbb{R} H_p(\rho),
\end{equation*}
the \textit{weak unstable} and \textit{weak stable} subspaces at the point
$\rho$ respectively.

\subsubsection{Adapted coordinates}\label{banff}

Let us now describe the construction of a local system of coordinates which is adapted to the stable and unstable directions near a point. In the sequel, these coordinates will be considered as fixed, and used to state Theorem \ref{Cyril}.
\begin{lemme} \label{adap}
Let $\rho\in K$. There exists an adapted
system of symplectic coordinates $(y^\rho,\eta^\rho)$ on a neighbourhood
of $\rho$ in
$T^*X$ such that the following holds:
\[
\begin{aligned}
&(i)~~\rho\equiv (0,0)\\
&(ii)~~E_{\rho}^+ = span \{\frac{\partial}{\partial y^\rho_i}(\rho), ~~~
i=2,...,d\} ,\\
&(iii)~~E_{\rho}^- = span \{\frac{\partial}{\partial \eta_i^\rho}(\rho), ~~~
i=2,...,d\} ,\\
&(iv)~~ \eta_1^\rho= p-1  \text{ is the energy coordinate},\\
&(v)~~\big{\langle}\frac{\partial}{\partial
y^\rho_i}(\rho),\frac{\partial}{\partial
y^\rho_j}(\rho)\big{\rangle}_{g_{ad}(\rho)}=\delta_{i,j},~~~i,j=2,...,d.
\end{aligned}
\]
\end{lemme}
\begin{proof}

We may identify a neighbourhood of $\rho\in T^*X$ with a neighbourhood of
$(0,0)\in T^*\mathbb{R}^d$. Let us take $e_1^\rho=H_p(\rho)$, and complete it
into a basis $(e_1^\rho,...,e_d^\rho)$ of $E_\rho^{+0}$ such that
$\langle e_i^\rho,e_j^\rho\rangle_{g_{ad}(\rho)}=1$ for $2\leq i,j \leq d$.

Since $E^{\pm 0}$ are Lagragian subspaces (which follows from the hyperbolicity assumption), it is then possible to find vectors $(f_1^\rho,...,f_d^\rho)$ such that
$E_\rho^-=
span \{f_2^\rho,...,f_d^\rho\}$ and such that
$\omega(f_j^\rho,e_k^\rho)=\delta_{j,k}$ for any
$1\leq j,k\leq d$.
In particular, we have $\omega(f_1^\rho,e_1^\rho)=dp(f_1)=1$.

 From Darboux's theorem, there exists a non-linear symplectic chart
$(y^\flat,\eta^\flat)$ near the origin such that $\eta_1^\flat=p-1$.
There also exists a linear symplectic transformation $A$ such that the
coordinates $(y,\eta)=A(y^\flat,\eta^\flat)$ satisfy
$\eta_1=\eta_1^\flat$. as well as
 \begin{equation*}\eta_1=p-1, ~~ \frac{\partial}{\partial y_j}(0,0)=e_j \text{  and  }
\frac{\partial}{\partial \eta_j}(0,0)=f_j, ~~ j=1,...,d.
\end{equation*}
\end{proof}

We will often write
\begin{equation}\label{pavillon}
\begin{aligned}
\boldsymbol{y}^\rho := (y_2^\rho,...,y_d^\rho) ~~\text{ and }
\boldsymbol{\eta}^\rho:= (\eta^\rho_2,...,\eta_d^\rho)  .
\end{aligned}
\end{equation}

For any $\epsilon >0$, write $D_\epsilon= \{u\in \mathbb{R}^{d-1},
|u|<\epsilon\}$. We define the following polydisk centred at $\rho$:
\begin{equation}\label{Enee}
U^\rho(\epsilon) \equiv \{(y^\rho,\eta^\rho): |y_1^\rho|<\epsilon,
|\eta_1^\rho|<\delta, \boldsymbol{y}^\rho\in D_\epsilon, \boldsymbol{\eta}\in D_{\epsilon}\},
\end{equation}
where $\delta$ comes from (\ref{struc}).

We also define \emph{unstable Lagrangian manifolds}, which are needed in the statement of Theorem \ref{Cyril}.

\begin{definition}
Let $\Lambda\subset \mathcal{E}$ be an isoenergetic Lagrangian manifold
(not
necessarily connected) included in a small neighbourhood $W$ of a point
$\rho\in K$,
and let $\gamma>0$. We will say that $\Lambda$ is a
$\gamma$\textit{-unstable Lagrangian manifold} (or
that $\Lambda$ is in the $\gamma$-unstable cone) in the coordinates
$(y^\rho,\eta^\rho)$ if it
can be written in the form
\begin{equation*}\Lambda = \{(y^\rho;0,F(y^\rho)) ; y^\rho\in D \},
\end{equation*}
where $D\subset \mathbb{R}^d$, is an open subset with finitely many
connected components, and with piecewise smooth boundary, and $F :
\mathbb{R}^{d} \longrightarrow
\mathbb{R}^{d-1}$ is a smooth function with $\|dF\|_{C^0}\leq \gamma$.
\end{definition}

Note that, since $F$ is defined on $\mathbb{R}^{d}$, a $\gamma$-unstable
manifold may always be seen as a submanifold of a \textit{connected
$\gamma$-unstable Lagrangian manifold}.

Let us also note that, since $\Lambda$ is isoenergetic and is Lagrangian, an immediate computation shows that $F$ does not depend on $y_1^\rho$, so that $\Lambda$ can actually be put in the form
\begin{equation*}\Lambda = \{(y^\rho,0,f(\boldsymbol{y}^\rho)) ; y^\rho\in D \},
\end{equation*}
where $f :
\mathbb{R}^{d-1} \longrightarrow
\mathbb{R}^{d-1}$ is a smooth function with $\|df\|_{C^0}\leq \gamma$.

\subsubsection{Hypotheses on the incoming Lagrangian
manifold}\label{scotiabank}
Let us consider an isoenergetic Lagrangian manifold $\Lag\subset
\mathcal{E}$ of the form 
$$\Lag:=\{(x,\Dir(x)),x\in
X_1\},$$
where $X_1$ is a closed subset of $X\backslash X_0$ with finitely
many connected components and piecewise smooth boundary, and $\Dir :
X_2\ni x\longrightarrow \Dir(x)\in T_x^*X$ is a smooth
co-vector field defined on some neighbourhood $X_2$ of $X_1$.

We make the following additional hypothesis on $\Lag$:

\begin{hyp}[Invariance hypothesis]\label{chaise}
We suppose that $\Lag$ satisfies the following invariance hypotheses.
\begin{equation}\label{invfut}
\forall t \geq 0, \Phi^t(\Lag)\cap \mathcal{DE}_- = \Lag\cap \mathcal{DE}_-.
\end{equation}
\end{hyp}

\begin{ex} \label{sega}
Given a $\xi\in \mathbb{R}^d$ with $|\xi|^2=1$, the Lagrangian manifold $\Lambda_\xi$
defined in the introduction fulfils Hypothesis \ref{chaise}.
\end{ex}

\begin{ex} \label{playstation}
Suppose that $(X\backslash X_0,g) \cong (\mathbb{R}^d\backslash B(0,R),g_{Eucl}) $ for
some $R>0$. Then the incoming spherical Lagrangian, defined by
\begin{equation*}\Lambda_{sph}:= \{ (x,-\frac{x}{|x|}); |x|>R\},
\end{equation*}
fulfils Hypothesis \ref{chaise}.
\end{ex}

We also make the following transversality assumption on the Lagrangian
manifold $\Lag$. It roughly says that $\Lag$ intersects the stable manifold transversally.

\begin{hyp}[Transversality hypothesis]\label{Happy}
We suppose that $\Lag$ is such that, for any $\rho\in K$, for
any $\rho'\in \Lag$, for any $t\geq 0$, we have
\begin{equation*}
\Phi^t(\rho')\in W_{loc}^-(\rho)\Longrightarrow W_{loc}^-(\rho) \text{ and } \Phi^t(\Lag) \text{ intersect transversally at } \Phi^t(\rho'),
\end{equation*} 
that is to say
\begin{equation}\label{telus}
T_{\Phi^t(\rho')}\Lag \oplus T_{\Phi^t(\rho')} W_{loc}^-(\rho) = T_{\Phi^t(\rho')} \mathcal{E}.
\end{equation}
\end{hyp}
Note that (\ref{telus}) is equivalent to $T_{\Phi^t(\rho')}\Lag \cap T_{\Phi^t(\rho')} W_{loc}^-(\rho) = \{0\}$.

On $X=\mathbb{R}^d$, Hypothesis \ref{Happy} is likely to hold for almost every $\xi\in \mathbb{S}^{d-1}$, at least for a generic $V$. In \cite{Ing2}, the author shows that this hypothesis is satisfied for every $\xi$ on manifolds of non-positive curvature which have several Euclidean ends (like the one in Figure \ref{example}), when there is no potential.
\subsection{Statement of the result}
Let us now state the main result of this section, which describes the
“truncated evolution” of Lagrangian manifolds.
\paragraph{Truncated Lagrangians}

Let $(W_a)_{a\in A}$ be a finite family of open sets in $T^*X$.
Let
$N\in\mathbb{N}$, and
let $\alpha=\alpha_0,\alpha_1...\alpha_{N-1}\in A^{N}$.
Let $\Lambda$ be a Lagrangian manifold
in $T^*X$. We define the sequence of (possibly empty) Lagrangian manifolds
$(\Phi_\alpha^{k}(\Lambda))_{0\leq k \leq N}$ by recurrence by:
\begin{equation*}\Phi_\alpha^{0}(\Lambda)= \Lambda \cap W_{\alpha_0},~~~~
\Phi_\alpha^{k+1}(\Lambda) =
W_{\alpha_{k+1}} \cap  \Phi^1 (\Phi_{\alpha}^{k}
(\Lambda)).
\end{equation*}

In the sequel, we will consider families with indices in $A=A_1\sqcup A_2 \sqcup
\{0\}$.
For any $\alpha\in A^{N}$ such that $\alpha_{N-1}\neq 0$, we will define
\begin{equation}\label{rollingstone}
\tau(\alpha):= \max \{1\leq i \leq N-1; \alpha_i=0\}
\end{equation}
 if there exists 
$1\leq i \leq N-1$ with $\alpha_i=0$, and $\tau(\alpha)=0$ otherwise.

\begin{theoreme} \label{Cyril}
Suppose that the manifold $X$ satisfies Hypothesis
\ref{Guepard} at infinity, that the Hamiltonian flow $(\Phi^t)$ satisfies
Hypothesis \ref{sieste}, and that the Lagrangian manifold $\Lag$ satisfies
the invariance Hypothesis \ref{chaise} as well as the transversality
Hypothesis \ref{Happy}.

Fix $\ins>0$ small enough. There exists $\secur>0$ such that the following
holds.
Let $(W_a)_{a\in A_1}$ be any open cover of $K$ in $T^*X$ of diameter
$<\secur$, such that there exist points $\rho_a\in W_a\cap K$, and
such that the adapted coordinates $(y^a,\eta^a)$ centred on $\rho^a$ are well defined on $W_a$
for every $a\in A_1$. Then we may complete this cover into $(W_a)_{a\in A}$ an
open cover of
$\mathcal{E}$ in $T^*X$ where $A=A_1\sqcup A_2\sqcup\{0\}$ (with $W_0$ defined as in (\ref{frite})) such that the
following holds.

There exists $\tins\in \mathbb{N}$ such that for all
$N\in
\mathbb{N}$, for all $\alpha\in A^{N}$ and all $a\in A_1$, then
$W_a\cap\Phi_\alpha^{N}(\Lag)$ is either empty, or is a Lagrangian manifold
in some unstable
cone in the coordinates $(y^{a},\eta^{a})$.

Furthermore, if $N-\tau(\alpha)\geq \tins$, then
$W_a\cap\Phi_\alpha^{N}(\Lag)$ is a $\ins$-unstable Lagrangian manifold
in the coordinates $(y^{a},\eta^{a})$.
\end{theoreme}

\begin{remarque}
For a sequence $\alpha\in A^N$, $N-\tau(\alpha)$ corresponds to the time spent in
the interaction region. Our last statement therefore says that if a part
of $\Lag$ stays in the interaction region for long enough when propagated, then its tangents will form a small angle with the unstable direction at $\rho^a$.
\end{remarque}

\begin{remarque}\label{bayonne}
The constant $\secur$ and the sets $(W_a)_{a\in A_2}$ depend on the Lagrangian manifold $\Lag$. If we take a whole family of Lagrangian manifolds $(\mathcal{L}_z)_{z\in Z}$ satisfying Hypothesis \ref{chaise} and Hypothesis \ref{Happy}, then we will need some additional conditions on the whole family to be able to find a common choice of $\secur$ and $(W_a)_{a\in A_2}$ independent of $z\in Z$. An example of such a condition will be provided by equations (\ref{UT}) and (\ref{cognet}). Note that these equations are automatically satisfied if $Z$ is finite.
\end{remarque}

\subsection{Proof of Theorem \ref{Cyril}}
\begin{proof}
From now on, we will fix a $\ins>0$.

Let $\rho_0\in K$, and consider
the system of adapted coordinates in a neighbourhood of
$\rho_0$ constructed in section \ref{banff}. Recall that the set
$U^{\rho_0}(\epsilon)$ was defined in (\ref{Enee}).
We define a \textit{Poincaré section} by
\begin{equation*}
\Sigma^{\rho_0}=\Sigma^{\rho_0}(\epsilon)=\{(y^{\rho_0},\eta^{\rho_0})\in
U^{\rho_0}(\epsilon) ; y_1^{\rho_0}=\eta_1^{\rho_0}=0 \}.
\end{equation*}

Note that the spaces $E_{\rho_0}^\pm$ are tangent to $\Sigma^{\rho_0}$, and that the
coordinates $(\boldsymbol{y}^{\rho_0},\boldsymbol{\eta}^{\rho_0})$ introduced in (\ref{pavillon}) form a symplectic chart on $\Sigma^{\rho_0}$.

Actually, we will often need a non-symplectic system of coordinates built from the coordinates $(y^\rho,\eta^\rho)$.

Before building this non-symplectic system of coordinates, let us explain why it is a crucial ingredient of our argument. The main tool in the proof of Theorem \ref{Cyril} is the so-called “Inclination lemma”,
which roughly says that a Lagrangian manifold which intersects the stable manifold transversally will get more and more unstable when propagated in the future. This is a very easy result in the case of linear hyperbolic diffeomorphisms, but we must add some quantifiers in the case of non-linear dynamics to make it rigorous. Namely, one can say, as in \cite[Proposition 5.1]{NZ}, that given a $\gamma\bel 0$, there exists $\epsilon_\gamma\bel 0$ such that if $\Lambda$ is a $\gamma$-unstable Lagrangian manifold included in some $U^{\rho}(\epsilon_\gamma)$, then for any $\rho'$, $\Phi^1 (\Lambda)\cap U^{\rho'}(\epsilon_\gamma)$ is still $\gamma$-unstable.

However, we may not use this result directly for the following reason. The smaller we take $\epsilon$, the longer the points of the Lagrangian manifold $\Lag$ may spend in the part of the interaction region which is not affected by the hyperbolic dynamics before entering in some $U^\rho(\epsilon)$ for some $\rho\in K$. Yet the longer they spend in this “intermediate” region, the more stable the Lagrangian manifold may a priori become. To avoid such a circular reasoning, we should introduce another system of coordinates, in which the description of the propagation of the Lagrangian manifolds in the intermediate region is easier.
\subsubsection{Alternative coordinates} \label{Alt}
In this paragraph, we will describe a system of “alternative”, or
“twisted” coordinates built from
the one we introduced in section \ref{banff}, but which may differ
slightly from them.

Given a $\rho\in K$, we introduce a system of smooth coordinates
$(\tilde{y}^\rho,\tilde{\eta}^\rho)$ as follows.

On $\Sigma^\rho$, these coordinates are such that
\begin{equation*}
\begin{aligned} W_{loc}^{0+}(\rho)\cap
\Sigma^\rho &\equiv \{(\tilde{\boldsymbol{y}}^\rho,0) ; \tilde{\boldsymbol{y}}^\rho\in D_\epsilon\},\\
W_{loc}^{0-}(\rho)\cap
\Sigma^\rho
&\equiv \{(0,\tilde{\boldsymbol{\eta}}^\rho) ; \tilde{\boldsymbol{\eta}}^\rho\in D_\epsilon\},
\end{aligned}
\end{equation*}
and if we denote by $L_\rho$ the map 
\begin{equation}\label{claudie}
L_\rho: (\boldsymbol{y}^\rho,\boldsymbol{\eta}^\rho) \mapsto (\tilde{\boldsymbol{y}}^\rho ,\tilde{\boldsymbol{\eta}}^\rho)
\end{equation}
defined in a neighbourhood of $(0,0)$, we have
\begin{equation}\label{ourson}
dL_\rho (0,0)= Id_{\mathbb{R}^{2d-2}}.
\end{equation}

Now, if $\hat{\rho}$ has straight coordinates $(y^\rho(\hat{\rho}), \eta^\rho(\hat{\rho}))$, we let $\hat{\rho}'\in \Sigma^\rho$ be the point with straight coordinates $(0,\boldsymbol{y}^\rho(\hat{\rho}),0,\boldsymbol{\eta}^\rho(\hat{\rho}))$. We do then define the twisted coordinates of $\hat{\rho}$ by
\begin{equation*}
\begin{aligned}
\tilde{y}^\rho_1(\hat{\rho}) &= y^\rho_1(\hat{\rho}),\\
\tilde{\eta}^\rho_1(\hat{\rho}) &= \eta^\rho_1(\hat{\rho}),\\
\tilde{\boldsymbol{y}}^\rho(\hat{\rho})&= \tilde{\boldsymbol{y}}^\rho(\hat{\rho}'),\\
\tilde{\boldsymbol{\eta}}^\rho(\hat{\rho})&=\tilde{\boldsymbol{\eta}}^\rho(\hat{\rho}').
\end{aligned}
\end{equation*}

Note that this system of coordinates doesn't have to be symplectic.

We have
\begin{equation}\label{mataman}
\begin{aligned}
\frac{\partial \boldsymbol{y}_j^\rho}{\partial \tilde{y}^\rho_1}&= \frac{\partial \boldsymbol{\eta}_j^\rho}{\partial \tilde{y}^\rho_1}=0 ~~\text{ for } j=1,...,d-1,\\
\frac{\partial y^\rho_1}{\partial \tilde{y}^\rho_1}&= 1
\end{aligned}
\end{equation}

Given a $\rho\in K$, and $\epsilon,\epsilon'>0$, we define
\begin{equation}
\tilde{U}^\rho(\epsilon,\epsilon') \equiv
\{(\tilde{y}^\rho,\tilde{\eta}^\rho): |\tilde{y}_1^\rho|<\epsilon,
|\tilde{\eta}_1^\rho|<\delta, \tilde{\boldsymbol{y}}^\rho\in D_{\epsilon'},
\tilde{\boldsymbol{\eta}}^\rho\in D_{\epsilon}\},
\end{equation}
where $\delta$ is an energy interval on which the dynamics remains
uniformly hyperbolic.

Finally, the Poincaré section in the alternative coordinates is represented as

\begin{equation*}\tilde{\Sigma}^\rho(\epsilon,\epsilon') :=
\{(\tilde{y}^{\rho},\tilde{\eta}^{\rho})\in
\tilde{U}^{\rho}(\epsilon,\epsilon') ;
\tilde{y}_1^{\rho}=\tilde{\eta}_1^{\rho}=0 \}.
\end{equation*}

In the sequel, we will be working most of the time in a situation where $\epsilon'<< \epsilon$ (that is, with sets much thinner in the unstable direction than in the stable direction).

The main reason why we needed to introduce alternative coordinates is that they give a simpler expression for the \emph{Poincaré map} (see Remark \ref{whyalternative}). Let us now define this map.

\subsubsection{The Poincaré map}
Let $\rho_0\in K$, and let $\epsilon \bel 0$ be small enough so that the twisted coordinates around $\rho_0$ and $\Phi^1(\rho_0)$ are well defined in some neighbourhoods $\tilde{U}^{\rho_0}(\epsilon,\epsilon)$ and $\tilde{U}^{\Phi^1(\rho_0)}(\epsilon,\epsilon)$. The \textit{Poincaré map} $\kappa_{\rho_0}$ is
defined, for
$\rho\in
\tilde{\Sigma}^{\rho_0}(\epsilon)$ near $\rho_0$, by taking the intersection of the
trajectory $(\Phi^s(\rho))_{|s-1|\leq \epsilon}$ with the section
$\tilde{\Sigma}^{\Phi^1(\rho_0)}$ (this intersection consists of at most one point). In the sequel, we will sometimes omit the
reference to $\rho_0$ and simply write the Poincaré map
$\kappa$.

The map $\kappa_{\rho_0}$ need not be symplectic, since it is defined in the twisted coordinates which need not be symplectic. 
However, if we had defined the Poincaré map in the straight coordinates, it would have been automatically symplectic. 
The linearisation of the two systems of coordinates are identical at $\rho_0$ by equation (\ref{ourson}). 
Therefore, by using the hyperbolicity assumption, we see that the differential of $\kappa$ at $\rho_0$ takes the form

\begin{equation*}
d\kappa(\rho_0)\equiv \begin{pmatrix}
A & 0\\
0 & ^t A^{-1}
\end{pmatrix},
\end{equation*}
and there
exists
\begin{equation} \label{Tata}
\nu=e^{-\lambda}<1
\end{equation}
such that the matrix $A$ satisfies
\begin{equation}\label{fourmi}
\|A^{-1}\|\leq \nu,
\end{equation}
where $\|\cdot\|$ corresponds to the matrix norm.
Hence, the Poincaré map $\kappa_{\rho_0}$ takes
the form
\begin{equation}\label{Claire}
\kappa_{\rho_0}(\tilde{\boldsymbol{y}}^{\rho_0},\tilde{\boldsymbol{\eta}}^{\rho_0})= \big{(} A
\tilde{\boldsymbol{y}}^{\rho_0}+\tilde{\alpha}(\tilde{\boldsymbol{y}}^{\rho_0},\tilde{\boldsymbol{\eta}}^{\rho_0}),^tA^{-1}
\tilde{\boldsymbol{\eta}}^{\rho_0}+\tilde{\beta}(\tilde{\boldsymbol{y}}^{\rho_0},\tilde{\boldsymbol{\eta}}^{\rho_0})\big{)},
\end{equation}
and the functions $\tilde{\alpha}$ and $\tilde{\beta}$ satisfy:
\begin{equation}\label{poincarestabel}
\tilde{\alpha}(0,\tilde{\boldsymbol{\eta}}^{\rho_0})=\tilde{\beta}(\tilde{\boldsymbol{y}}^{\rho_0},0)\equiv
0\text{
  and   } d\tilde{\alpha}(0,0)=d\tilde{\beta}(0,0)=0.
\end{equation}
We therefore have
\begin{equation} \label{buis}
\|\tilde{\alpha}\|_{C^1(V)}\leq C_0 \epsilon,~~
\|\tilde{\beta}\|_{C^1(V)}\leq C_0 \epsilon,
\end{equation}
for some constant $C_0$, since $\kappa$ is uniformly $C^2$.

\begin{remarque}\label{whyalternative}
Equation (\ref{poincarestabel}) is the main reason why we needed to introduce alternative coordinates, and will play a key role in the proof of Lemma \ref{Benjamin}. If we had defined the Poincaré map in the straight coordinates, we wouldn't have had $\alpha(0,\boldsymbol{\eta}^{\rho_0})= 0$ or $\beta(\boldsymbol{y}^{\rho_0},0)= 0$
\end{remarque}
\begin{remarque} \label{Collot}
By compactness of the trapped set, the constants $C_0$ and $\nu$ may be chosen
independent of the point $\rho_0$.
We may also find a $\mathcal{C}>1$ such that, independently of $\rho_0$
and $\rho_1$ in $K$, we have
\begin{equation}
\|A\|\leq \mathcal{C}.
\end{equation}
Finally, by possibly taking $C_0$ larger, we may assume that all the second derivatives of the map $L_\rho$ defined in (\ref{claudie}) are bounded by $C_0$ independently on $\rho\in K$.
\end{remarque}

\subsubsection{Changes of coordinates and Lagrangian manifolds}
Let us describe how a Lagrangian manifold is affected when we go from
twisted coordinates to straight coordinates centred at the same point.

\begin{lemme}\label{yaki}
Suppose that a Lagrangian manifold $\Lambda\subset \tilde{U}^\rho(\epsilon,\epsilon)$ may be written in the twisted
coordinates centred on $\rho\in K$ as $\Lambda=
\{(\tilde{y}_1^\rho,\tilde{\boldsymbol{y}}^\rho;0,\tilde{F}(\tilde{y}^\rho));
\tilde{y}^\rho\in \tilde{D}_\rho\}$, where $\tilde{D}_\rho\subset
\mathbb{R}^{d}$ is a small open set, and with $\|d\tilde{F}\|_{C^0}\leq
\gamma$. Suppose furthermore that 
\begin{equation*}
C_0\epsilon \gamma < 1
\end{equation*}

Then, in the straight coordinates, $\Lambda$ may be written as:
\begin{equation*}\Lambda= \{(y_1^\rho,\boldsymbol{y}^\rho;0,f(\boldsymbol{y}^\rho)); {\boldsymbol{y}}^\rho\in {D}_\rho\}, 
\end{equation*}
with $\|df\|_{C^0}\leq \gamma (1-C_0\gamma\epsilon)^{-1}(1+2C_0\epsilon)$.
\end{lemme}

\begin{proof}

To lighten the notations, we will not write the indices $\rho$.

Points on $\Lambda$ are parametrized by the coordinate $\tilde{y}$.
We may hence see their straight coordinates $u$, $s$ as functions of $\tilde{y}$.

By equations (\ref{ourson}), (\ref{mataman}) and Remark
\ref{Collot}, we have
\begin{equation*} 
\begin{aligned}
\frac{\partial y}{\partial \tilde{y}} &= \frac{\partial y}{\partial \tilde{y}}+ \frac{\partial y}{\partial \tilde{\eta}} \frac{\partial \tilde{F}(\tilde{y})}{\partial\tilde{y}}\\
&= I+R
\end{aligned}
\end{equation*}
with $\|R\|\leq C_0 \gamma \epsilon< 1$.

Therefore, on $\Lambda$, $\tilde{y} \mapsto y$ is invertible. We may hence write $\boldsymbol{\eta}$ as
a function of $y$, and
we have
\begin{equation*}\frac{\partial \boldsymbol{\eta}}{\partial y} = \frac{\partial \tilde{y}}{\partial y}\Big{[}
\frac{\partial \boldsymbol{\eta}}{\partial \tilde{y}}+ \frac{\tilde{F}(\tilde{y})}{\partial \tilde{y}}\frac{\partial \boldsymbol{\eta}}{\partial \tilde{\boldsymbol{\eta}}}\Big{]}=
(I+R)^{-1}(\gamma (I+ R')),
\end{equation*}
with $\|R'\|\leq 2C_0 \epsilon$. Hence $\Big{\|}\frac{\partial \boldsymbol{\eta}}{\partial y}\Big{\|} \leq \gamma (1-C_0\gamma\epsilon)^{-1}(1+2C_0\epsilon)$.

That $\boldsymbol{\eta}$ is actually independent of $y_1$ comes from the fact that $\Lambda$ is an isoenergetic Lagrangian manifold, and that we are working in symplectic coordinates.
\end{proof}

Let us now describe the change between two systems of twisted coordinates.
Let $\rho,\rho'\in K$. If they are close enough to each other, the map
${L} : (\tilde{y}^\rho, \tilde{\eta}^\rho)\mapsto (\tilde{y}^{\rho'},
\tilde{\eta}^{\rho'}) $ is well
defined on a set containing both $\rho$ and $\rho'$, of diameter
$d(\rho,\rho')$.

Combining the fact that the (un)stable subspaces $E^{\pm}_\rho$ are Hölder
continuous with respect to $\rho\in K^\delta$ with some Hölder exponent
$\mathsf{p}>0$, and point (v) of Lemma \ref{adap}, we get:
\begin{equation} \label{molson}
d{L}_{(0,0)}=\mathsf{L}+R_{\rho,\rho'},
\end{equation}
where 
\begin{equation} \label{evier}
\|R_{\rho,\rho'}\|\leq C d^\mathsf{p}(\rho,\rho') \text{   for some }\mathsf{p}>0,
\end{equation}
 and where $\mathsf{L}$ is of the form
\begin{equation*}
\mathsf{L}=\begin{pmatrix}
U_y&0\\
0& \mathsf{L}_\eta
\end{pmatrix},
\end{equation*}
for some unitary matrix $U_y$. Here, $\mathsf{L}_\eta$ might not be unitary, but it is invertible, and by compactness of $K$, $\|\mathsf{L}_\eta\|^{-1}$ may be bounded independently on $\rho$.

Now, by compactness, the second derivatives of $L$ may be bounded independently of $\rho$ and $\rho'$. Therefore, for any $\rho''$ in a neighbourhood of $\rho$, we have
\begin{equation}\label{heyjoe}
dL_{\rho''}= dL_{(0,0)} + R_{\rho''},
\end{equation}
with $R_{\rho''}\leq C' d(\rho,\rho'')$ and $C'$ independent of $\rho'$.

 By possibly enlarging $C_0$, we may assume that $\|\mathsf{L}_\eta\|^{-1}\leq C_0$. We may also assume that $C_0/2$ is larger than the constants $C$ and $C'$ appearing in the bounds on $R_{\rho,\rho'}$ and $R_{\rho''}$.

We will use the previous remarks in the form of the following lemma, which
describes the effect of a change of twisted coordinates on a Lagrangian
manifold.
\begin{lemme}\label{Gates}
Let $\rho,\rho'\in K$ be such that $d(\rho,\rho')<\epsilon$, and let
$\Lambda$ be a Lagrangian manifold which
may be written in the twisted coordinates centred on $\rho$ as $\Lambda=
\{(\tilde{y}_1^\rho,\tilde{\boldsymbol{y}}^\rho;0,\tilde{F}^\rho(\tilde{y}^\rho));
\tilde{y}^\rho\in \tilde{D}_\rho\}$, where $\tilde{D}_\rho\subset
\mathbb{R}^{d}$ is a small open set, and with
$\|d\tilde{F}^\rho\|_{C^0}\leq
\gamma< \frac{1}{4 C_0 \epsilon^\mathsf{p}}$.

Then, $\Lambda\cap \tilde{U}^{\rho'}(\epsilon,\epsilon)$ may be
written in the coordinates centred at $\rho'$ as
\begin{equation*}\Lambda\cap \tilde{U}^{\rho'}(\epsilon,\epsilon)=
\{(\tilde{y}_1^{\rho'},\tilde{\boldsymbol{y}}^{\rho'};0,\tilde{F}^{\rho'}(\tilde{y}^{\rho'}));
\tilde{y}^{\rho'}\in \tilde{D}_{\rho'}\},
\end{equation*} where $\tilde{D}_{\rho'}\subset
\mathbb{R}^{d}$ is a small open set, and with
\begin{equation*}\|d\tilde{F}^{\rho'}\|_{C^0}\leq (
\gamma(1+C_0\epsilon^\mathsf{p})+C_0\epsilon^\mathsf{p})
(1-2\gamma C_0\epsilon^\mathsf{p})^{-1}<\infty.
\end{equation*}
\end{lemme}
\begin{proof}
Consider points on $\Lambda$. By assumption, their $\tilde{\eta}^\rho$
coordinate is a function of their $\tilde{y}^\rho$ coordinate. Therefore,
using the map $L$, their coordinates
$(\tilde{y}^{\rho'},\tilde{\eta}^{\rho'})$ may be seen as functions of
$\tilde{y}^\rho$.

Let us denote by $L_y$ and $L_\eta$ the two components of $L$.
 By definition, we have
\begin{equation*}\tilde{y}^{\rho'}=L_y(\tilde{y}^\rho,\tilde{\eta}^\rho)=
L_y(\tilde{y}^\rho,\tilde{F}^\rho(\tilde{y}^\rho)),
\end{equation*}
where $\tilde{F}^\rho(\tilde{y}^\rho)$
satisfies $\big{\|}\frac{\partial \tilde{F}^\rho(\tilde{y}^\rho)}{\partial
\tilde{y}^\rho}\big{\|}\leq \gamma$. Therefore, we have:
\[
\begin{aligned}\frac{\partial \tilde{y}^{\rho'}}{\partial \tilde{y}^{\rho}}
&=\frac{\partial L_y}{\partial \tilde{y}^\rho} +
\frac{\partial \tilde{F}^\rho(\tilde{y}^\rho)}{\partial \tilde{y}^\rho}
\frac{\partial L_y}{\partial \tilde{\eta}^\rho}\\
&= U +\tilde{R},
\end{aligned}
\]
where $U$ is unitary.

By equations (\ref{molson}) and (\ref{heyjoe}), we have $\|\tilde{R}\|\leq 2\gamma C_0
\epsilon^\mathsf{p}<1$ by assumption.
Therefore, $\tilde{y}^\rho\mapsto \tilde{y}^{\rho'}$ is invertible, and we
have $\big{\|}\frac{\partial
\tilde{y}^{\rho}}{\partial\tilde{y}^{\rho'}}\big{\|}\leq (1-2\gamma C_0
\epsilon^\mathsf{p})^{-1}$. We may see $\tilde{\eta}^{\rho'}$ as a function
of $\tilde{y}^{\rho'}$, and we have
\[
\begin{aligned}
\Big{\|}\frac{\partial \tilde{\eta}^{\rho'}}{\partial\tilde{y}^{\rho'}}\Big{\|} &=
\Big{\|}\frac{\partial \tilde{y}^{\rho}}{\partial\tilde{y}^{\rho'}}\frac{\partial
\tilde{\eta}^{\rho'}}{\partial\tilde{y}^{\rho}}+ \frac{\partial
\tilde{y}^{\rho}}{\partial\tilde{y}^{\rho'}}\frac{\partial
\tilde{\eta}^{\rho}}{\partial\tilde{y}^{\rho}} \frac{\partial
\tilde{\eta}^{\rho'}}{\partial\tilde{\eta}^{\rho}}\Big{\|}\\
& \leq (1-2\gamma C_0\epsilon^\mathsf{p})^{-1}(C_0\epsilon^\mathsf{p}+
\gamma(1+C_0\epsilon^\mathsf{p})),
\end{aligned}
\]
and the lemma follows.

\end{proof}

\subsubsection{Propagation for bounded times}
Let us fix a $\nu_1\in (\nu,1)$, where $\nu$ was defined in (\ref{Tata}). Recall that $\mathsf{p}$ was defined in 
(\ref{evier}) as the Hölder exponent of the stable and unstable directions. From now on, we fix an $\epsilon>0$ small
enough so that
\begin{equation} \label{belette}
\frac{\nu+C_0\epsilon^\mathsf{p}}{\nu^{-1}-C_0\epsilon^\mathsf{p}}<\nu_1,
\text{ and   }
\frac{C_0\epsilon^\mathsf{p}}{\nu^{-1}-2C_0 \epsilon^\mathsf{p}} <
\frac{\ins(1-\nu_1)}{8}.
\end{equation}
\begin{equation} \label{diamant}
\big{(}1-\frac{(1+\nu_1)\ins}{1+2C_0\epsilon^\mathsf{p})}C_0\epsilon^{\mathsf{p}}\big{)}^{-1}
\Big{(}\ins \frac{(1+\nu_1)(1+C_0\epsilon^\mathsf{p})}{2+4C_0\epsilon^\mathsf{p}}
+C_0\epsilon^\mathsf{p}\Big{)}<
\frac{\ins}{1+2C_0\epsilon^\mathsf{p}}.\end{equation}
This is possible because $\frac{1+\nu_1}{2}<1$. We also ask that
$C_0\epsilon^\mathsf{p}<1/2$. Note that, although condition (\ref{diamant}) looks horrible, it is designed to work well with Lemma \ref{Gates}.

Let us introduce a first decomposition of the energy layer.
Recall that we defined $\mathcal{W}_0$ in (\ref{frite}) as the external part of the energy layer. We define
 $\close := \{\rho \in \mathcal{E} \backslash \ext ; d(\rho, K) <
\epsilon/2)\}$ for the part of the energy layer close to the trapped set, and
$\inter:= \{\rho \in \mathcal{E} \backslash \ext ; d(\rho, K) \geq
\epsilon/2)\}$ for the intermediate region. See figure \ref{schema} for a
representation of these different sets.
Note that we will later introduce a finer open cover of the energy layer,
using the sets $W_a$ appearing in the statement of the theorem.

\begin{figure}

   \includegraphics[scale=0.4]{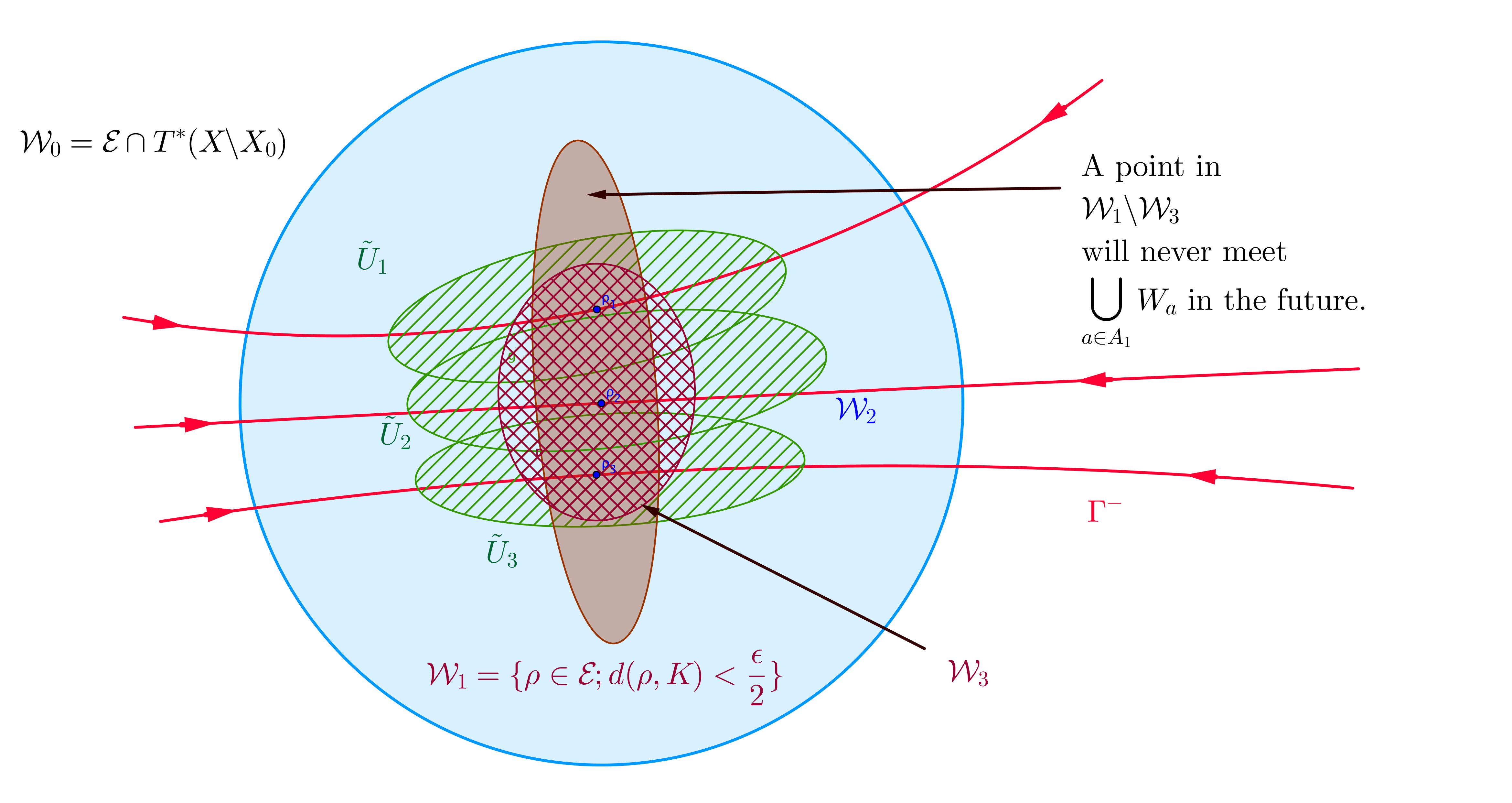}
    \caption{A representation of some of the different sets we introduce
in the proof of Theorem \ref{Cyril}, intersected with a Poincaré section.} \label{schema}
\end{figure}

The following lemma tells us that the set $\inter$ is a transient set,
that is to say, points spend only a finite time inside it.

\begin{lemme} \label{abeille}
There exists $N_\epsilon\in \mathbb{N}$ an integer which depends on
$\epsilon$ such that $\forall \rho \in \inter$, we have either
$\Phi^{\exit} (\rho) \in \ext$ or $\Phi^{-\exit} (\rho) \in \ext$.
\end{lemme}
\begin{proof}
This result comes from the uniform transversality of the stable and
unstable manifolds (which is a direct consequence of the compactness of
$K$).

It gives us the existence of a $d_1(\epsilon)>0$ such that, for
all $\rho \in \inter\cup \close$,
\begin{equation*}d(\rho,\Gamma^+)+d(\rho,\Gamma^-)\leq 2d_1 \Rightarrow d(\rho, K)
\leq \epsilon/2.
\end{equation*}
We may therefore write
\begin{equation*}\inter = \{ \rho \in \inter; d(\rho, \Gamma^-)>d_1\} \cup \{ \rho \in
\inter; d(\rho, \Gamma^-)>d_1\}.
\end{equation*}
A point in the first set will leave the interaction region in finite time
in the future, while a point in the second set will leave it in finite
time in the past. By compactness, we can find a uniform $N_\epsilon$ as
the one in the statement of the lemma.
\end{proof}

The following lemma is a consequence of the transversality assumption we made. It tells us that when we propagate $\Lag$ during a finite time $N$ and restrict it to a small set $\tilde{U}^\rho(\epsilon,\varrho)$ close to the trapped set, we obtain a finite union of Lagrangian manifolds in the alternative coordinates. Here, the size $\varrho$ of the set in the unstable direction depends on $N$, but its size $\epsilon$ in the stable direction does not.

\begin{lemme} \label{Ukraine}
Let $N\in \mathbb{N}$.
There exists $\mathcal{N}_N\in \mathbb{N}$, $\tilde{\varrho}_N>0$ and
$\tilde{\gamma}_N>0$ such that $\forall
0<\varrho\leq\tilde{\varrho}_N$, $\forall \rho \in
K$, $\forall 1\leq t \leq N$, $\Phi^{t}
(\mathcal{L}_0)\cap \tilde{U}^{\rho}(\epsilon,\varrho)$ can be
written in the coordinates $(\tilde{y}^{\rho},\tilde{\eta}^{\rho})$ as
the
union of at most $\mathcal{N}_N$ disjoint Lagrangian manifolds, which are all
$\tilde{\gamma}_N$-unstable :
\begin{equation*}\Phi^{t} (\mathcal{L}_0)\cap
\tilde{U}^{\rho}(\epsilon,\varrho)\equiv
\bigcup_{l=0}^{l(\varrho)} \hat{\Lambda}_{l}, 
\end{equation*}
with $l(\varrho)\leq \mathcal{N}_N$ and
\begin{equation*}\hat{\Lambda}_{l}=
\{(\tilde{y}_1^{\rho},\tilde{\boldsymbol{y}}^{\rho};0,f^{l}(\tilde{y}^{\rho})),
\tilde{\boldsymbol{y}}^{\rho}\in
D_{\varrho}\},
\end{equation*}
for some smooth functions $f^{l}$ with
$\|df^{l}(\tilde{y}^{\rho})\|_{C^0(D_{\epsilon})}\leq
\tilde{\gamma}_N$.
\end{lemme}
\begin{proof}
Let us consider a $1\leq t \leq N$.
First of all, $\Phi^{t}$ being a symplectomorphism, it sends Lagrangian
manifolds to Lagrangian manifolds. The restriction of a Lagrangian
manifold to a
region of phase space is a union of Lagrangian manifolds.

We now have to prove that, if we take $\varrho$ small enough, these
Lagrangian manifolds are all
$\tilde{\gamma}_N$
unstable, for some $\tilde{\gamma}_N>0$ which is independent of $\rho$.

Let $\rho\in K$. By
hypothesis, $W^-_{loc}(\rho)$ and $\Phi^t(\Lag)$ are transverse when they intersect.

Therefore, in a small neighbourhood of the stable manifold
$\{\tilde{\boldsymbol{y}}^\rho=0\}$, each connected component of $\Phi^t(\Lag)$ may be projected smoothly on the twisted unstable
manifold $\{\tilde{\boldsymbol{\eta}}^\rho=0\}$. That is to say, there exists a $\varrho>0$ and a $\gamma>0$ such that each connected component of
$\Phi^{t} (\Lag)\cap \tilde{U}^\rho(\epsilon,\varrho)$ is $\gamma$-unstable
in the twisted coordinates around $\rho$, for some $\gamma>0$.

Now, since the changes of coordinates between twisted coordinates are
continuous, we may use the compactness of $K$ to find uniform constants
$\varrho>0$ and $\gamma>0$ such that each connected component of
$\Phi^{t} (\Lag)\cap \tilde{U}^\rho(\epsilon,\varrho)$ is $\gamma$-unstable
in the twisted coordinates around $\rho$, independently of $\rho\in K$
and $1\leq t \leq N$.

By compactness of $\tilde{U}^{\rho}(\epsilon,\varrho)$, the number of Lagrangian manifolds making up $\Phi^{t}(\Lag) \cap
\tilde{U}^{\rho}(\epsilon,\varrho)$ is finite. This concludes the proof of the lemma.
\end{proof}

Applying this lemma to $N=\exit+2$, we define the following constants,
which we shall need later in the proof (recall that $\ins$ has been fixed).
\begin{equation}\label{diagnostic}
(\gamma_0,\varrho_0):=(\tilde{\gamma}_{\exit+2},\tilde{\varrho}_{\exit+2})
\end{equation}

\begin{equation}\label{rentree}
\wait:= \Big{\lfloor}\frac{\log(\ins/4\gamma_0)}{\log \big{(} (1+
\nu_1)/2 \big{)}}\Big{\rfloor}+1,
\end{equation}
\begin{equation*}\tins:= \wait+\exit+2,
\end{equation*}
\begin{equation} \label{theatre}
\varrho_1:= \min\big{(}\frac{\epsilon}{2\gamma_0},\varrho_0\big{)}, ~~~~~
\varrho_2:=\min \Big{(}
\big{(}\mathcal{C}+C_0\epsilon^\mathsf{p}\big{)}^{-\tins}
\varrho_1, \tilde{\delta}_{\tins}\Big{)}
\end{equation}
where $\mathcal{C}$
comes from Remark \ref{Collot}, and $C_0$ comes from equation
(\ref{buis}).
\begin{remarque}
As explained in Lemma \ref{abeille}, $\exit$ is the maximal time spent by a trajectory in the intermediate region $\inter$. The time $\wait$ will be the time necessary to incline a $\gamma_0$-unstable Lagrangian manifold to a $\ins$-unstable Lagrangian manifold, as explained in Proposition \ref{certifie}. As for the constant $\varrho_2$, it has been chosen small enough so that at each step of the aforementioned propagation during a time $\wait$, the Lagrangian manifolds we consider are contained in a single coordinate chart, as explained in Proposition \ref{certifie}.
\end{remarque}

\begin{remarque}
The constant $\secur$ in Theorem \ref{Cyril} will depend only on $\gamma_0$ and $\varrho_0$. Therefore, the proof of Lemma \ref{Ukraine} tells us that if we consider a whole family of Lagrangian manifolds $(\mathcal{L}_z)_{z\in Z}$ satisfying Hypothesis \ref{chaise} and Hypothesis \ref{Happy}, we will be able to find an $\secur\bel 0$ uniform in $z\in Z$ provided we have the following uniform transversality condition:
\begin{equation}\label{UT}
\forall t\in N, \forall\rho\in K, \exists \delta, \gamma \bel 0 \text{ such that } \forall z\in Z,~ \Phi^t(\mathcal{L}_z)\cap \tilde{U}^\rho(\epsilon,\delta) \text{ is } \gamma-\text{unstable}.
\end{equation}
\end{remarque}

\begin{lemme}\label{assimil}
There exists a neighbourhood $\set$ of $\Gamma^-\cap \close$ in
$\mathcal{E}$,
 a finite set of points $(\rho_i)_{i\in I}\subset K$ and
$0<\epsilon_1<\varrho_1,$ such that the following holds.

(i) The sets $\big{(}\tilde{U}_i \big{)}_{i\in I}
:=\big{(}\tilde{U}^{\rho_i}(\epsilon,\varrho_{2})\big{)}_{i\in I}$ form an
open cover of a neighbourhood of $\set$.

(ii)
$\rho\in \big{[}\close \backslash \set\big{]}\cup \{\rho'\in \inter;
d(\rho',\Gamma^-)>d_1\} \Longrightarrow  \forall t\geq 0, ~
d(\Phi^t(\rho)),K)\geq\epsilon_1.
$

(iii) For any open set $W$ of diameter $<\epsilon_1$ included in $\set$,
there exists an $i\in I$ such
that $W\subset
\tilde{U}_i$.
\end{lemme}

\begin{proof}
The sets $\big{(}\tilde{U}^\rho(\epsilon,\varrho_{2})\big{)}_{\rho\in K}$
form an open cover of a neighbourhood of $(\Gamma^-\cap \close)$. Let us
denote by $\set$ such a neighbourhood.

By compactness, we may extract from it a
finite open cover $\big{(}\tilde{U}_i \big{)}_{i\in I}
:=\big{(}\tilde{U}^{\rho_i}(\epsilon,\varrho_{2})\big{)}_{i\in I}$, which
still satisfies (i).

Since $\set$ is a
neighbourhood of $\Gamma^-\cap \close$,
there exists a constant $\varrho_2'>0$ such that the
following holds:
\begin{equation*}\forall\rho\in \close \backslash \set\text{, we have } d(\rho,
\Gamma^-)>\varrho_2'.
\end{equation*}

Therefore, there exists $0<\epsilon_1<\min(\varrho_1,\epsilon)$ such
that
\begin{equation*}
\rho\in \big{[}\close \backslash \set \big{]}\cup \{\rho'\in \inter;
d(\rho',\Gamma^-)\geq d_1\} \Longrightarrow  \forall t\geq 0, ~
d(\Phi^t(\rho)),K)>\epsilon_1,
\end{equation*}
which is (ii). Finally, since the set $\tilde{U}_i$ are open, we may shrink $\epsilon_1$ so that (iii) is satisfied.
\end{proof}

\begin{remarque} \label{grantorino}
The constant $\secur$ appearing in Theorem \ref{Cyril} will be smaller
that $\epsilon_1$ (see Lemma \ref{torticoli}), therefore each of the sets
$(W_a)_{a\in A_1}$  will be contained in some $\tilde{U}_i$.
Furthermore, we will have  $W_a\subset
\{\rho\in \mathcal{E}; d(\rho,K)<\secur\}$. Hence, a point
$\rho\in \big{[}\close \backslash \set \big{]}\cup \{\rho'\in \inter;
d(\rho',\Gamma^-)\geq d_1\}$ will not be contained in any of
the sets $(W_a)_{a\in A_1}$ when propagated in the future.
\end{remarque}

Lemma \ref{Ukraine} tells us that $\Phi^{N_\epsilon}(\Lag)\cap
\tilde{U}_i$ consists of finitely many
$\gamma_0$-unstable Lagrangian manifolds.
Our aim will now be to take a Lagrangian manifold included in a $
\tilde{U}_{i_1}$, to propagate it during some time $N\geq N_1$, then to
restrict it to a $\tilde{U}_{i_2}$, for
$i_1,i_2\in I$. The remaining part of the Lagrangian, which is
in $\close \backslash \set$, will not meet the sets
$(W_a)_{a\in A_1}$ when propagated in the future, as explained in Remark \ref{grantorino}.

\subsubsection{Propagation in the sets $\tilde{U}_i$} \label{Poutine}
For $N\in \mathbb{N}$ and
$\iota=(i_0i_1...i_{N-1})\in I^{N}$, we define
\begin{equation*}\Phi_{\iota}(\Lambda):= \Phi^1(\tilde{U}_{i_{N-1}}\cap
\Phi^1(...\Phi^1(\tilde{U}_{i_0}\cap\Lambda)...)).
\end{equation*}

The propagation of Lagrangian manifolds in the sets $\tilde{U}_i$ is
described in the following proposition,
which is the cornerstone of the proof of Theorem \ref{Cyril}.
Recall that $\ins$ was chosen arbitrarily at the beginning of the proof, and that $\wait$ was defined in (\ref{rentree}).

\begin{proposition}\label{certifie}
Let $N\geq \wait$, $\iota=(i_0i_1...i_{N-1})\in I^{N}$ and $i\in I$. Let
$\Lambda^0\subset \tilde{U}_{i_0}$ be an isoenergetic Lagrangian manifold
which is $\gamma_0$-unstable in the twisted coordinates centred on
$\rho_{i_0}$.
Then $\tilde{U}_i\cap\Phi_\iota (\Lambda)$ is a Lagrangian manifold
contained in
$\tilde{U}_{i}$, and it is
$\frac{\ins}{(1+2C_0\epsilon^\mathsf{p})^2}$-unstable in the
twisted coordinates centred on
$\rho_{i}$.
\end{proposition}
\begin{proof}
The first part of the proof consists in understanding how
$\Phi^n(\Lambda^0)$ behaves for $n\leq \wait$, in the twisted coordinates
centred on $\rho_{i_0}$. This is the content of the following lemma, which
is an adaptation to our context of the
“Inclination lemma” (See \cite[Theorem 6.2.8]{KH}; see also \cite[Proposition 5.1]{NZ} for a statement closer to our context and notations).
\begin{lemme} \label{Benjamin}
$\Phi^{\wait} (\Lambda^0)$ is a Lagrangian manifold,
which can be
written in the chart $(\tilde{y}^{\Phi^{\wait}(\rho_{i_0})} ,
\tilde{\eta}^{\Phi^{\wait}(\rho_{i_0})} )$ in the form:
\begin{equation*}\Phi^{\wait} (\Lambda^0)\equiv
\{(\tilde{y}_1^{\Phi^{\wait}(\rho_{i_0})},\tilde{\boldsymbol{y}}^{\Phi^{\wait}(\rho_{i_0})};0,f^{\wait}(\tilde{\boldsymbol{y}}^{\Phi^{\wait}(\rho_{i_0})})),
\tilde{\boldsymbol{y}}^{\Phi^{\wait}(\rho_{i_0})}\in D_{\wait}\},
\end{equation*}
with $D_{\wait}\subset B(0,\varrho_1)$ and
$\|df^{\wait}\|_{C^0(D_k)}\leq
\frac{(1+\nu_1)\ins}{4}$.
\end{lemme}
Note that $\Phi^{\wait} (\Lambda^0)$ is a priori not contained in a single
set $\tilde{U}_i$, but the lemma states that it is contained in the set
$\tilde{U}^{\Phi^{\wait}(\rho_{i_0})}(\epsilon,\varrho_1)$, where the
twisted coordinates are well defined.
\begin{proof}

By assumption, $\Lambda^0$ may be put in the form
\begin{equation*}\Lambda^0\equiv
\{(\tilde{y}_1^{\rho_{i_0}},\tilde{\boldsymbol{y}}^{\rho_{i_0}};0,f^0(\tilde{\boldsymbol{y}}^{\rho_{i_0}})),
|\tilde{\boldsymbol{y}}^{\rho_{i_0}}|<\varrho_2\},~~~~
\text{with  } \|df^0(\tilde{\boldsymbol{y}}^{\rho_{i_0}})\|_{C^0}\leq
\gamma_0.
\end{equation*}

We will consider restrictions of the Lagrangian manifolds at intermediate
times to the Poincaré sections centred at $\Phi^k(\rho_{i_0})$:
\begin{equation*}\Lambda^k_{sec} := \Phi^k (\Lambda^0) \cap
\Sigma^{\Phi^k(\rho_{i_0})}(\epsilon,\varrho_0).
\end{equation*}
We have $\Lambda_{sec}^{k+1}= \kappa^k (\Lambda_{sec}^{k})$, where
$\kappa^k:= \kappa_{\Phi^k(\rho_{i_0}),\Phi^{k+1}(\rho_{i_0})}$ is of the form
(\ref{Claire}).
From the equation (\ref{Claire}) and the definition of $\mathcal{C}$, we
see that the maximal rate of expansion in
the unstable direction is bounded by $(\mathcal{C}+C_0\epsilon^\mathsf{p})$.
Therefore, the definition of $\varrho_2$ implies
that for any $k\leq \wait$, the projection of $\Lambda^k_{sec}$
on the
unstable direction is supported in $B(0,\varrho_1)$.

To lighten the notations, we will write $\tilde{\boldsymbol{y}}^k$ and $\tilde{\boldsymbol{\eta}}^k$
instead of $\tilde{\boldsymbol{y}}^{\Phi^k(\rho_{i_0})}$ and $\tilde{\boldsymbol{\eta}}^{\Phi^k(\rho_{i_0})}$.

Let $k \geq 0$, and suppose we may write
\begin{equation*}\Lambda^k_{sec}\equiv \{(\tilde{\boldsymbol{y}}^k,f^k(\tilde{\boldsymbol{y}}^k)); \tilde{\boldsymbol{y}}^k\in
D_k \},
\end{equation*}
where $D_k\subset B(0,\varrho_1)$, and $\|df^{k}\|_{C^0}\leq \gamma_k$ for
some $0<\gamma_{k}\leq \gamma_0$.

Note that the key point in the following computations is that, since we
have chosen “alternative”
coordinates, we have $|\partial_{\boldsymbol{\eta}}\tilde{\alpha}^k(\tilde{\boldsymbol{y}}^k,\tilde{\boldsymbol{\eta}}^k)|\leq C_0 \tilde{\boldsymbol{y}}^k\leq C_0
\varrho_1$.

The projection of $\Phi^1_{|\Lambda^k_{sec}}$ on the horizontal subspace
reads
\begin{equation*}\tilde{\boldsymbol{y}}^k \mapsto \tilde{\boldsymbol{y}}^{k+1}= \pi \Phi^1
(\tilde{\boldsymbol{y}}^k,f^k(\tilde{\boldsymbol{y}}^k))= A_k \tilde{\boldsymbol{y}}^k+
\tilde{\alpha}^k(\tilde{\boldsymbol{y}}_k, f^k(\tilde{\boldsymbol{y}}_k)),
\end{equation*}
where for each $k$, $A_k$ is a matrix as in (\ref{fourmi}). 

By differentiating, we obtain :

\begin{equation*}\frac{\partial \tilde{\boldsymbol{y}}^{k+1}}{\partial \tilde{\boldsymbol{y}}^k}= A_k +
\frac{\partial \tilde{\alpha}^k}{\partial \tilde{\boldsymbol{y}}^k} + \frac{\partial
\tilde{\alpha}^k}{\partial \tilde{\boldsymbol{\eta}}^k}\frac{\partial f_k}{\partial
\tilde{\boldsymbol{y}}^k}= A_k+r_k,
\end{equation*}
where $r_k$ has entries bounded by $C_0 \varrho_1 \gamma_0\leq C_0 \epsilon$.

Therefore, the map is invertible, and $\tilde{\boldsymbol{y}}^{k+1}\mapsto \tilde{\boldsymbol{y}}^k$
is contracting. This implies that $\Lambda_{sec}^{k+1}$ can be represented
as a graph
\begin{equation*}\Lambda^{k+1}_{sec}\equiv \{(\tilde{\boldsymbol{y}}^{k+1},f^{k+1}(\tilde{\boldsymbol{y}}^{k+1}));
\tilde{\boldsymbol{y}}^{k+1}\in
D_{k+1} \},
\end{equation*}
with 
\begin{equation*}
f^{k+1}(\tilde{\boldsymbol{y}}^{k+1})=  {^t }A_k^{-1} f^k
(\tilde{\boldsymbol{y}}^k)+\tilde{\beta}_k(\tilde{\boldsymbol{y}}^k,f^k(\tilde{\boldsymbol{y}}^k)).
\end{equation*}

Differentiating with respect to $\tilde{\boldsymbol{y}}^{k+1}$, we get

\begin{equation*}\frac{\partial f^{k+1}}{\partial \tilde{\boldsymbol{y}}^{k+1}} =
\Big{(}\frac{\partial \tilde{\boldsymbol{y}}^{k}}{\partial \tilde{\boldsymbol{y}}^{k+1}} \Big{)}
\Big{[} (^t A_k^{-1}+
\partial_{\boldsymbol{\eta}}\tilde{\beta}^k(\tilde{\boldsymbol{y}}^k,f^k(\tilde{\boldsymbol{y}}^k)))\frac{\partial
f^k}{\partial\tilde{\boldsymbol{y}}^k} (\tilde{\boldsymbol{y}}^k) + \partial_{\tilde{\boldsymbol{y}}} \tilde{\beta}^k
(\tilde{\boldsymbol{y}}^k,f^k(\tilde{\boldsymbol{y}}^k))\Big{]}
\end{equation*}
Therefore, we have
\[\begin{aligned}
\Big{\|}\frac{\partial f^{k+1}}{\partial \tilde{\boldsymbol{y}}^{k+1}}\Big{\|} &\leq
\frac{\|^tA_k^{-1}\| \gamma_k + |\partial_{\tilde{\boldsymbol{y}}}\tilde{\beta}| +
|\partial_{\tilde{\boldsymbol{\eta}}}\tilde{\beta}| \gamma_k}{\nu^{-1}
 -|\partial_{\tilde{\boldsymbol{y}}}\tilde{\alpha}|  - |\partial_{\tilde{\boldsymbol{\eta}}}\tilde{\beta}|  \gamma_k}\\
&\leq \frac{\gamma_k\nu + C_0\epsilon^\mathsf{p} (1 +\gamma_k)
}{\nu^{-1}-2C_0\epsilon^\mathsf{p}}\\
&\leq \nu_1 \gamma_k + \frac{(1-\nu_1)\ins}{8}= \gamma_k
\big{(}\nu_1+\frac{\ins(1-\nu_1)}{8\gamma_k}\big{)},
\end{aligned}
\]
where the last inequality comes from (\ref{belette}).
First of all, the fact that this slope is bounded uniformly on
$\Lambda^{k+1}_{sec}$ implies that $\Lambda^{k+1}_{sec}$ can indeed be
written in the form
\begin{equation*}\Lambda^{k+1}_{sec}\equiv \{(\tilde{\boldsymbol{y}}^{k+1},f^{k+1}(\tilde{\boldsymbol{y}}^{k+1});
\tilde{\boldsymbol{y}}^{k+1}\in D_{k+1} )\},
\end{equation*}
where $D_{k+1}\subset B(0,\varrho_1)$, and $\|df^{k+1}\|_{C^0}\leq
\gamma_{k+1}$, where $\gamma_{k+1} \leq \gamma_k
\big{(}\nu_1+\frac{\ins(1-\nu_1)}{8\gamma_k}\big{)}$.

Now, if $\gamma_k> \ins/4$, then $\nu_1+\frac{\ins(1-\nu_1)}{8\gamma_k}
<\frac{1+\nu_1}{2}<1$, so that $\gamma_k$ decreases exponentially fast,
while if $\gamma_k\leq \frac{(1+\nu_1)\ins}{4}$, then $\gamma_{k+1}<
\frac{(1+\nu_1)\ins}{4}$.

The time $\wait$ has been chosen large enough so that $\gamma_{\wait}<
\frac{(1+\nu_1)\ins}{4}$, which concludes the proof of the lemma.
\end{proof}

After times $N>\wait$, the Lagrangian manifold may not be included in
$\tilde{U}^{\Phi^N(\rho_{i_0})}(\epsilon,\varrho_1)$. Therefore, we may have to
change of coordinates. By Lemma \ref{Benjamin}, at time $\wait$, our
Lagrangian manifold $\Phi^{\wait}(\Lambda^0)$ is included in
$\tilde{U}^{\Phi^{\wait}(\rho_{i_0})}(\epsilon,\varrho_1)$ and is
$\frac{(1+\nu_1)\ins}{4}$-unstable.

We want to study $\tilde{U}_j\cap \Phi^{\wait}(\Lambda^0)$ for $j\in I$,
in the coordinates centred at $\rho_j$, and to
apply the computations made in the proof of Lemma \ref{Benjamin} again.
Let us see how all this works.

If, for some $j\in I$, $\tilde{U}_j\cap \Phi^{\wait}(\Lambda^0)\neq
\emptyset$, then $d(\Phi^{\wait}(\rho_{i_0}),\rho_j)<\epsilon$.
Now, by applying Lemma \ref{Gates} as well as equation (\ref{diamant}), we
obtain that
$\Phi^{\wait}(\Lambda^0)\cap \tilde{U}_j$ is
$\frac{\ins}{2}$-unstable in the twisted coordinates centred at $\rho_j$.

We may continue this argument of changing coordinates and propagating to
any time $N\geq \wait$: we always obtain a single Lagrangian manifold
which is
$\frac{(1+\nu_1)\ins}{4}$-unstable. This concludes the proof of
Proposition \ref{certifie}, because we assumed that
$C_0\epsilon^\mathsf{p}<1/2$.
\end{proof}
\begin{remarque} \label{urss}
In \cite{NZ}, Proposition 5.1, the authors prove using the chain rule that for each $\ell\in
\mathbb{N}$, there exists a constant $\mathcal{C}_\ell$ large enough such
that the following holds. If $i_1,i_2\in I$
and if $\Lambda\subset \tilde{U}_{i_1}$ is a
Lagrangian manifold in some unstable cone, generated by a function $f$ in
the coordinates $(\tilde{y}^{\rho_{i_1}},\tilde{\eta}^{\rho_{i_1}})$ with
$\|f\|_{C^\ell}\leq \mathcal{C}_\ell$, then $\Phi^1(\Lambda)\cap
\tilde{U}_{i_2}$ is a
union of finitely many Lagrangian manifolds, all of which are in some
unstable cone in the coordinates
$(\tilde{y}^{\rho_{i_2}},\tilde{\eta}^{\rho_{i_2}})$, and are generated by
functions with a $C^\ell$ norm smaller than $\mathcal{C}_\ell$.

In particular, this shows that on the Lagrangian manifold
$\Phi^N_\iota(\Lambda)$ described
in Proposition \ref{certifie}, the function $s^{\rho_i}(y^{\rho_i})$ has a $C^\ell$ norm smaller than $\mathcal{C}_\ell$, where
$\mathcal{C}_\ell$ is a constant independent on $N$.

\end{remarque}

\subsubsection{Properties of the sets $(W_a)_{a\in A_1}$}
\label{marguerite}
The following lemma is an adaptation of Lemma \ref{Ukraine} to the
“straight coordinates”. Note that the main reason why we want to use these straight coordinates is because they are symplectic, which will play a crucial role in the proof of Theorem \ref{ibrahim}.

\begin{lemme} \label{torticoli}
There exists $\secur<\epsilon_1$ such that, if $(W_a)_{a\in A_1}$ is an
adapted cover of $K$ of diameter $\secur$ such that for each $a\in A_1$, $W_a\cap \ext=\emptyset$, and there exists a point $\rho_a\in W_a\cap K\neq
\emptyset$. Then there exist
$\mathcal{N}_{\tins}\in \mathbb{N}$ and $\gamma'$ such that the following holds.

For each $a\in A_1$, for each
$1\leq N \leq \tins$, the set $\Phi^{N}(\Lag)\cap W_a$ consists of
at most $\mathcal{N}_{\tins}$ Lagrangian manifolds, all of which are 
$\gamma'$-unstable in the straight coordinates centred on $\rho_a$.
\end{lemme}

\begin{proof}
Let us choose $\varepsilon_0>0$ small enough so that $C_0 \varepsilon_0\tilde{\gamma}_{\tins} <1$ and such that each set of diameter
smaller that $\varepsilon_0$ and which intersects $K$ is contained in
some $\tilde{U}^\rho(\epsilon,\delta)$, with
$\delta<\tilde{\delta}_{\tins}$. 
By applying Lemma \ref{Ukraine}, we know that
there exists $\mathcal{N}_{\tins}\in \mathbb{N}$,
$\tilde{\delta}_{\tins}>0$ and
$\tilde{\gamma}_{\tins}>0$ such that $\forall
0<\delta\leq\tilde{\delta}_{\tins}$, $\forall \rho \in
K$, $\forall 1\leq N \leq {\tins}$, $\Phi^{N}
(\mathcal{L}_0)\cap \tilde{U}^{\rho}(\epsilon,\delta)$ can be
written in the coordinates $(\tilde{y}^{\rho},\tilde{\eta}^{\rho})$ as
the
union of at most $\mathcal{N}_{\tins}$ Lagrangian manifolds, which are all
$\tilde{\gamma}_{\tins}$-unstable. This gives us the statement in the
twisted coordinates. To go to the straight coordinates, we may simply use
Lemma \ref{yaki} thanks to the assumption made on $\varepsilon_0$.
\end{proof}

For any $a\in A_1$, $1\leq k \leq \tins$,  $W_a\cap \Phi^{k}(\Lag)$
consists of finitely many Lagrangian manifolds.
 Let us define $\mathsf{d}_{a,k}$ as the minimal
distance (with respect to the distance $d$) between the Lagrangian manifolds which make up $ W_a\cap
\Phi^{k}(\Lag),$
with the convention that this quantity is equal to $+\infty$ if  $W_a\cap
\Phi^{k}(\Lag)$  consists of a single Lagrangian manifold or is empty. We then set

  \begin{equation*}\mathsf{d}:=\min(\varepsilon_0,\min\limits_{\underset{1\leq k \leq
\tins}{a\in A_1} } \{ \mathsf{d}_{a,k}\})>0.
\end{equation*}

\begin{remarque}
If we consider a whole family of Lagrangian manifolds $(\mathcal{L}_z)_{z\in Z}$ satisfying Hypothesis \ref{chaise} and Hypothesis \ref{Happy}, we will be able to apply Theorem \ref{Cyril} to them with sets $(W_a)_{a\in A_2}$ independent of $z\in Z$ provided the constant $\mathsf{d}$ is well-defined, that is to say, provided we have
\begin{equation}\label{cognet}
\inf\limits_{\underset{1\leq k \leq
\tins}{a\in A_1, z\in Z} } \{ \mathsf{d}^ z_{a,k}\}>0,
\end{equation}
where $\mathsf{d}^ z_{a,k}$ is the minimal
distance between the Lagrangian manifolds which make up $ W_a\cap
\Phi^{k}(\mathcal{L}_z),$
with the convention that this quantity is equal to $+\infty$ if  $W_a\cap
\Phi^{k}(\mathcal{L}_z)$  consists of a single Lagrangian manifold or is empty.
\end{remarque}

The flow $(\Phi^t)$ is $C^1$ with respect to time, hence Lipschitz on $[0, \tins]$. Therefore, there exists a constant $C \bel 0$ such that for all $t\in [0,\tins]$, for all $\rho_1,\rho_2\in \mathcal{E}$, we have 
\begin{equation*}
d(\Phi^t(\rho_1),\Phi^t(\rho_2))\leq C d(\rho_1,\rho_2).
\end{equation*}
 We take 
\begin{equation*}\taille := \mathsf{d}/C.
\end{equation*}

We now complete $(W_a)_{a\in A_1}$ to cover the whole energy layer.

\subsubsection{Construction and properties of the sets $(W_a)_{a\in A_2}$}

Recall that $W_0= T^*(X\backslash X_0)$, and that $b$ is the boundary
defining function introduced in Hypothesis \ref{Guepard}.

We build the sets $(W_a)_{a\in A_2}$ so that, if we set $A=A_1\cup A_2\cup
\{0\}$, the following holds:
\begin{itemize}
\item  Each of the sets $(W_a)_{a\in A_2}$ has a diameter smaller than
$\taille $.

\item For each $a\in A_2$, we have $d(W_a, K)>\taille/2$.

\item $(W_a)_{a\in A}$ is an open cover of $\mathcal{E}$.
\end{itemize}

Our next lemma is the first brick of the proof of the uniqueness of the
Lagrangian manifold making up $\Phi^N_\alpha(\Lag)$. It relies on the fact
that the sets $(W_a)_{a\in A_2}$ have been built small enough.

\begin{lemme}\label{tisane}
Let $k\leq \tins$, $\alpha\in A^{k}$, and $a\in A_1$. Then the set
$W_a\cap\Phi_{\alpha}^{k}(\Lag)$ is empty or consists of a single Lagrangian manifold.
\end{lemme}
\begin{proof}
Let us suppose that $\Phi^{k}(\Lag)\cap W_a$ is non-empty.
We have seen in Lemma \ref{torticoli} that it consists of
finitely many Lagrangian manifolds, with a distance between them larger
than $\mathsf{d}$.
Therefore, for any $1\leq k'\leq k$, the sets $\Phi^{-k'}
(\Phi^{k}(\Lag)\cap W_a)$ consist of Lagrangian manifolds which are at a
distance larger than $\taille$ from each other. Because of the assumption (\ref{descendance}) we made, we have $\alpha_{k'}\in A_2$ for some $k'\leq k$. Since the sets
$(W_a)_{a\in A_2}$ have a diameter smaller than $\taille$, they separate
the Lagrangian manifolds which make up $\Phi^{-k'} (\Phi^{k}(\Lag)\cap
W_a)$. We deduce from this the lemma.
\end{proof}

\subsubsection{Structure of the admissible sequences}
We will now state two of lemmas which put some constraints on the
sequences $\alpha\in A^{N}$, with $\alpha_N\in A_1$ such that
$\Phi^N_{\alpha}(\Lag)\neq \emptyset$.

The first of these lemmas tell us that we may restrict ourselves to
sequences such that $\alpha_k \neq 0$ for $k\geq 1$.
\begin{lemme}\label{violoncelle}
Let $N\in \mathbb{N}$, and let $\alpha\in A^{N}$, and $a\in A_1$.
Suppose that $\alpha_k=0$ for some $1\leq k \leq N-1$, and
that $W_a\cap \Phi^N_{\alpha}(\Lag)\neq \emptyset$.
Then $W_a\cap \Phi^N_{\alpha}(\Lag)\subset\Phi^{N-k}_{\alpha_{k+1}...\alpha_{N-1}}(\Lag)$.
\end{lemme}

\begin{proof}
By hypothesis,
$\Phi_{\alpha_1...\alpha_k}^k(\Lag)\subset \mathcal{W}_0$, and it intersects
$\close$ in the future. We have $\mathcal{W}_0=\inco \cup \mathcal{DE}_+$, and a point in $\mathcal{DE}_+$ cannot intersect $\close$ in the future. Therefore, the points in $\Phi_{\alpha_1...\alpha_k}^k(\Lag)$ which intersect $\close$ in he future are all in $\mathcal{DE}_-$. But by Lemma \ref{saladin}, the point in $\mathcal{DE}_-$ can only have pre-images in $W_0$. Therefore, we have $$W_a\cap \Phi^N_{\alpha}(\Lag)\subset W_a\cap \Phi^N_{0...0\alpha_{k+1}...\alpha_{N-1}}(\Lag)\subset \Phi^{N-k}_{\alpha_{k+1}...\alpha_{N-1}}(\Lag),$$
where the second inclusion comes from Hypothesis \ref{chaise}.
\end{proof}

Let us now take advantage of Remark \ref{grantorino} to show that, from
time $k\geq \exit +2$, all the interesting dynamics takes place in $\set$.

\begin{lemme}\label{cardiff}
Let $N\geq N_\epsilon+2$, $\alpha\in A^{N}$ with $\alpha_i\neq 0$ for
$i\geq 1$.

Let $\exit +2\leq k\leq N$, and $\rho\in \Phi^k_{\alpha_1...\alpha_k}(\Lag)$ be
such that $\Phi^{N-k}(\rho)\in W_{a}$ for some $a\in A_1$. Then $\rho\in
\set$.
\end{lemme}
\begin{proof}
If $\rho\in \close$, then the result follows from Remark \ref{grantorino}.
We must therefore check that we cannot have $\rho\in \inter\cup \ext$.
First of all, note that Lemma \ref{saladin} implies that we cannot have
$\rho\in \ext$. This lemma also implies that for each $a'\in A_1\cup A_2$, we have
\begin{equation}\label{teigne}
\Phi^1( W_{a'}\backslash W_0)\cap \inco=\emptyset.
\end{equation}

Suppose now that $\rho\in \inter$. Since $k\geq N_\epsilon+2$, and
$\alpha_i\neq 0$ for $i\geq 1$, we have $\Phi^{-N_\epsilon-1}(\rho)\in
W_{a'}$ for some $a'\in A_1\cup A_2$. Therefore, by
equation  (\ref{teigne}), we have $\Phi^{-\exit}(\rho)\notin \ext$.

By the proof of Lemma \ref{abeille}, this would imply that
$d(\rho,\Gamma^-)\geq d_1$. By Remark \ref{grantorino}, this implies
that we cannot have $\Phi^{N-k}(\rho)\in W_{a}$ for some $a\in A_1$, a
contradiction.
\end{proof}

\subsubsection{End of the proof of Theorem \ref{Cyril}}
Let $N\geq 0$, $\alpha\in A^{N}$ and $a\in A_1$. If $N\leq
\tins$, the result of Theorem \ref{Cyril} is a consequence of Lemma
\ref{torticoli} and Lemma \ref{tisane}.

Consider now $N\geq \tins>N_\epsilon+2$. We will assume that
$W_a\cap\Phi_\alpha^N(\Lag)\neq \emptyset$.
Thanks to Lemma \ref{violoncelle} and to Hypothesis \ref{chaise}, we may
assume that $\alpha_i\neq 0$ for all $i\geq 1$.

From Lemma \ref{cardiff}, we deduce that
\begin{equation} \label{voyageur}
W_a\cap\Phi_\alpha^N(\Lag)\subset
\bigcup\limits_{\underset{\iota_{N-N_\epsilon}=i_\alpha}{\iota \in
I^{N-\exit-1}}} \Phi_{\iota}
\big{(}\Phi_{\alpha_{1}...\alpha_{\exit+2}}^{\exit+2}(\Lag) \big{)},
\end{equation}
where $i_\alpha\in I$ is such that $W_{\alpha_N}\subset
\tilde{U}_{i_\alpha}$.

Let us define $\Lambda_k:=\{\rho\in \Phi_\alpha^k(\Lag);  \forall k'\geq 0,
\Phi^{k'}(\rho)\in W_{\alpha_{k+k'}} \}.$

By Lemma \ref{cardiff}, for each $k\geq N_\epsilon+2$, we have
$\Lambda_k\subset \set\cap W_{\alpha_k}$. Therefore, by Lemma
\ref{assimil} (iii), there exists a $i_k\in I$ such that $\Lambda_k\subset
\tilde{U}_{i_k}$, and we obtain that
\begin{equation*} W_a\cap\Phi_\alpha^N(\Lag)\subset
\Phi_{i_{N_\epsilon+2}...i_N}^{N-N_\epsilon-2} \big{(}
\Phi_{\alpha_{1}...\alpha_{\exit+2}}^{\exit+2}(\Lag) \big{)}.
\end{equation*}

We know from Lemma \ref{Ukraine} and Lemma \ref{tisane} that
$\Phi_{\alpha_{1}...\alpha_{\exit+2}}^{\exit+2}(\Lag)$ consists of a
single Lagrangian manifold, which is $\gamma_0$-unstable in the
coordinates centred on any point of $K$.
Applying Proposition \ref{certifie},
we know that the right hand side of (\ref{voyageur}) is a Lagrangian
manifold which
is
$\frac{\ins}{(1+2C_a\epsilon^\mathsf{p})^2}$-unstable in the twisted
coordinates centred on $\rho_{i_\alpha}$.

We first apply Lemma \ref{Gates} to write this Lagrangian manifold in
the twisted coordinates centred on $\rho_a$. Thanks to equation
(\ref{diamant}), it is
$\frac{\ins}{(1+2C_a\epsilon^\mathsf{p})}$-unstable.
We then use Lemma \ref{yaki} to write this Lagrangian manifold in the
straight coordinates centred on $\rho_{\alpha_N}$, and we deduce that it is
$\ins$-unstable. This concludes the proof of Theorem \ref{Cyril}.
\end{proof}

 \begin{remarque}\label{peace}
 Therefore, in the
coordinates $(y^{a},\eta^{a})$, $W_a\cap \Phi_\alpha^N(\Lag)$ may be put in the form
\begin{equation*}
W_a\cap \Phi_\alpha^N(\Lag) \equiv \{(y_1^{a}, \boldsymbol{y}^{a},0, f_{N,\alpha,a}(\boldsymbol{y}^a)), y^{a}\in D_{N,\alpha,a} \},
\end{equation*}
for some open set $D_{N,\alpha,a}\subset \mathbb{R}^d$.

 Remark \ref{urss} tells us that for any $\ell\in \mathbb{N}$, the functions $f_{N,\alpha,a}$ have $C^\ell$ norms which are
bounded independently of $N$, $\alpha$ and $a$.
 \end{remarque}

\section{Generalized eigenfunctions} \label{GE}

We shall state our results about generalized eigenfunctions under rather general assumptions. We shall then explain why these assumptions hold in the case of distorted plane waves on manifolds which are Euclidean near infinity.

In the sequel, we will consider a Riemannian manifold $(X,g)$ with a real-valued potential $V\in C_c^\infty(X)$, and define the Schrödinger operator
\begin{equation*}
P_h= -h^2\Delta_g -c_0 h^2 + V(x).
\end{equation*}
Here $c_0\bel 0$ is a constant, which will be 0 in the case of Euclidean near infinity manifolds (see \ref{origami} for the definition of such manifolds).

Before stating our assumptions, let us recall a few definitions and facts from semiclassical analysis.

\subsection{Refresher on semiclassical analysis}
\subsubsection{Pseudodifferential calculus} \label{greve}

We shall use the class $S^{comp}(T^*X)$ of symbols $a\in C_c^\infty(T^*X)$,
which may depend on $h$, but whose seminorms and supports are all bounded
independently of $h$. We will sometimes write $S^{comp}(X)$ for the set of symbols in $S^{comp}(T^*X)$ which depend only on the base variable.
If $U$ is an open subset of $T^*X$, we will denote by $S^{comp}(U)$ the set of functions in $S^{comp}(T^*X)$ whose support is contained in $U$.

\begin{definition}\label{defsymbclassique}
Let $a\in S^{comp}(T^*Y)$. We will say that $a$ is a \emph{classical symbol} if there exists a sequence of symbols $a_k\in S^{comp}(T^*Y)$ such that for any $n\in \mathbb{N}$, 
$$a-\sum_{k=0}^n h^k a_k \in h^{n+1} S^{comp}(T^*Y).$$
We will then write $a^0(x,\xi):= \lim\limits_{h\rightarrow 0} a(x,\xi;h)$ for the \emph{principal symbol} of $a$.
\end{definition}

We associate to $S^{comp}(T^*X)$ the class of pseudodifferential operators
$\Psi_h^{comp}(X)$, through a surjective quantization map
\begin{equation*}Op_h:S^{comp}(T^*X)\longrightarrow \Psi^{comp}_h(X).
\end{equation*} This quantization
map is defined using coordinate charts, and the standard Weyl quantization
on $\mathbb{R}^d$. It is therefore not intrinsic. However, the principal
symbol map
\begin{equation*}\sigma_h : \Psi^{comp}_h (X)\longrightarrow S^{comp}(T^*X)/
h S^{comp}(T^*X)
\end{equation*} is intrinsic, and we have
\begin{equation*}\sigma_h(A\circ B) = \sigma_h (A) \sigma_h(B)
\end{equation*}
and
\begin{equation*}\sigma_h\circ Op: S^{comp}(T^* X) \longrightarrow S^{comp} (T^*X) /h
S^{comp}(T^*X)
\end{equation*}
is the natural projection map.

For more details on all these maps and their construction, we refer the reader
to \cite[Chapter 14]{Zworski_2012}.

For $a\in S^{comp}(T^*X)$, we say its essential support is equal to a given
compact $K\Subset T^*X$,
\begin{equation*} \text{ ess supp}_h a = K \Subset T^*X,
\end{equation*}
if and only if, for all $\chi \in S(T^*X)$,
\begin{equation*}supp \chi \subset (T^*X\backslash K) \Rightarrow \chi a \in h^\infty S(T^*
X).
\end{equation*}
For $A\in \Psi^{comp}_h(X), A=Op_h(a)$, we define the wave front set of $A$ as:
\begin{equation*}WF_h(A)= \text{ ess supp}_h a,
\end{equation*}
noting that this definition does not depend on the choice of the
quantisation. When $K$ is a compact subset of $T^*X$ and $WF_h(A)\subset K$, we will sometimes say that $A$ is \emph{microsupported} inside $K$.

Let us now state a lemma which is a consequence of Egorov theorem \cite[Theorem 11.1]{Zworski_2012}. Recall that $U(t)$ is the Schrödinger propagator $U(t)= e^{-it P_h/h}$.
\begin{lemme}\label{theclash}
Let $A,B\in \Psi_h^{comp}(X)$, and suppose that $\Phi^t( WF_h(A))\cap WF_h(B)=\emptyset$. Then we have
\begin{equation*}
A U(t) B= O_{L^2\rightarrow L^2}(h^\infty).
\end{equation*}
\end{lemme}

If $U,V$ are bounded open subsets of $T^*X$, and if $T,T' : L^2(X)\rightarrow L^2(X)$ are bounded operators, we shall say that $T\equiv T'$ \emph{microlocally} near $U\times V$ if there exist bounded open sets $\tilde{U}\supset \overline{U}$ and $\tilde{V} \supset \overline{V}$ such that for any $A,B\in \Psi_h^{comp}(X)$ with $WF(A)\subset \tilde{U}$ and $WF(B)\subset \tilde{V}$, we have
\begin{equation*}
A(T-T')B = O_{L^2\rightarrow L^2} (h^\infty)
\end{equation*}

\paragraph{Tempered distributions}
Let $u=(u(h))$ be an $h$-dependent family of distributions in $\mathcal{D}'(X)$. We say it is \emph{$h$-tempered} if for any bounded open set $U\subset X$, there exists $C\bel 0$ and $N\in \mathbb{N}$ such that
\begin{equation*}
\|u(h)\|_{H_h^{-N}(U)}\leq C h^{-N},
\end{equation*}
where $\|\cdot\|_{H_h^{-N}(U)}$ is the semiclassical Sobolev norm.

For a tempered distribution $u=(u(h))$, we say that a point $\rho\in T^*X$ does not lie in the wave front set $WF(u)$ if there exists a neighbourhood $V$ of $\rho$ in $T^*X$ such that for any $A\in \Psi_h^{comp}(X)$ with $WF(a)\subset V$, we have $Au=O(h^\infty)$.
\subsubsection{Lagrangian distributions and Fourier Integral Operators}\label{DSK}
\paragraph{Phase functions}
Let $\phi(x,\theta)$ be a smooth real-valued function on some open subset $U_\phi$ of $X\times \mathbb{R}^L$, for some $L\in \mathbb{N}$. We call $x$ the \emph{base variables} and $\theta$ the \emph{oscillatory variables}. We say that $\phi$ is a \emph{non degenerate phase function} if the differentials $d(\partial_{\theta_1} \phi)...d(\partial_{\theta_L}\phi)$ are linearly independent on the \emph{critical set }
\begin{equation*}
C_\phi:=\{ (x,\theta); \partial_\theta \phi =0 \} \subset U_\phi.
\end{equation*}
In this case
\begin{equation*}
\Lambda_\phi:= \{(x,\partial_x \phi(x,\theta)); (x,\theta)\in C_\phi \} \subset T^*X
\end{equation*}
is an immersed Lagrangian manifold. By shrinking the domain of $\phi$, we can make it an embedded Lagrangian manifold. We say that $\phi$ \emph{generates } $\Lambda_\phi$.

\paragraph{Lagrangian distributions}
Given a phase function $\phi$ and a symbol $a\in S^{comp}(U_\phi)$, consider the $h$-dependent family of functions
\begin{equation}\label{massai}
u(x;h)= h^{-L/2} \int_{\mathbb{R}^L} e^{i\phi(x,\theta)/h} a(x,\theta;h) d\theta.
\end{equation}
We call $u=(u(h))$ a \emph{Lagrangian distribution}, (or a \emph{Lagrangian state}) generated by $\phi$. By the method of non-stationary phase, if $supp ~a$ is contained in some $h$-independent compact set $K\subset U_\phi$, then
\begin{equation*}
WF_h(u)\subset \{(x,\partial_x \phi(x,\theta)); (x,\theta)\in C_\phi\cap K\}\subset \Lambda_\phi.
\end{equation*}

\begin{definition}\label{Grenoble}
Let $\Lambda\subset T^*X$ be an embedded Lagrangian submanifold. We say that an $h$-dependent family of functions $u(x;h)\in C_c^\infty(X)$ is a (compactly supported and compactly microlocalized) \emph{Lagrangian distribution associated to $\Lambda$}, if it can be written as a sum of finitely many functions of the form (\ref{massai}), for different phase functions $\phi$ parametrizing open subsets of $\Lambda$, plus an $O(h^\infty)$ remainder. We will denote by $I^{comp}(\Lambda)$ the space of all such functions.
\end{definition}

\paragraph{Fourier integral operators}
Let $X, X'$ be two manifolds of the same dimension $d$, and let $\kappa$ be a symplectomorphism from an open subset of $T^*X$ to an open subset of $T^*X'$. Consider the Lagrangian
\begin{equation*}
\Lambda_\kappa =\{(x',-\nu';x,\nu); \kappa(x,\nu)=(x',\nu')\}\subset T^*X'\times T^*X= T^*(X'\times X).
\end{equation*}
A compactly supported operator $U:\mathcal{D}'(X)\rightarrow C_c^\infty(X')$ is called a (semiclassical) \emph{Fourier integral operator} associated to $\kappa$ if its Schwartz kernel $K_U(x',x)$ lies in $h^{-d/2}I^{comp}(\Lambda_\kappa)$. We write $U\in I^{comp}(\kappa)$. The $h^{-d/2}$ factor is explained as follows: the normalization for Lagrangian distributions is chosen so that $\|u\|_{L^2}\sim 1$, while the normalization for Fourier integral operators is chosen so that $\|U\|_{L^2(X)\rightarrow L^2(X')} \sim 1$.

Note that if $\kappa\circ \kappa'$ is well defined, and if $U\in I^{comp}(\kappa)$ and $U'\in I^{comp}(\kappa')$, then $U\circ U'\in I^{comp} (\kappa\circ \kappa')$.

If $U\in I^{comp}(\kappa)$ and $O\subset T^*X$ is an open bounded set, we shall say that $U$ is \emph{microlocally unitary }near $O$ if $U^* U \equiv I_{L^2(X)\rightarrow L^2(X)}$ microlocally near $O\times \kappa (O)$.

\subsubsection{Local properties of Fourier integral operators}\label{chaussette}
In this section we shall see that, if we work locally, we may describe many Fourier integral operators without the help of oscillatory coordinates. In particular, following \cite[4.1]{NZ}, we will recall the effect of a Fourier integral operator on a Lagrangian distribution which has no caustics. We will recall in section \ref{birette} how this formalism may be applied to the study of the Schrödinger propagator.

Let $\kappa :T^*\mathbb{R}^d\rightarrow T^*\mathbb{R}^d$ be a local symplectic diffeomorphism. By performing phase-space translations, we may assume that $\kappa$ is defined in a neighbourhood of $(0,0)$ and that $\kappa(0,0)=(0,0)$. 

Without loss of generality, we can find linear Lagrangian subspaces, $\Gamma_j$, $\Gamma_j^\perp\subset T^*\mathbb{R}^d$, $j=0,1$, with the following properties:

\begin{itemize}
\item $\Gamma_j^\perp$ is transversal to $\Gamma_j$;
\item if $\pi_j$ (resp. $\pi_j^\perp$) is the projection $T^*\mathbb{R}^d\rightarrow \Gamma_j$ along $\Gamma_j^\perp$ (resp. the projection $T^*\mathbb{R}^d\rightarrow \Gamma_j^\perp$ along $\Gamma_j$), then, for some neighbourhood $U$ of $\rho_0$, the map
\begin{equation*}
\kappa(U) \times U\ni (\kappa(\rho),\rho)\mapsto \pi_1(\kappa(\rho))\times \pi_0^\perp\in \Gamma_1\times \Gamma_0^\perp
\end{equation*} 
is a local diffeomorphism from the graph of $\kappa|_U$ to a neighbourhood of the origin in $\Gamma_1\times \Gamma_0^\perp$.
\end{itemize}

Let $A_j$, $j=0,1$ be linear symplectic transformations with the properties
\begin{equation*}
A_j(\Gamma_j)=\{(x,0)\}\subset T^*\mathbb{R}^d \text{  and  } A_j(\Gamma_j^\perp)=\{(0,\xi)\}\subset T^*\mathbb{R}^d,
\end{equation*}
and let $M_j$ be \emph{metaplectic quantizations} of the $A_j$'s as defined in \cite[Appendix to chapter 7]{DSj}. Then the rotated diffeomorphism
\begin{equation*}
\tilde{\kappa}:= A_1\circ \kappa \circ A_0^{-1}
\end{equation*}
is such that the projection from the graph of $\tilde{\kappa}$
\begin{equation}\label{tantine}
T^*\mathbb{R}^d\times T^*\mathbb{R}^d \ni (x^1,\xi^1; x^0,\xi^0)\mapsto (x^1,\xi^0)\in \mathbb{R}^d\times \mathbb{R}^d,~~ (x^1,\xi^1)= \tilde{\kappa}(x^0,\xi^0),
\end{equation}
is a diffeomorphism near the origin. It then follows that there exists a unique function $\tilde{\psi}\in C^\infty(\mathbb{R}^d\times \mathbb{R}^d)$ such that for $(x^1,\xi^0)$ near $(0,0)$,
\begin{equation*}
\tilde{\kappa}(\tilde{\psi}'_\xi(x^1,\xi^0),\xi^0)=(x^1,\tilde{\psi}'_x(x^1,\xi^0)),~~ \det \tilde{\psi}''_{x\xi}\neq 0 \text{ and } \tilde{\psi}(0,0)=0.
\end{equation*}
The function $\tilde{\psi}$ is said to \emph{generate} the transformation $\tilde{\kappa}$ near $(0,0)$.

Note that if $\tilde{T}\in I^{comp}(\tilde{\kappa})$, then
\begin{equation}\label{echarpe}
T:= M_1^{-1}\circ \tilde{T}\circ M_0\in I^{comp}(\kappa).
\end{equation}

Thanks to assumption (\ref{tantine}), a Fourier integral operator $\tilde{T}\in I^{comp}(\tilde{\kappa})$ may then be written in the form
\begin{equation}\label{cookie}
\tilde{T} u(x^1):= \frac{1}{(2\pi h)^d}\int \int_{\mathbb{R}^{2n}} e^{i(\tilde{\psi}(x^1,\xi^0)-\langle x^0,\xi^0 \rangle/h} \alpha (x^1,\xi^0;h) u(x^0) dx^0 d\xi^0,
\end{equation}
with $\alpha\in S^{comp}(\mathbb{R}^{2d})$.

Now, let us state a lemma which was proven in \cite[Lemma 4.1]{NZ}, and which describes the effect of a Fourier integral operator of the form (\ref{cookie}) on a Lagrangian distribution which projects on the base manifold without caustics.

\begin{lemme}\label{pacal}
Consider a Lagrangian $\Lambda_0=\{(x_0,\phi_0'(x_0)); x\in \Omega_0\},
\phi_0\in C_b^\infty(\Omega_0)$, contained in a small neighbourhood
$V\subset T^*\mathbb{R}^d$ such that $\kappa$ is generated by $\psi$ near
$V$. We assume that
\begin{equation*}
\kappa(\Lambda_0)=\Lambda_1 = \{(x,\phi_1'(x)); x\in
\Omega_1\},~~\phi_1\in C_b^\infty(\Omega_1).
\end{equation*}
Then, for any symbol $a\in S^{comp}(\Omega_0)$,
 the application of a Fourier integral operator $T$ of the form (\ref{cookie}) to the Lagrangian state
\begin{equation*}
a(x) e^{i\phi_0(x)/h}
\end{equation*}
associated with $\Lambda_0$ can be expanded, for any $L \bel 0$, into
\begin{equation*}
T (a e^{i\phi_0/h})(x) = e^{i\phi_1(x)/h} \Big{(} \sum_{j=0}^{L-1} b_j(x)
h^j+ h^L r_L(x,h) \Big{)},
\end{equation*}
where $b_j\in S^{comp}$, and for any $\ell\in \mathbb{N}$, we have
\begin{equation*}
\begin{aligned}
\|b_j\|_{C^\ell(\Omega_1)}&\leq C_{\ell,j}
\|a\|_{C^{\ell+2j}(\Omega_0)},~~~~0\leq j\leq L-1,\\
\|r_L(\cdot,h)\|_{C^\ell(\Omega_1)}&\leq C_{\ell,L}
\|a\|_{C^{\ell+2L+n}(\Omega_0)}.
\end{aligned}
\end{equation*}
The constants $C_{\ell,j}$ depend only on $\kappa$, $\alpha$ and
$\sup_{\Omega_0} |\partial^\beta \phi_0|$ for $0<|\beta|\leq 2\ell +j$.
\end{lemme}
\subsection{Assumptions on the generalized eigenfunctions}
We consider generalized eigenfunctions of $P_h$ at energy $1$, that is to say, a family of smooth functions $E_h\in C^\infty(X)$ indexed by $h\in(0,1]$ which satisfy
\begin{equation*}
(P_h-1) E_h=0.
\end{equation*} 

We will furthermore assume that these generalized eigenfunctions may be decomposed as follows.

\begin{hyp}\label{pied}
We suppose that $E_h$ can be put in the form
\begin{equation}\label{Poitiers}
E_h= E_h^0+E_h^1,
\end{equation}
where $E_h^0$ is a tempered distribution which is a Lagrangian state associated to a Lagrangian manifold which satisfies Hypothesis \ref{chaise} of invariance, as well as Hypothesis \ref{Happy} of transversality, and where $E_h^1$ is a tempered distribution such that
for each $\rho\in WF_h(E_h^1)$, we have $\rho\in \mathcal{E}$.

Furthermore, we suppose that $E_h^1$ is \emph{outgoing} in the sense that there exists $\epsilon_2\bel 0$ such that for all $\chi,\chi'\in C_c^\infty$ such that $\chi \equiv 1$ on $\{x\in X; b(x)\geq \epsilon_2\}$, there exists $T_\chi\bel 0$ such that for all $t\geq T_\chi$, we have
\begin{equation}\label{outgoing}
\Phi^t(WF\big{(}(1-\chi) \chi' E_h^1\big{)}\big{)} \cap spt (\chi) = \emptyset.
\end{equation}
\end{hyp}

The most natural example of such generalized eigenfunctions is given by distorted plane waves, which we are now going to define. Note that they depend on a parameter $\xi\in \Bound$, so that they actually form a whole family of generalized eigenfunctions.

It is also possible to define generalized eigenfunctions which satisfy Hypothesis \ref{pied} on manifolds which are hyperbolic near infinity. This is done in \cite[Appendix B]{Ing2}; the construction mainly follows \cite[Section 6]{DG}, but some work has to be done to check that $E_h^1$ is a tempered distribution.
\subsection{Distorted plane waves on Euclidean near infinity manifolds} \label{origami}
\begin{definition} We say that $X$ is \emph{Euclidean near infinity} if there exists a
compact set $X_0\subset X$ and a $R_0>0$ such that $X\backslash X_0$ has
finitely many connected components, which we denote by $X_1,...,X_l$, such
that for each $1\leq i \leq l$, $(X_i,g)$ is isometric to
$\big{(}\mathbb{R}^d\backslash B(0,R_0), g_{Eucl} \big{)}$.
\end{definition}

The surface in figure \ref{example} is an example of a Euclidean near
infinity manifold. Note that we may assume that $supp~ V\subset X_0$. Note
also that any Euclidean near infinity manifold fulfils hypothesis
\ref{Guepard}. Indeed, we may take a boundary defining function $b$ such
that $b(x)= (1+|x|^2)^{-1/2}$ if $x\in X_i$ which we identify with
$\mathbb{R}^d\backslash B(0,R_0)$.

To define distorted plane waves, we will simply give a definition of each of the two terms which compose them as in (\ref{Poitiers}).
\subsubsection{Definition of $E_h^0$}\label{country}
By definition of a Euclidean near infinity manifold, we have:
\begin{equation*}X=X_0 \sqcup \big{(}
\bigsqcup_{i=1}^N X_i\big{)},
\end{equation*}
with $X_0$ compact, and for each $1\leq i \leq N$, there exists an
isometric isomorphism
\begin{equation} \label{Torsac}
x_i : X_i \longrightarrow\mathbb{R}^d\backslash B(0,R_0),
\end{equation}
equipped with the Euclidean metric $g_0$.

The boundary of $\overline{X}$ may then be identified with a union of
spheres:
\begin{equation*}\partial X \cong \bigsqcup_{i=1}^N S_i,
\end{equation*}
with $S_i \cong \mathbb{S}^{n}.$

Let $\xi\in \partial \overline{X}$. We have
$\xi \in S_i$ for some $1\leq i \leq m$. Take a
smooth
function $\tilde{\chi} : X \longrightarrow [0;1]$ which vanishes outside of
$X_i$, and which is equal to 1 in a neighbourhood of $S_i$.

We define the incoming wave $E_h^0$ by
$E_h^0(\xi,\cdot) : X \longrightarrow \mathbb{C}$ by:
\begin{equation*}E_h^0(\xi,x) = \tilde{\chi}(x) e^{ \frac{i}{h} x_i(x)\cdot \xi}
\text { if } x\in X_i,~~~~ 0 \text{ otherwise}.
\end{equation*}

If we write $\mathcal{L}_0$ for the Lagrangian submanifold (with
boundaries) $X_i\times
\{\xi\}\subset T^*X$, then $E_h^0$ is a Lagrangian distribution associated to
$\mathcal{L}_0$, which satisfies Hypothesis \ref{chaise} of invariance.

\subsubsection{Definition of the distorted plane waves}\label{defdistorted}
Let us set 
\begin{equation*}
F_h:= -[P_h,\tilde{\chi}]E_h^0(\xi).
\end{equation*}
Note that we have $F_h\in S^{comp}(X)$.

Recall that the outgoing resolvent $R_h(1)$ is defined as
$R_h(1):= \lim\limits_{\epsilon\rightarrow 0^+}
(P_h-(1+i\epsilon)^2)^{-1}$, the limit being taken in the topology of
bounded operators from $L^2_{comp}(X)$ to $L^2_{loc}(X)$.

We shall use the following resolvent estimate, which was proven in
\cite{NZ}.
\begin{theoreme}\label{weasel} [Resolvent estimates for Euclidean near infinity manifolds]
Let $X$ be a Euclidean near infinity manifold such that
Hypothesis \ref{sieste} on hyperbolicity and Hypothesis \ref{Husserl} on topological pressure hold. Then for any $\chi \in
C_c^\infty(X)$, there exists
$C>0$ such that for all $0<h<h_0$, we have
\begin{equation}\label{japon}
\| \chi R_h(1)\chi \|_{L^2(X)\rightarrow L^2(X)} \leq
C\frac{\log(1/h)}{h}
\end{equation}
\end{theoreme}

We define
\begin{equation*}E_h^1 := R_h(1)F_h,
\end{equation*}
which is a tempered distribution thanks to Theorem \ref{weasel}. 

We then define the distorted plane wave as :
\begin{equation*}E_h^{\xi} := E_h^0+E_h^1.
\end{equation*}

To check the outgoing assumption on $E_h^1$, we must explain why that exists $\epsilon_2\bel 0$ such that for all $\chi,\chi'\in C_c^\infty$ such that $\chi \equiv 1$ on $\{x\in X; b(x)\geq \epsilon_2\}$, there exists $T_{\chi}\bel 0$ such that for all $t\bel T_{\chi}$, we have
\begin{equation}\label{sautemouton}
\Phi^t(WF\big{(}(1-\chi) \chi' E_h^1\big{)} \big{)} \cap spt (\chi)= \emptyset.
\end{equation}

From \cite[\S 6.2]{DG}, we know that for any $\rho\in WF_h(E_h^1)$, we have $\rho\in \mathcal{E}$, and either $\rho\in \Gamma^+$ or there exists a $t\bel 0$ such that $\Phi^{-t}(\rho) = (x,\xi)$ where $x\in spt (\partial \tilde{\chi})$, where $\tilde{\chi}$ is as in Section \ref{country}.

We may take $\epsilon_2 < \epsilon_0$ small enough so that $spt(\tilde{\chi})\subset \{x\in X, b(x) \bel \epsilon_2 \}$. Suppose that $\rho=(x,\xi)$ is such that $x\in spt(1-\chi) $ and $\pi_X(\Phi^{t}(\rho))\in spt(\chi)$. Then, by geodesic convexity, $(x,-\xi)\in \mathcal{DE}_+$. Therefore, since $spt(\tilde{\chi})\subset \{x\in X, b(x) \bel \epsilon_2 \}$ and $spt(1-\chi)\subset\{x\in X, b(x)< \epsilon_2\}$ and since $b$ decreases in the future along the trajectory of $(x,-\xi)$, it is impossible that there exists $t\bel 0$ such that $\Phi^{-t}(\rho) = (x,\xi)$ where $x\in spt (\partial \tilde{\chi})$. Therefore, if $\rho\in \Phi^t(WF\big{(}(1-\chi) \chi' E_h^1\big{)}\big{)} \cap spt (\chi)$, we must have $\rho \in \mathcal{DE}_+$.

On the other hand, if $\rho\in \mathcal{DE}_+$, then (\ref{sautemouton}) is always satisfied as long as $T_\chi$ is large enough so that $\Phi^{T_\chi}\big{(}\mathcal{DE}_+\cap T^*(spt(1-\chi))\big{)}\cap T^*spt(\chi)=\emptyset$. This shows that $E_h^1$ is outgoing.

Finally, one readily checks that we have, in the sense of PDEs:
\begin{equation*}(P_h-1) E^{\xi}_h=0.
\end{equation*}
We will sometimes simply write $E_h$ instead of $E_h^{\xi}$, to
avoid cumbersome notations.

The definition of $E_h$ seems to depend on the choices of the cut-off
functions we
made. Actually, the distorted plane waves can be defined in a much more
intrinsic fashion, using the structure of resolvent at infinity. We don't want to enter into the details here (see \cite[Section 6]{DG}, and
the references therein, or \cite[Chapter 2]{Mel}).

\subsection{Topological pressure} \label{press}
We shall now give a definition of topological pressure, so as to formulate Hypothesis \ref{Husserl}.
Recall that the distance $d$ was defined in section \ref{averse}, and that it was associated to the adapted metric. We say
that a set $\mathcal{S}\subset K$ is $(\epsilon,t)$-separated if for
$\rho_1, \rho_2\in \mathcal{S}$, $\rho_1 \neq \rho_2$, we have
$d(\Phi^{t'}(\rho_1),\Phi^{t'}(\rho_2))>\epsilon$ for some $0\leq t \leq
t'$. (Such a set is necessarily finite.)

The metric $g_{ad}$ induces a volume form $\Omega$ on any $d$-dimensional
subspace of $T(T^*\mathbb{R}^d)$. Using this volume form, we will define
the unstable Jacobian on $K$. For any $\rho\in K$, the determinant map
\begin{equation*}\Lambda^n d\Phi^t(\rho)|_{E_\rho^{+0}} : \Lambda^n E_\rho^{+0}
\longrightarrow \Lambda^n E_{\Phi^t(\rho)}^{+0}
\end{equation*}
can be identified with the real number
\begin{equation*}\det\big{(} d\Phi^t(\rho)|_{E_\rho^{+0}}\big{)} :=
\frac{\Omega_{\Phi^t(\rho)}\big{(}d\Phi^tv_1 \wedge
d\Phi^tv_2\wedge...\wedge d\Phi^tv_n\big{)}}{\Omega_\rho(v_1\wedge v_2
\wedge... \wedge v_n)},
\end{equation*}
where $(v_1,...,v_n)$ can be any basis of $E_\rho^{+0}$. This number
defines the unstable Jacobian:
\begin{equation}\label{defJaco}
\exp \lambda^+_t(\rho) := \det\big{(}
d\Phi^t(\rho)|_{E_\rho^{+0}}\big{)}.
\end{equation}
From there, we take
\begin{equation*}Z_t(\epsilon,s):= \sup \limits_{\mathcal{S}} \sum_{\rho \in \mathcal{S}}
\exp(-s\lambda_t^+(\rho)),
\end{equation*}
where the supremum is taken over all $(\epsilon,t)$-separated sets. The
pressure is then defined as
\begin{equation*}\mathcal{P}(s):= \lim \limits_{\epsilon \rightarrow 0} \limsup \limits_{t
\rightarrow \infty} \frac{1}{t} \log  Z_t(\epsilon,s) .
\end{equation*}
This quantity is actually independent of the volume form $\Omega$ and of the metric chosen: after
taking logarithms, a change in $\Omega$ or in the metric will produce a term $O(1)/t$,
which is not relevant in the $t\rightarrow \infty$ limit.
\begin{hyp} \label{Husserl}
We assume the following inequality on the topological pressure  associated with $\Phi^t$ on $K$:
\begin{equation} \label{Laurence}
\mathcal{P}(1/2)<0.
\end{equation}
\end{hyp}
We will give an equivalent definition of topological pressure in section
\ref{Benoit}, better suited to our purpose.

\subsection{Statement of the results concerning distorted plane waves}\label{statementMain}
We may now formulate our main result.

\begin{theoreme}\label{ibrahim}
 Suppose that the manifold $X$ satisfies Hypothesis \ref{Guepard} at infinity, and that the Hamiltonian flow $(\Phi^t)$ satisfies
Hypothesis \ref{sieste} on hyperbolicity and Hypothesis \ref{Husserl} concerning the topological pressure. Let $E_h$ be a generalized eigenfunction of the form described in Hypothesis \ref{pied}, where $E_h^0$ is associated to a Lagrangian manifold $\Lag$ which satisfies
the invariance Hypothesis \ref{chaise} as well as the transversality
Hypothesis \ref{Happy}.
 
Then there exists a finite set of points $(\rho_b)_{b\in B_1}\subset K$ and a family $(\Pi_b)_{b\in B_1}$ of operators in $\Psi^{comp}_h(X)$ microsupported in a small neighbourhood of $\rho_b$ such that $\sum_{b\in B_1} \Pi_b = I$ microlocally on a neighbourhood of $K$ in $T^*X$ such that the following holds.

Let $\mathcal{U}_b: L^2(X)\longrightarrow L^2(\mathbb{R}^d)$ be a Fourier integral operator quantizing the symplectic
change of local coordinates $\kappa_b:(x,\xi) \mapsto
(y^{\rho_b},\eta^{\rho_b})$, and which is microlocally unitary on the microsupport of $\Pi_b$.

For any $r>0$, there exists $M_{r}>0$ such that
 we have
\begin{equation}\label{greenday}\mathcal{U}_b \Pi_b E_h(y^{\rho_b}) = \sum_{n=0}^{\lfloor
M_{r,\ell} |\log h|\rfloor}
\sum_{\beta\in \tilde{\mathcal{B}}_n} e^{i \phi_{n,\beta,b}(y^{\rho_b})/h}
a_{n,\beta,b}(y^{\rho_b};h) + R_r,
\end{equation}
where $a_{n,\beta,b}\in S^{comp}(\mathbb{R}^d)$ are classical symbols, and each $\phi_{n,\beta,b}$ is a smooth function independent of $h$, and defined in a neighbourhood of the support of
$a_{n,\beta,b}$. The set $\tilde{\mathcal{B}}_n$ will be defined in (\ref{BN}) . Its cardinal behaves like some exponential of $n$.

We have the following estimate on the remainder
 \begin{equation*}\|R_r\|_{L^2}=O(h^r).
\end{equation*}

 For any $\ell\in \mathbb{N}$, $\epsilon>0$, there exists $C_{\ell,\epsilon}$
such that for all $n\geq 0$, for all $h\in (0,h_0]$, we have
\begin{equation}\label{sheriff}
\sum_{\beta\in \tilde{\mathcal{B}}_n} \|a_{n,\beta,b}\|_{C^\ell} \leq C_{\ell,\epsilon}
e^{n(\mathcal{P}(1/2)+\epsilon)}.
\end{equation}
\end{theoreme}

\begin{remarque}
This theorem can be considered as a quantum analogue of Theorem \ref{Cyril}. Indeed, as we explained in section \ref{boulgakov}, we will prove it by describing the evolution of the Schrödinger flow of Lagrangian states, while Theorem \ref{Cyril} described the evolution by the Hamiltonian flow of associated Lagrangian manifolds. Actually, the sets containing the microsupports of the operators $(\Pi_b)_{b\in B_1}$ will be built from the sets $(W_a)_{a\in A_1}$ constructed in Theorem \ref{Cyril}, as explained in section \ref{Benoit}.
\end{remarque}

\begin{remarque}
The remainder $R_r$ is compactly microlocalised, since the other two terms in the decomposition (\ref{greenday}) are compactly microlocalised. Therefore, for any $\ell \in \mathbb{N}$, by possibly taking $M_r$ larger, we may ask that
$$\|R_r\|_{C^\ell}= O(h^r).$$
\end{remarque}

Theorem \ref{ibrahim} may be used to identify the semiclassical measures associated to our generalized eigenfunctions, as in Theorem \ref{blacksabbath20}. We shall do this only microlocally close to the trapped set, since the expression for the semiclassical measure on the whole manifold may become very complicated. 
 
Let us denote by $\pi_b$ the principal symbol of the operators $\Pi_b$ introduced in the statement of Theorem \ref{ibrahim}. The following corollary is a more precise version of  (the second part of) Theorem \ref{blacksabbath20}.

 \begin{corolaire} \label{blacksabbath}
There exists a constant $0< c \leq 1$ and functions $e_{n,\beta,b}$ for $n\in \mathbb{N}$, $\beta\in \tilde{\mathcal{B}}_n$ and $b\in B_1$ such that for any $a\in C_c^\infty(T^*X)$ and for any $\chi\in C_c^\infty(X)$, we have
\begin{equation*}\langle Op_h(\pi_b^2 a) \chi E_h,  \chi E_h\rangle = \int_{T^*X} a(x,v)
\mathrm{d}\mu_{b,\chi}(x,v) +
O(h^c),
\end{equation*}
with \begin{equation*}\mathrm{d}\mu_{b,\chi}(\kappa_b^{-1}(y^{\rho_b},\eta^{\rho_b})) = \sum_{n=0}^{\infty} \sum_{\beta\in
\tilde{\mathcal{B}}_n}  e_{n,\beta,b}(y^{\rho_b}) \delta_{\{\eta^{\rho_b}=\partial
\phi_{j,n}(y^{\rho_b})\}} d y^{\rho_b},
\end{equation*}
The functions $e_{n,\beta,b}$ satisfy an exponential decay estimate as in (\ref{sheriff}).
\end{corolaire}
The functions $e_{n,\beta,b}$ will be closely related to $a^0_{n,\beta,b}(y^{\rho_b})$, the principal symbol of $a_{n,\beta,b}(y^{\rho_b})$ appearing in (\ref{greenday}). Actually, $e_{n,\beta,b}(y^{\rho_b})$ will either be the square of the modulus of $a^0_{n,\beta,b}(y^{\rho_b})$, or the square of the modulus of the sum of a finite number of $a^0_{n,\beta,b}(y^{\rho_b})$, for different values of $n$ and $\beta$. These different terms will come from the fact that a point may belong to $\Phi_\beta^{n,t_0}(\Lag)$ for several values of $n,\beta$.

\subsection{Strategy of proof}
To study the asymptotic behaviour of the distorted plane wave as $h$ goes
to zero, we would like to write that $\tilde{U}(t)E_h=E_h$, where
$\tilde{U}(t):=e^{it/h}U(t)$. However, this
equation can only be formal, because $E_h\notin L^2(X)$. Instead, we use
the following lemma from \cite{DG} (Lemma 3.10):
\begin{lemme}\label{voyante}
Let $\chi\in C_c^\infty (X)$. Take $t \in \mathbb{R}$, and a cut-off
function $\chi_t\in C_c^\infty(X)$ be supported in the interior of a
compact set $K_t$, such that
\begin{equation*}d_g(supp \chi, supp (1-\chi_t))>2 |t|,
\end{equation*}
where $d_g$ denotes the Riemannian distance on $M$.
Then, for any $\xi\in
\mathbb{S}^d$, we have
\begin{equation}\label{fmjh}
\chi E_h =\chi \tilde{U}(t) \chi_t
E_h + O(h^\infty
\|E_h\|_{L^2(K_t)}).
\end{equation}
\end{lemme}

Since $E_h$ is a tempered distribution by assumption, we have, for any $t>0$ and $\chi\in C_c^\infty(X)$:
\begin{equation*}\|\chi E_h- \chi \tilde{U}(t) \chi_{t}
E_h\|_{L^2} = O(h^\infty),
\end{equation*}
where $\chi_t$ is as in Lemma \ref{voyante}.

We may then iterate this equation as follows: we write that $\chi_{t}=\chi
+ \chi_{t}(1-\chi)$, and obtain
\begin{equation*}\chi E_h= \chi \tilde{U}(t) \big{(} (1-\chi)\chi_{t}\big{)}E_h + \chi \tilde{U}(t)\chi
\tilde{U}(t)\chi_{t} E_h+O(h^\infty).
\end{equation*}
We may iterate this method to times $Nt\leq Mt |\log h|$ for any given
$M>0$. We obtain
\begin{equation}\label{blues}
\chi E_h =(\chi \tilde{U}(t))^N \chi_t E_h+ \sum_{k=1}^{N} (\chi \tilde{U}(t))^k
(1-\chi)\chi_{t} E_h +O(h^\infty).
\end{equation}
Now, choose $\chi\in C_c^\infty(X)$ as in Hypothesis \ref{pied}, and take $t\bel T_\chi$.

\begin{lemme}\label{flightoftherat}
Let $t>T_\chi$, $M>0$, and $\chi\in C_c^\infty(X)$ be such that $\chi\equiv 1$ on $\{x\in X, b(b)\bel \epsilon_2\}$, where $\epsilon_2<\epsilon_0$ is as in Hypothesis \ref{pied}. For any $k\leq M|\log h|$, we
have \begin{equation*}\|(\chi \tilde{U}(t))^k (1-\chi)\chi_{t}
E_h^1\|_{L^2} = O(h^\infty).
\end{equation*}
\end{lemme}
\begin{proof}
We only have to prove that $\|(\chi \tilde{U}(t)) (1-\chi)\chi_{t}
E_h^1\|_{L^2} = O(h^\infty)$. This is a consequence of (\ref{outgoing}) in Hypothesis \ref{pied}.
\end{proof}

Therefore, we have for any $\chi \in C_c^\infty(X)$ as in Lemma \ref{flightoftherat}:
\begin{equation}\label{ripoux}
\chi E_h =(\chi \tilde{U}(t))^N \chi_t
E^0_h+(\chi \tilde{U}(t))^N \chi_t E^1_h+ \sum_{k=1}^{N} (\chi \tilde{U}(t))^k
(1-\chi)\chi_{t} E^0_h +O(h^\infty).
\end{equation}
Let us now introduce tools from \cite{NZ} to analyse these terms in more details.
\section{Tools for the proofs of Theorem \ref{ibrahim}}\label{tools}
\subsection{Another definition of Topological Pressure}\label{Benoit}

Recall that $\mathcal{E}_\mathbf{E}$ and $K_\mathbf{E}$ were defined in (\ref{deflayer}) and (\ref{deftrap}) respectively.
For any $\delta\bel 0$ small enough so that (\ref{struc}) holds, we define
\begin{equation*}
\mathcal{E}^\delta := \bigcup_{|\mathbf{E}-1|< \delta} \mathcal{E}_{\mathbf{E}}, ~~ 
K^\delta := \bigcup_{|\mathbf{E}-1|< \delta} K_{\mathbf{E}}.
\end{equation*}

Let $\mathcal{W}=(W_a)_{a\in A_1}$ be a finite open cover of
$K^{\delta/2}$, such that the $W_a$ are all strictly included in
$\mathcal{E}^\delta$ and of diameter $<\varepsilon_0$, where $\varepsilon_0$
comes from Theorem \ref{Cyril}.
For any $T\in \mathbb{N}^*$, define $W(T):=(W_\alpha)_{\alpha \in A_1^T}$ by
\begin{equation*}W_\alpha := \bigcap_{k=0}^{T-1} \Phi^{-k}(W_{a_k}),
\end{equation*} where
$\alpha=a_0,..,a_{T-1}$.
Let $\mathcal{A}'_T$ be the set of $\alpha\in A_1^T$ such that $W_\alpha\cap
K^\delta \neq \emptyset$.
If $V\subset \mathcal{E}^ \delta$, $V\cap K^{\delta/2} \neq \emptyset$,
define
\begin{equation*}S_T(V) := -\inf_{\rho\in V\cap K^{ \delta/2}} \lambda^+_T(\rho),~~~~~\text{  with } \lambda_T^+\text{ as in (\ref{defJaco})}.
\end{equation*}
\begin{equation*}Z_T(\mathcal{W},s):=\inf \big{\{ } \sum_{\alpha\in \mathcal{A}_T} \exp
\{s S_T (W_\alpha)\} : \mathcal{A}_T\subset \mathcal{A'}_T, K^{\delta/2}
\subset \bigcup_{\alpha \in \mathcal{A}_T} W_\alpha \big{\} }
\end{equation*}
\begin{equation*} \mathcal{P}^\delta (s):= \lim \limits_{diam \mathcal{W} \rightarrow
0} \lim \limits_{T \rightarrow \infty} \frac{1}{T} \log
Z_T(\mathcal{W},s).
\end{equation*}
The topological pressure is then:
\begin{equation} \label{Franck}
\mathcal{P} (s) = \lim \limits_{\delta \rightarrow 0} \mathcal{P}^
\delta (s).
\end{equation}
Recall that we assumed that
\begin{equation*}\mathcal{P}(1/2)<0.
\end{equation*}
Let us fix $\epsilon_0>0$ so that $\mathcal{P}(1/2)+2\epsilon_0 < 0$. Then there exists
$t_0>0$, and $\hat{\mathcal{W}}$ an open cover of $K^\delta$
with $diam (\hat{\mathcal{W}})<\secur$ such that
\begin{equation}\label{Tea}
\Big{|} \frac{1}{t_0} \log Z_{t_0} (\hat{\mathcal{W}},s) -
\mathcal{P}^\delta (s) \Big{|} \leq \epsilon_0.
\end{equation}
We can find $\mathcal{A}_{t_0}$ so that $\{W_\alpha, \alpha \in
\mathcal{A}_{t_0} \}$ is an open cover of $K^\delta$ in
$\mathcal{E}^\delta $ and such that
\begin{equation*} \sum_{\alpha\in \mathcal{A}_{t_0}} \exp \{s S_{t_0} (W_\alpha) \} \leq
\exp \{ t_0 (\mathcal{P}^\delta (s) + \epsilon_0) \} .
\end{equation*}
Therefore, if we take $\delta$ small enough, and if we rename $\{W_\alpha,
\alpha \in \mathcal{A}_{t_0} \}$ as $\{V_b, b \in B_1\}$, we have:
\begin{equation}\label{rainy} \sum_{b\in B_1} \exp \{\frac{1}{2} S_{t_0} (V_b) \} \leq \exp \{ t_0
(\mathcal{P}(1/2) + 2\epsilon_0) \} .
\end{equation}
By taking $t_0$ large enough, we can assume
that $\log (1+\epsilon_0)+t_0
(\mathcal{P}(1/2)+\epsilon_0)<0$.

\subsubsection*{A new open cover of $\mathcal{E}$}

By hypothesis, the diameter of $\hat{\mathcal{W}}$ in (\ref{Tea}) is smaller
than $\varepsilon_0$, so that we may apply Theorem \ref{Cyril} to it. We
complete it into an open cover $(W_a)_{a\in A}$ as in Theorem \ref{Cyril},
and if $\alpha\in A^N$ for some $N\geq 0$, we define as previously
$W_\alpha := \bigcap_{k=0}^{N-1} \Phi^{-k}(W_{a_k})$.

Let us rewrite as $(V_b)_{b\in B_2}$ the sets $(W_\alpha)_{\alpha\in
A^{t_0}}$ where $\alpha \in A^{t_0} \backslash \mathcal{A}_{t_0}$  such that
$\alpha_k\neq 0$ for some $0\leq k \leq t_0-1$. We will also write $V_0$ for
the set $W_{0,0,...,0}$.

If we write $B=B_1\sqcup B_2 \sqcup \{0\}$, the sets $(V_b)_{b\in B}$ form an
open cover of $\mathcal{E}$ in $T^*X$.

Actually, by compactness of the interaction region, we may find a $\delta\bel \delta'\bel 0$ small enough so that (\ref{struc}) holds and such that, by replacing $V_0$ by $V_0\cap \mathcal{E}^\delta$, $(V_b)_{b\in B}$ forms an open cover of $\mathcal{E}^{\delta'}$ included in $\mathcal{E}^\delta$.

If $\beta=b_0...b_{N-1}\in B^{N}$ for some $N\in \mathbb{N}$, and if
$\Lambda$ is a Lagrangian manifold, we will define for each $0\leq k \leq N-1$ the set $\Phi^{k,t_0}_\beta(\Lambda)$ by

\begin{equation*}
\begin{aligned}
\Phi^{0,t_0}_\beta(\Lambda)&= \Lambda \cap V_{b_0}\\
\Phi^{k,t_0}_\beta (\Lambda) &=\Phi^{t_0}\big{(} V_{b_{k}}\cap \Phi^{k-1,t_0}_{\beta}(\Lambda)\big{)}~~ \text{ for } 1\leq k \leq N-1.
\end{aligned}
\end{equation*}

By definition of the sets $b\in B$, we have $\Phi^{N,t_0}_\beta(\Lambda) =
\Phi^{Nt_0}_{\alpha_\beta}(\Lambda)$, where $\alpha_\beta\in A^{Nt_0}$ is the
concatenation of all the sequences which make up the $b_k$, $0\leq k\leq
N-1$.

Therefore, once we have fixed a point $\rho^b\in K\cap V_b$ for each $b\in B_1$, we have the following analogous of Theorem \ref{Cyril}.
\begin{corolaire}\label{marianne}
There exists $\tins'\in \mathbb{N}$ such that for all
$N\in
\mathbb{N}$, for all $\beta\in B^{N}$ and all $b\in B_1$, then
$V_b\cap\Phi_\beta^{N,t_0}(\Lag)$ is either empty, or is a Lagrangian manifold
in some unstable
cone in the coordinates $(y^{\rho_b},\eta^{\rho_b})$.

Furthermore, if $N-\tau(\beta)\geq \tins'$, then
$V_b\cap\Phi_\beta^{N}(\Lag)$ is a $\ins$-unstable Lagrangian manifold
in the coordinates $(y^{\rho_b},\eta^{\rho_b})$.
\end{corolaire}

\begin{remarque}[New definition of the sets $(V_b)_{b\in B_1}$]\label{animals}
The sets $(V_b)_{b\in B_1}$ form an open cover of $K$. By compactness, they form an open cover of $\{\rho\in \mathcal{E}; d(\rho,K)\leq \epsilon_3\}$ for some $\epsilon_3 \bel 0$. Hence, if for each $b\in B_2$ we replace each $V_b$ by $V_b\cap \{\rho\in \mathcal{E}; d(\rho, K) \bel \epsilon_3/ 2\}$ which we still denote by $V_b$, the sets $(V_b)_{b\in B}$ still form an open cover of $\mathcal{E}$, and the conclusions of Corollary \ref{marianne} do still apply. 

By adapting the proof of Lemma \ref{abeille}, we see that by possibly enlarging $\tins'$, we may suppose that for all $b\in B_2$, for all $\rho\in V_b$, we have $\Phi^{\tins' t_0}(\rho)\in V_0\backslash \big{(}\bigcup_{b\in B_2} V_b\big{)}$ or
$\Phi^{-\tins' t_0}(\rho)\in V_0\backslash \big{(} \bigcup_{b\in B_2} V_b\big{)}$.

Note also that thanks to Lemma \ref{saladin}, we have for any $b\in B_1\cup B_2$, for any $k\geq 1$,
$\Phi^{k t_0}(\Lag\cap V_b))\cap \ext \cap \inco=\emptyset.
$
\end{remarque}

\begin{remarque}\label{napoleon}
In \cite{NZ}, Proposition 5.2, the authors proved the following statement.
There exists a $\gamma_1>0$ such that the following holds. Let $b,b'\in
B_1$, and let $\Lambda$ be a Lagrangian manifold contained in $V_b$,
$\gamma$-unstable in the coordinates $(y^{\rho_b},\eta^{\rho_b})$ for some $\gamma\leq
\gamma_1$. Then $\Phi^ {t_0}(\Lambda)\cap V_{b'}$ is also a Lagrangian
manifold which is $\gamma$-unstable in the coordinates
$(y^{\rho_{b'}},\eta^{\rho_{b'}})$.

Furthermore, the map $y^{\rho_b}\mapsto y^{\rho_{b'}}$ obtained by projecting
$\Phi^{t_0}|_{\Lambda}$ onto the planes $\{(y^{\rho_b},\eta^{\rho_b}); \eta^{\rho_b}=0\}$ and
$\{(y^{\rho_{b'}},\eta^{\rho_{b'}}); \eta^{\rho_{b'}}=0\}$ satisfies the following estimate on
its domain of definition:
\begin{equation*}
\det \Big{(}\frac{\partial y^{\rho_{b'}}}{\partial y^{\rho_b}}\Big{)} = (1+O
(\epsilon^\mathsf{p})) e^{\lambda_{t_0}^+(\rho_b)},
\end{equation*}
where $\lambda_{t_0}^+(\rho_b)$ is the unstable Jacobian of $\rho_b$,
defined in (\ref{defJaco}).

In the sequel, we will always suppose that $\ins< \gamma_1$.
\end{remarque}

For each $b\in B_1$, we will denote by $\mathcal{U}_b$ a Fourier integral operator quantizing the local change of symplectic coordinates $(x,\xi)\mapsto (y^{\rho_b},\eta^{\rho_b})$.
\subsection{The Schrödinger propagator as a Fourier integral operator}\label{birette}
Let us explain how the formalism of section \ref{chaussette} may be used to describe the Schrödinger propagator $U(t)$ acting on $L^2(X)$. We shall state a lemma proven in \cite[Lemma 4.2]{NZ}. Recall that for $0<\delta< 1$, we defined $\mathcal{E}^\delta$ as $\bigcup_{|\mathbf{E}-1|<\delta} \mathcal{E}_{\mathbf{E}}$.
\begin{lemme}\label{diderot}
Let $V_0\Subset \mathcal{E}^\delta$, $V_1\subset \Phi^t(V_0)$ for some $t\bel 0$. Take some $\rho_0\in V_0\cap \mathcal{E}$ and set $\rho_1=\Phi^t(\rho_0)\in V_1$. Let $f_j: \pi(V_j)\rightarrow \mathbb{R}^d$, $j=0,1$ be local coordinates such that $f_0(\pi(\rho_0))=f_1(\pi(\rho_1))=0\in \mathbb{R}^d$. They induce on $V_0$ and $V_1$ the symplectic coordinates
\begin{equation*}
F_j(x,\xi):= \big{(}f_j(x), (df_j(x)^t)^{-1}\xi-\xi^{(j)}\big{)},~~ j=0,1,
\end{equation*}
where $\xi^{(j)}\in \mathbb{R}^d$ is fixed by the condition $F_j(\rho_j)=(0,0)$. Then the operator on $L^2(\mathbb{R}^d)$,
\begin{equation*}
T(t):= e^{-i\langle x, \xi^{(1)} \rangle /h}  (f_1^{-1})^* U(t) (f_0)^*e^{i\langle x,\xi^{(0)} \rangle/h}
\end{equation*}
is of the form (\ref{echarpe}) for some choice of the $A_j$'s microlocally near $(0,0)\times(0,0)$.
\end{lemme}
\subsection{Iterations of Fourier integral operators}\label{bismark}
We recall here the main results from \cite[\S 4]{NZ} concerning the
iterations of semiclassical Fourier integral operators in $T^*\mathbb{R}^d$.

Let $V\subset T^* \mathbb{R}^d$ be an open neighbourhood of $0$, and take
a sequence of symplectomorphisms $(\kappa_i)_{i=1,...,N}$ from $V$ to
$T^*\mathbb{R}^d$, such that $\forall i \in \{1,...,N\}$, we have
$\kappa_i(0)\in V$, and the following projection:
\begin{equation*}(x_1,\xi_1 ; x_0, \xi_0) \mapsto (x_1,\xi_0) \text{    where    }
(x_1,\xi_1)= \kappa (x_0, \xi_0 ) 
\end{equation*}
is a diffeomorphism close to the origin.
We consider Fourier integral operators $(T_i)$ which quantise $\kappa_i$
and which are
microlocally unitary near an open set $U\times U$, where $U \Subset V$ which
contains the origin. Let $\Omega\subset \mathbb{R}^d$ be an open set such
that $U\Subset T^*\Omega$, and, for all $i$, $\kappa_i(U) \Subset T^*\Omega$.
For each $i$, we take a smooth cut-off function $\chi_i\in C_c^\infty (U ;
[0,1] )$, and let
\begin{equation} \label{blabla}
S_i := Op_h(\chi_i) \circ T_i
\end{equation}
Let us consider a family of Lagrangian manifolds $\Lambda_k = \{
(x,\phi_k'(x)) ; x\in \Omega\} \subset T^*\mathbb{R}^d, ~~ k=0,...,N$ such
that:
\begin{equation} \label{leffe}
|\partial^\alpha \phi_k|\leq C_\alpha,~~~~0\leq k\leq N~~ \alpha\in
\mathbb{N}^d.
\end{equation}
We assume that there exists a sequence of integers $(i_k\in
\{1,...,J\})_{k=1,...,N}$ such that
\begin{equation*}\kappa_{i_{k+1}}(\Lambda_k\cap U) \subset \Lambda_{k+1},~~k=0,...,N-1.
\end{equation*}

We define $g_k$ by
\begin{equation*}g_k(x) = \pi\circ \kappa_{i_k}^{-1}(x,\phi_k'(x)).
\end{equation*}
That is to say,
$\kappa_{i_k}^{-1}(x,\phi_k'(x))=(g_k(x),\phi_{k-1}'(g_k(x)))$.

We will say that a point $x\in \Omega$ is $N$-admissible if we can define
recursively a sequence by $x^N=x$, and, for $k=N,...,1$,
$x^{k-1}=g_k(x^k)$. This procedure is possible if, for any $k$, $x^k$ is
in the domain of definition of $g_k$.

Let us assume that, for any admissible sequence $(x^N...x^0)$, the
Jacobian matrices are uniformly bounded from above:
\begin{equation*}\Big{\|}\frac{\partial x^k}{\partial x^l}\Big{\|}=\Big{\|}\frac{\partial
(g_{k+1}\circ g_{k+2}\circ ...\circ g_l)}{\partial x^l}(x^l)\Big{\|}\leq
C_D,~~~0\leq k<l\leq N,
\end{equation*}
where $C_D$ is independent on $N$. This assumption roughly says that the
maps $g_k$ are (weakly) contracting.

We will also use the notation
\begin{equation*}D_k:= \sup\limits_{x\in \Omega} |\det
dg_k(x)|^{1/2},~~~J_k:=\prod_{k'=1}^k D_{k'}, 
\end{equation*}
and assume that the $D_k$'s are uniformly bounded: $1/C_D\leq D_k\leq C_D$.

The following result can be found in \cite[Proposition 4.1]{NZ}.

\begin{proposition} \label{haroun}
We use the above definitions and assumptions, and take $N$ arbitrarily large,
possibly varying with $h$.
Take any $a\in S^{comp}$ and consider the Lagrangian state $u=a
e^{i\phi_0/h}$ associated with the Lagrangian $\Lambda_0$. Then we may
write:
\begin{equation*} (S_{i_N}\circ...\circ S_{i_1}) (a e^{i\phi_0/h})(x) = e^{i\phi_N(x)/h}
\big{(} \sum_{j=0}^{L-1} h^j a_j^N(x) + h^L R_L^N(x,h)\big{)},
\end{equation*}
where each $a_j^N\in C_c^\infty(\Omega)$ depends on $h$ only through $N$, and
$R_L^N\in C^\infty ((0,1]_h,\mathcal{S}(\mathbb{R}^d))$.
If $x^N\in \Omega$ is $N$-admissible, and defines a sequence $(x^k),
k=N,...,1$, then
\begin{equation*}|a_0^N(x^N)| = \Big{(}\prod_{k=1}^N \chi_{i_k} (x^k,\phi_k' (x^k))|\det
dg_k(x^k)|^{\frac{1}{2}}\Big{)} |a(x^0)|,
\end{equation*}
otherwise $a^N_j(x^N)=0, ~~ j=0,...,L-1$.
We also have the bounds
\begin{equation}\label{republique}
\|a_j^N\|_{C^\ell(\Omega)}\leq C_{j,\ell} J_N(N+1)^{\ell+3j}
\|a\|_{C^{\ell+2j}(\Omega)},~~~j=0,...,L-1, \ell\in \mathbb{N},
\end{equation}
\begin{equation}\label{coligny}\|R_L^N\|_{L^2(\mathbb{R}^d)}\leq C_L
\|a\|_{C^{2L+d}(\Omega)}(1+C_0h)^N\sum_{k=1}^N J_k k^{3L+d},
\end{equation}
\begin{equation}\label{antigna}
\|R_L^N\|_{C^\ell(\mathbb{R}^d)}\leq C_{L,l} h^{-d/2-\ell}
\|a\|_{C^{2L+d}(\Omega)}(1+C_0h)^N\sum_{k=1}^N J_k k^{3L+d}.
\end{equation}

The constants $C_{j,\ell},C_0$ and $C_L$ depend on the constants in
(\ref{leffe}) and on the operators $\{S_j\}_{j=1}^J$.
\end{proposition}

We shall mainly be using this proposition in the case where for all $k$, we have $D_k\leq \nu < 1$. In this case, the estimates (\ref{republique}), (\ref{coligny}) and (\ref{antigna}) imply that for any $\ell\in \mathbb{N}$, there exists $C_{\ell}$ independent of $N$ such that for any $N\in \mathbb{N}$, we have
\begin{equation}\label{highway78}
\|a^N\|_{C^\ell}\leq \|a^N_0\|_{C^\ell} \big{(} 1+ C_\ell h \big{)}.
\end{equation}

\subsection{Microlocal partition}\label{partition}
We take a partition of unity $\sum_{b\in B} \pi_b$ such that :
\begin{equation*}\sum_{b\in B} \pi_b(x)\equiv 1 \text { for all } x\in
\mathcal{E}^{\delta'},
\end{equation*} and $supp (\pi_b) \subset V_b \subset
\mathcal{E}^\delta$ for all $b\in B$.

For $b\in B_1\cup B_2$, we set $\Pi_b:= Op_h(\pi_b)$. We have
\begin{equation*}WF_h(\Pi_b) \subset V_b\cap \mathcal{E}^\delta,~~ \text{ and } \Pi_b=\Pi_b^*.
\end{equation*}

We then set 
\begin{equation*}
\Pi_0:= Id- \sum_{b\in B_1\cup B_2} \Pi_b.
\end{equation*}

We can decompose the propagator at time $t_0$ into:
\begin{equation*}\tilde{U}(t_0)= \sum_{b\in B} \tilde{U}_b, \text{      where
 } \tilde{U}_b:=  \Pi_b e^{it_0/h} U(t_0).
\end{equation*}
The propagator at time $N t_0$ may then be decomposed as follows:
\begin{equation}\label{triceratops}\tilde{U}(Nt_0)= \sum_{\beta\in B^{N}}
\tilde{U}_{\beta},
\end{equation}
where $\tilde{U}_\beta := \tilde{U}_{\beta_{N-1}}\circ...\circ
\tilde{U}_{\beta_0}$. 

\subsection{Hyperbolic dispersion estimates}
We will use the following hyperbolic dispersion estimate, coming from \cite[Proposition 6.3]{NZ}. the proof of which can be found in
\cite[Section 7]{NZ}.

\begin{lemme} [Hyperbolic dispersion estimate]\label{grenadine}
Let $M>0$ be fixed. There exists a $h_0>0$ and a $C>0$ such that for any
$0<h<h_0$, for any $N<M \log(1/h)$, for any $\beta \in
B_1^{N}$, we have
\begin{equation} \label{louisxvi}
\|\tilde{U}_{\beta}\|_{L^2\rightarrow L^2}\leq C h^{-d/2}
(1+\epsilon_0)^{N}
\prod_{j=1}^{N} \exp \big{[} \frac{1}{2} S_{t_0} (V_{\beta_j})\big{]}.
\end{equation}
\end{lemme}

\section{Proof of Theorem \ref{ibrahim}}\label{preuve}
\begin{proof}
Having introduced these different tools, we may now come back to the proof of Theorem \ref{ibrahim}.
\subsection{Decomposition of $\chi E_h$}
Let $\chi\in C_c^\infty(X)$ be as in Lemma \ref{flightoftherat}. We may suppose $T_\chi\leq t_0$. Then, by equation (\ref{ripoux}), we have:
\begin{equation}\label{canardsauvage}
\chi E_h =(\chi \tilde{U}(t_0))^N \chi_{t_0} E_h+ \sum_{k=1}^{N} (\chi \tilde{U}(t_0))^k
(1-\chi)\chi_{t_0} E^0_h +O(h^\infty),
\end{equation}
where the cut-off
function $\chi_{t_0}\in C_c^\infty(X)$ is such that
\begin{equation*}d_X(supp~ \chi, supp (1-\chi_{t_0}))> 2|t_0|,
\end{equation*}
where $d_X$ denotes the Riemannian distance on $X$.

We shall require the following lemma. The proof of (i) is the same as that of Lemma \ref{abeille}, while the proof of (ii) essentially follows from point (3) of Hypothesis \ref{Guepard}.
\begin{lemme}\label{ghost}
(i) There exists $N_\chi\in \mathbb{N}$ such that for any $N\in \mathbb{N}$ if $\rho\in supp (\chi_{t_0})$ and $\Phi^N(\rho)\in supp (\chi)$, then for any $N_\chi\leq k \leq N-N_\chi$, we have $\Phi^{kt_0}(\rho)\in V_b$ for some $b\in B_1\cup B_2$.

(ii) If $\rho \in \mathcal{E}$ is such that $\Phi^{kt_0}(\rho)\in V_0$ for some $k\in \mathbb{N}$, but $\Phi^{(k+1)t_0}(\rho)\in V_b$ for some $b\in B_1\cup B_2$, then $\Phi^{k'}(\rho)$ is in $\mathcal{DE}_-$ (and hence in $V_0$) for any $k'\leq k$.
\end{lemme}

From  Lemma \ref{ghost}, we deduce that for any $k\geq 2N_\chi +2$, we have
\begin{equation}\label{apolinaris}
(\chi \tilde{U}(t_0))^k= \sum_{l=0}^{N_\chi+1}(\chi \tilde{U}(t_0))^{N_\chi+1} \Big{(} \sum_{\beta\in B^{k-2N_\chi-2 +l}} \tilde{U}_\beta \Big{)} \big{(}\chi \tilde{U}_0\big{)}^{N_\chi-l} + O_{L^2\rightarrow L^2}(h^\infty).
\end{equation}

For any $N\in \mathbb{N}\backslash\{0\}$, define
$\mathcal{B}_N\subset (B_1\cup B_2)^N$ by:
\begin{equation}\label{reveillon}
\begin{aligned}
\mathcal{B}_N:=& (B_1\cup B_2)^N ~~~~&\text{ if } N\leq 2\tins'+2,\\
\mathcal{B}_N:=& (B_1\cup B_2)^{\tins'+1} B_1^{N-2\tins'-2} (B_1\cup B_2)^{\tins'+1} ~~~~&\text{ otherwise}.
\end{aligned}
\end{equation}

\begin{lemme}\label{ibnsahl}
For any $N\geq 2\tins'+2$, for any $\beta\in (B_1\cup B_2)^N \backslash
\mathcal{B}_N$, we have
\begin{equation*}\|\tilde{U}_\beta\|_{L^2\rightarrow L^2}= O(h^\infty).
\end{equation*}
\end{lemme}
\begin{proof}
Let $\beta\in (B_1\cup B_2)^N\backslash \mathcal{B}_N$. Then there exists $\tins'+2\leq k \leq N-\tins'+2$ such that $\beta_k\in B_2$.
Recall from Remark \ref{animals} that $\tins'$ is such that for any $\rho\in V_{\beta_k}$, we have $\Phi^{\tins' t_0}(\rho)\in V_0\backslash \big{(}\bigcup_{b\in B_2} V_b\big{)}$ or $\Phi^{-\tins' t_0}(\rho)\in V_0\backslash \big{(}\bigcup_{b\in B_2} V_b\big{)}$. The result then follows from Lemma \ref{theclash}.
\end{proof}

Equation (\ref{apolinaris}) may then be rewritten as
\begin{equation}\label{apolinaris2}
(\chi \tilde{U}(t_0))^k= \sum_{l=0}^{N_\chi+1}(\chi \tilde{U}(t_0))^{N_\chi+1} \Big{(} \sum_{\beta\in \mathcal{B}_{k-2N_\chi-2 +l}} \tilde{U}_\beta \Big{)} \big{(}\chi \tilde{U}_0\big{)}^{N_\chi-l} + O_{L^2\rightarrow L^2}(h^\infty)
\end{equation}
By summing over $k$ and reordering the terms, we get, for any $K\bel 2N_\chi+3\tins' +4$
\begin{equation}\label{apolinaris3}
\begin{aligned}
\sum_{k=0}^K(\chi \tilde{U}(t_0))^k&= \sum_{n=1}^{(K-3\tins'- N_\chi-4)} \sum_{l=0}^{N_\chi+1}(\chi \tilde{U}(t_0))^{N_\chi+1} \Big{(} \sum_{\beta\in \mathcal{B}_{n+R\tins'+2}} \tilde{U}_\beta \Big{)} \big{(}\chi \tilde{U}_0\big{)}^{l}\\
& - \sum_{n=K-2N_\chi-2}^{(K-3\tins'-N_\chi-4)} \sum_{l=0}^{(K-3\tins'-N_\chi-4-n)} (\chi \tilde{U}(t_0))^{N_\chi+1} \Big{(} \sum_{\beta\in \mathcal{B}_{n+3\tins'+2}} \tilde{U}_\beta \Big{)} \big{(}\chi \tilde{U}_0\big{)}^{l}\\
&+ \sum_{l=0}^{3\tins'+N_\chi+3} \big{(} \chi \tilde{U}(t_0)\big{)}^l + O_{L^2\rightarrow L^2}(h^\infty).
\end{aligned}
\end{equation}

Let us note that from Lemma \ref{pacal} and Hypothesis \ref{chaise}, for each $0\leq l \leq N_\chi$, there exists $\chi_l\in S^{comp}(X)$ such that
\begin{equation}\label{inv}
\big{(} \chi \tilde{U}_0\big{)}^{N_\chi-l} (1-\chi)\chi_{t_0} E_h^0 = \chi_l E_h^0+O(h^\infty).
\end{equation}
Let us introduce the notation
\begin{equation}\label{australia}
\overline{\chi} := \sum_{l=0}^{N_\chi+1} \chi_l.
\end{equation}

Thanks to equation (\ref{apolinaris3}), we can study the different terms in equation (\ref{canardsauvage}).
The first term in the right hand side of (\ref{canardsauvage}) may be bounded
by the following lemma.

\begin{lemme}\label{infinitif} Let $r>0$. We may find a constant
$M_{r}\geq 0$ such that for any $M > M_{r}$, for any $M_{r}|\log h|
\leq N \leq M |\log h|$, we have:
\begin{equation*}\|\big{(}\chi \tilde{U}(t_0)\big{)}^N
\chi_{t_0} E_h\|_{L^2}=O(h^r).
\end{equation*}

\end{lemme}
\begin{proof}
We use (\ref{apolinaris2}), Lemma
\ref{grenadine} and the topological pressure assumption to obtain:
\[
\begin{aligned}
\big{\|}\big{(}\chi \tilde{U}(t_0)\big{)}^N
\chi_{t_0} E_h\big{\|}_{L^2} &\leq  
 C \Big{\|}\sum_{\beta\in \mathcal{B}_{N-2N_\chi-2}} \tilde{U}_\beta
\chi_{t_0} E_h\Big{\|} +O(h^\infty)\\
& \leq C \sum_{\beta\in B_1^{N-2\tins-2N_\chi-4}} \big{\|}\tilde{U}_\beta
\chi_{t_0} E_h\big{\|}\\
&\leq C h^{-d/2} (1+\epsilon_0)^{N} \sum_{\beta\in B_1^{N-2\tins-2N_\chi-4}}
\prod_{j=1}^{N-2\tins-2}  \exp \big{[} \frac{1}{2} S_{t_0}
(V_{\beta_j})\big{]} \big{\|}\chi_{t_0} E_h\big{\|}\\
&\leq C  h^{-d/2} (1+\epsilon_0)^{N}  \Big{(}\sum_{b\in
B_1} \exp \big{[} \frac{1}{2} S_{t_0} (V_{b})\big{]}\Big{)}^N \big{\|}\chi_{t_0}
E_h\big{\|} \\
&\leq  C  h^{-d/2} (1+\epsilon_0)^{N} \exp \{ Nt_0
(\mathcal{P}(1/2) + 2N\epsilon_0) \} \big{\|}\chi_{t_0} E_h\big{\|}.
\end{aligned}
\]
By assumption, $E_h$ is a tempered distribution, so that $\|\chi_{t_0} E_h\|_{L^2}\leq C/h^{\mathrm{r}''}$.
Therefore
\begin{equation*}\Big{\|}\big{(}\chi \tilde{U}(t_0)\big{)}^N
\chi_{t_0} E_h\Big{\|}_{L^2} \leq C h^{-r''-d/2-\epsilon} \exp \{ Nt_0
(\mathcal{P}(1/2) + 2N\epsilon_0) \},
\end{equation*}
for some small $\epsilon$.
The lemma follows by taking $M_r$ large enough.
\end{proof}
Using Lemma \ref{infinitif}, and equation(\ref{apolinaris3}), we may rewrite equation (\ref{canardsauvage}) as
\begin{equation*}
\begin{aligned}
\chi E_h &= \sum_{n=1}^{M_r |\log h|} \sum_{l=0}^{N_\chi+1}(\chi \tilde{U}(t_0))^{N_\chi+1} \Big{(} \sum_{\beta\in \mathcal{B}_{n+3\tins'+2}} \tilde{U}_\beta \Big{)} \big{(}\chi \tilde{U}_0\big{)}^{l} (1-\chi)\chi_{t_0} E^0_h\\
& - \sum_{n=M_r|\log h|-N_\chi}^{M_r|\log h|} \sum_{l=0}^{M_r|\log h|-N_\chi-2-n} (\chi \tilde{U}(t_0))^{N_\chi+1} \Big{(} \sum_{\beta\in \mathcal{B}_{n+3\tins'+2}} \tilde{U}_\beta \Big{)} \big{(}\chi \tilde{U}_0\big{)}^{l}(1-\chi)\chi_{t_0} E^0_h\\
&+ \sum_{l=0}^{3\tins'+N_\chi+3} \big{(} \chi \tilde{U}(t_0)\big{)}^l(1-\chi)\chi_{t_0} E^0_h + O_{L^2}(h^r).
\end{aligned}
\end{equation*}
The second term may be bounded by $O(h^r)$ thanks to Lemma \ref{infinitif}. By using equations (\ref{inv}) and (\ref{australia}), we get

\begin{equation}\label{canardsauvage2}
\begin{aligned}
\chi E_h 
&= \sum_{n=1}^{M_r |\log h|}(\chi \tilde{U}(t_0))^{N_\chi+1} \Big{(} \sum_{\beta\in \mathcal{B}_{n+3\tins'+2}} \tilde{U}_\beta \Big{)}\overline{\chi} E^0_h\\
&+ \sum_{l=0}^{3\tins'+N_\chi+3} \big{(} \chi \tilde{U}(t_0)\big{)}^l(1-\chi)\chi_{t_0} E^0_h + O_{L^2}(h^r)
\end{aligned}
\end{equation}

\subsection{Evolution of the WKB states}
\subsubsection{Construction of $\tilde{\mathcal{B}}_0$}
From now on, we fix $b\in B_1$ and $r\bel 1$. We may write
\begin{equation}
\mathcal{U}_b\Pi_b\sum_{l=0}^{3\tins'+N_\chi+3}\big{(}\chi \tilde{U}(t_0)\big{)}^l(1-\chi)\chi_{t_0} E_h^0 = 
\sum_{l=0}^{N_\chi+ 3\tins'+3}\sum_{\beta\in B^l} \mathcal{U}_b\Pi_b U_\beta^\chi(1-\chi)\chi_{t_0} E_h^0,
\end{equation}
where we have used the notation 
\begin{equation}\label{petanque}
U_\beta^\chi = \chi  \tilde{U}_{\beta_l} \chi... \chi \tilde{U}_{\beta_0}.
\end{equation}

Note that each of the $\mathcal{U}_b\Pi_b U^\chi_\beta$ is a Fourier Integral Operator from $L^2(X)$ to $L^2(\mathbb{R}^d)$. Thanks to Corollary \ref{marianne}, we may use Lemma \ref{pacal} to describe the action of each of these Fourier Integral Operators on the Lagrangian state $(1-\chi)\chi_{t_0} E_h^0$.
If we denote by $\tilde{\mathcal{B}}_0$ the set $\bigcup_{l=0}^{N_\chi+ 3\tins'+3} B^l$, we may write
\begin{equation}
\mathcal{U}_b\Pi_b\sum_{l=0}^{N_\chi+ 3\tins'+3}\big{(}\chi \tilde{U}(t_0)\big{)}^l(1-\chi)\chi_{t_0} E_h^0 = 
\sum_{\beta\in \tilde{\mathcal{B}}_0} e_{0,\beta,b},
\end{equation}
where $e_{0,\beta,b}(y^b)=e^{\phi_{0,\beta,b}(y^{\rho_b})/h}a_{0,\beta,b}(y^{\rho_b};h)$, with $a_{0,\beta,b}$ and $\phi_{0,\beta,b}$ as in the statement of Theorem \ref{ibrahim}.

Let us now consider the other terms on the right-hand side of equation (\ref{canardsauvage2}), which will be indexed by $\tilde{\mathcal{B}}_n$, $n\geq 1$.
\subsubsection{Evolution in the intermediate region}\label{dominus}
Let $n\geq 1$. and let  $\beta\in \mathcal{B}_{n+3\tins'+2}$. By definition of $\mathcal{B}_{n+3\tins'+2}$, we have
$\beta_i\in B_1$ for $\tins'+1\leq i\leq n+2\tins'+1$.

According to Theorem \ref{Cyril},
$\Phi^{2\tins'+1,t_0}_\beta (\Lag)$ consists of a
single Lagrangian manifold, which is $\ins$-unstable in the symplectic coordinates in $V_{\beta_{2\tins'+1}}$.

Therefore, we may say that $\tilde{U}_{\beta_0...\beta_{2\tins'+1}}
\big{(}\overline{\chi} E_h^0\big{)}$ is a Lagrangian state associated
to the Lagrangian manifold
$\Phi^{2\tins'+1,t_0}_\beta (\Lag)$. Thanks to Lemma \ref{diderot}, we may use Lemma
\ref{pacal}, to write:

\begin{equation*}
\big{(}\mathcal{U}_{\beta_{2\tins'+1}} \Pi_{\beta_{2\tins'+1}} \tilde{U}_{\beta_0...\beta_{2\tins'+1}}
\big{(}\overline{\chi} E_h^0\big{)}\big{)}(y^{\rho_{\beta_{2\tins'+1}}})=a(y^{\rho_{\beta_{2\tins'+1}}};h)e^{i\phi(y^{\rho_{\beta_{2\tins'+1}}})/h}
\end{equation*}
for some $a\in S^{comp}(\mathbb{R}^d)$.

\subsubsection{Propagation of Lagrangian states close to the trapped set}\label{consul}

To lighten the notations, let us write
$\hat{n}:=n+2\tins'+1$.

 For each $2\tins'+1\leq k\leq \hat{n}$, we write
\begin{equation*}T_{\beta_{k'+1},\beta_{k'}} := \mathcal{U}_{\beta_{k'+1}}
\tilde{U}_{\beta_{k'+1}} \mathcal{U}^*_{\beta_k'}.
\end{equation*}

$T_{\beta_{k'+1},\beta_{k'}}$ is an operator quantising the map
$\kappa_{\beta_{k'},\beta_{k'+1}}$ obtained by expressing $\Phi^{t_0}$ in
the coordinates $(y^{\beta_{k'}},\eta^{\beta_{k'}}) \mapsto
(y^{\beta_{k'+1}},\eta^{\beta_{k'+1}})$. It is of the form (\ref{blabla}).

We will write
\begin{equation*}T_{\beta}^{2\tins'+1,\hat{n}}:= T_{\beta_{\hat{n}},\beta_{\hat{n}}}
\circ...\circ
T_{\beta_{2\tins'+2},\beta_{2\tins'+1}}.
\end{equation*}

Thanks to Remark \ref{napoleon}, we may apply Proposition \ref{haroun} to
describe the action of $T_{\beta}^{2\tins'+1,\hat{n}}$ on the
Lagrangian state $\mathcal{U}_{\beta_{2\tins'}+1} \tilde{U}_{\beta_0...\beta_{2\tins'+1}} \big{(}\overline{\chi} E_h^0\big{)}$.
Note that \begin{equation*}
T_{\beta}^{2\tins'+1,\hat{n}}\mathcal{U}_{\beta_{2\tins}+1} \tilde{U}_{\beta_0...\beta_{2\tins+1}}= \mathcal{U}_{\beta_{\hat{n}}}
\tilde{U}_{\beta_0...\beta_{\hat{n}}}
\end{equation*}

We obtain that
$
\mathcal{U}_{\beta_{\hat{n}+1}} \Pi_{\beta_{\hat{n}+1}} \tilde{U}_{\beta_0...\beta_{\hat{n}}}
\big{(}\overline{\chi} E_h^0\big{)}) = e_{\hat{n},\beta}
$, with
\begin{equation} \label{Gainsbourg}
e_{\hat{n},\beta}(y)= a^{\hat{n},\beta}(y)
e^{i \phi_{\hat{n},\beta}(y)/h} ,~~~~~~ y\in \mathbb{R}^d.
\end{equation}

In the notations of Section \ref{bismark}, we have by Remark
\ref{napoleon} that for any
$\tins'+1\leq k'\leq \hat{n}$ $D_{k'}=S_T(V_{\beta_{k'}})
\big{(}1+O(\epsilon^\mathsf{p})\big{)}< 1$.
 We therefore set
\begin{equation}\label{goodtimes}
J_{\beta_{\tins'+1}...\beta_{\hat{n}}}:= \prod_{k'=\tins+1}^{\hat{n}} \Big{(}S_{t_0}(V_{\beta_{k'}})
\big{(}1+O(\epsilon^\mathsf{p})\big{)}\Big{)}.
\end{equation}

Thanks to equation (\ref{republique}) in Proposition \ref{haroun} and equation (\ref{highway78}), we obtain for any $\ell\in \mathbb{N}$:

\begin{equation}\label{republique2}
\|a^{\hat{n},\beta}\|_{C^\ell}\leq \big{(}1+ C_\ell h \big{)} C'_{\ell}J_{\beta_{\tins'+1}...\beta_{\hat{n}}}(\hat{n}+1)^{\ell},
\end{equation}
for some constants $C_{\ell},C'_\ell$.

\subsubsection{End of the propagation}\label{lizard}
Using equation (\ref{canardsauvage2}) and the results of the previous subsection, we have

\begin{equation}\label{canardsauvage3}
\begin{aligned}
\chi E_h &= \sum_{n=1}^{M_r |\log h|} (\chi \tilde{U}(t_0))^{N_\chi+1} \Big{(}
\sum_{\beta\in \mathcal{B}_{n+3\tins'+2}}\tilde{U}_{\beta_{\hat{n}}...\beta_n}
\mathcal{U}^*_{\beta_{\hat{n}}} e_{\hat{n},\beta}\Big{)}\\
&+\sum_{l=0}^{N_\chi+ 3\tins'+3}\big{(}\chi \tilde{U}(t_0)\big{)}^l(1-\chi)\chi_{t_0} E_h^0 +O_{L^2}(h^r),
\end{aligned}
\end{equation}
with
\begin{equation*}
\mathcal{U}_b\Pi_b\sum_{l=0}^{N_\chi+ 3\tins'+3}\big{(}\chi \tilde{U}(t_0)\big{)}^l(1-\chi)\chi_{t_0} E_h^0 = 
\sum_{\beta\in \tilde{\mathcal{B}}_0} e_{0,\beta,b}.
\end{equation*}

To finish the proof, we have to apply $\mathcal{U}_b \Pi_b (\chi \tilde{U}(t_0))^{N_\chi+1} \tilde{U}_{\beta_{\hat{n}}...\beta_{n}}
\mathcal{U}^*_{\beta_{\hat{n}}}$ to $e_{\hat{n},\beta}$.

To do this, one should once again decompose the propagator, and study
 \begin{equation}\label{tutiresoutupointes}
\sum_{\beta'\in B^{N_\chi+1}}\mathcal{U}_b \Pi_b  U_{\beta'}^{\chi} \tilde{U}_{\beta_{\hat{n}}...\beta_{n}}
\mathcal{U}^*_{\beta_{\hat{n}}}e_{\hat{n},\beta},
\end{equation}
with $U_{\beta'}^\chi$ as in (\ref{petanque}).
To analyse each of the terms on the right-hand side of (\ref{tutiresoutupointes}), we use once again Lemma \ref{pacal} (the lemma may be
applied, thanks to Theorem \ref{Cyril} and to Lemma \ref{diderot}).

We obtain that 
\begin{equation}\label{sheep}
\mathcal{U}_b \Pi_b  U_{\beta'}^{\chi} \tilde{U}_{\beta_{\hat{n}}...\beta_{n}}
\mathcal{U}^*_{\beta_{\hat{n}}}e_{\hat{n},\beta}(y) = a^{n,\beta,\beta'}(y)
e^{i \phi_{n,\beta,\beta'}(y)/h} ,~~~~~~ y\in \mathbb{R}^d,
\end{equation}
and thanks to equation (\ref{republique2}), we get
\begin{equation}\label{republique3}
\|a^{n,\beta,\beta'}\|_{C^\ell}\leq \big{(}1+ C_\ell h \big{)} C'_{\ell} J_{\beta_{\tins'+1}...\beta_{\hat{n}}}(\hat{n}+1)^{\ell},
\end{equation}
for some constants $C_{\ell},C'_\ell$.

For any $n\geq 1$, we write 
\begin{equation}\label{BN}\tilde{\mathcal{B}}_n= \mathcal{B}_{n+3 \tins'+2} \times B^{N_\chi+1}.
\end{equation}
As announced, the cardinal of $\tilde{\mathcal{B}}_n$ grows exponentially with $n$.
If $\beta= (\beta',\beta'')\in \tilde{\mathcal{B}}_n$ with $\beta \in \mathcal{B}_{n+2 \tins'+1}$, we define

\begin{equation*}
\begin{aligned}
a_{n,\beta,b} &=  a^{N_{n+2N_\chi+2,l},\beta,\beta'}\\
\phi_{n,\beta,b} &= \phi_{N_{n+2N_\chi+2,l},\beta,\beta'}.
\end{aligned}
\end{equation*}

With these notations, combining (\ref{canardsauvage3}) with (\ref{sheep}) gives us the decomposition (\ref{greenday}).

The key point to obtain estimate (\ref{sheriff}) is to notice that for any $N\geq \tins'+1$, we have thanks to (\ref{rainy}) 
\begin{equation}\label{presheriff}
\begin{aligned}
\sum_{\beta_{\tins'+1}...\beta_{\hat{N}}\in B_1^{N-\tins'-1}} J_{\beta_{\tins'+1}...\beta_{\hat{N}}} &= \Big{(}\sum_{b\in B_1} S_{t_0} (V_b)\big{(}1+O(\epsilon^\mathsf{p})\big{)}\Big{)}^{N-\tins'-1}\\
&\leq \exp \Big{[} (N-\tins'-1) (t_0 \mathcal{P}(1/2) \big{(}1+O(\epsilon^\mathsf{p})\big{)}\Big{]}.
\end{aligned}
\end{equation}
By applying (\ref{presheriff}) for $N= N_{n+2N_\chi+2,l}$, and combining it with (\ref{republique3}), we get (\ref{sheriff}).
\end{proof}
Note that, although the statement of Theorem \ref{ibrahim} describes the generalized eigenfunctions $E_h$ only very close to the trapped set, equation (\ref{canardsauvage3}) can be used to describe $E_h$ in any compact set, though in a less explicit way.

Using the estimate (\ref{sheriff}) as well as the fact that $\|\chi \tilde{U}(t_0)\|_{L^2\rightarrow L^2} \leq 1$ and $\|\mathcal{U}_b\|_{L^2\rightarrow L^2}\leq 1$, we deduce Theorem \ref{davies}.

\section{Semiclassical measures}\label{SC}
The main ingredient in the proof of Corollary \ref{blacksabbath} is  non-stationary phase. Let us recall the estimate we will use, and which can be proven by integrating by parts.

Let $a, \phi\in S^{comp}(X)$.
We consider the oscillatory integral:
$$I_h(a,\phi):= \int_{X} a(x) e^{\frac{i\phi(x,h)}{h}} \mathsf{d}x.$$
\begin{proposition}\label{nonstat}
Let $\epsilon>0$. Suppose that there exists $C>0$ such that, $\forall x\in spt(a), \forall  0<h<h_0$, $|\partial \phi(x,h)|\geq C h^{1/2-\epsilon}$. Then
$$I_h(a,\phi)=O(h^\infty).$$
\end{proposition}

We shall only give a sketch of proof here, and refer to \cite[\S7.7]{hormander35001analysis} for more details.

\begin{proof}[Sketch of proof]
To prove this result, we simply integrate by parts, noting that
$$I_h(a,\phi) = \frac{h}{i}\int_X \frac{a}{|\partial \phi|^2} \partial \phi \cdot \partial (e^{\frac{i\phi(x,h)}{h}}) \mathsf{d}x.$$
Hence, when we integrate by parts, the worst term in the integrand will involve second derivatives of $\phi$ times $\frac{h}{|\partial \phi|^2}$, and will therefore be a $O(h^{2\epsilon})$ by assumption. By integrating by parts more times, we will gain a factor $h^{2\epsilon}$ every time, so that $I_h(a,\phi)$ is actually a $O(h^\infty)$.
\end{proof}
Note that the sketch of proof above tells us that, if we could say that when $\partial \phi(x,h)$ is small, then the higher derivatives of $\phi$ are small as well, i.e., if we had
$$\forall k\geq 2, \exists C_k \text{ such that } |\partial^k \phi(x,h)|\leq C_k |\partial \phi(x,h)|,$$
then we would have $I_h(a,\phi)= O(h^\infty)$ provided $|\partial \phi(x,h)|\geq C h^{1-\epsilon}$. However, it is not clear that we can estimate the higher derivatives of the phase functions which appear in this section.

\subsection{Distance between the Lagrangian manifolds}
To take advantage of Proposition \ref{nonstat}, we need a lower bound on the distance between the Lagrangian manifolds which make up $\Phi^{n,t_0}(\Lag)\cap V_b$. To prove such a lower bound, let us first state an elementary topological lemma.

\begin{lemme}\label{vittel}
There exists $c_0\bel 0$ such that for any $\rho,\rho'\in T^*X_0\cap \mathcal{E}$ such that $d(\rho,\rho')< c_0$, there exists $b\in B$ such that $\rho,\rho'\in V_b$.
\end{lemme}

\begin{proof}
Suppose for contradiction that for any $\epsilon\bel 0$, there exists $\rho_\epsilon, \rho'_\epsilon$ such that $d(\rho_\epsilon,\rho'_\epsilon)< \epsilon$ and such that for all $b\in B$ such that $\rho_\epsilon\in V_b$, we have $y_\epsilon\notin V_b$. By compactness of $T^*X_0\cap \mathcal{E}$, we may suppose that $\rho_\epsilon$ converges to some $\rho$. We then have $\rho'_\epsilon \longrightarrow x$, and if $b\in B$ is such that $\rho\in V_b$, then and $\rho_\epsilon,\rho'_\epsilon\in V_b$ for $\epsilon$ small enough, a contradiction.
\end{proof}

We may now state our lower bound on the distance between the Lagrangian leaves which make up  $\Phi^{n,t_0}(\Lag)\cap V_b$.

Let $N\in \mathbb{N}$, $\beta \in B^N$ and $b\in B_1$. 
The set $\Phi_\beta^{n,t_0}(\Lag) \cap V_b$ may be written in the form $\{(y^{\rho_b}, \partial\tilde{\phi}_{n,\beta,b}(y^{\rho_b})\}$, for some smooth function $\tilde{\phi}_{N,\beta,b}$.

For any
$\beta\in B^N$,$\beta'\in B^{N'}$, let us denote by
$\sigma(\beta,\beta'):= \max(N-\tau(\beta),N'-\tau(\beta'))$, with $\tau(\beta)$ defined as in (\ref{rollingstone}).

\begin{proposition}\label{velvetunderground}
There exist constants $C_1',C'_2\bel 0$ such that for any $N,N'\in \mathbb{N}$, for any $\beta\in B^N,\beta'\in B^{N'}$, for any $b\in B_1$ and for any $y^{\rho_b}$, we have either $\partial \tilde{\phi}_{N,\beta,b}(y^{\rho_b}) = \partial \tilde{\phi}_{N',\beta',b}(y^{\rho_b}) $ or 
\begin{equation*}
|\partial \tilde{\phi}_{N,\beta,b} (y^{\rho_b})- \partial \tilde{\phi}_{N',\beta',b}(y^{\rho_b})|\geq C'_1 e^{C'_2 \sigma(\beta,\beta')}. 
\end{equation*}
\end{proposition}

\begin{proof}
Since $T^*X_0\cap \mathcal{E}$ is compact, we may find a constant $C\bel 0$ such that for any $\rho,\rho'\in \mathcal{E}\cap T^*X_0$, 
\begin{equation}\label{flowers}
d(\Phi^t(\rho),\Phi^t(\rho'))\leq e^{Ct} d(\rho,\rho'),
\end{equation}
where $d$ is the distance on the energy layer which we introduced in section \ref{averse}.

Let $b\in B_1$, and $y^{\rho_b}\in D_{\beta,b}\cap D_{\beta',b}$ be such that $\partial \tilde{\phi}_{N,\beta,b}(y^{\rho_b}) \neq \partial \tilde{\phi}_{N',\beta',b}(y^{\rho_b}) $.
 Let us denote by $\rho$ the point $(y^{\rho_b};\partial \tilde{\phi}_{N,\beta,b}(y^{\rho_b}))$ and by $\rho'$ the point $(y^{\rho_b};\partial \tilde{\phi}_{N',\beta',b}(y^{\rho_b}))$.
 
  We claim that there exists $0\leq k\leq \sigma(\beta,\beta')$ such that for each $b'\in B$, if $\Phi^{-k t_0}(\rho)\in V_{b'}$, then $\Phi^{-k t_0}(\rho')\notin V_b$.
Indeed, if no such $k$ existed, then for each $k$, there would exist $b_k\in B$ such that $\Phi^{-k t_0}(\rho)\in V_{b_k}$ and $\Phi^{-k t_0}(\rho')\in V_{b_k}$ for each $0\leq k\leq \sigma(\beta,\beta')$. We would then have $\rho\in \Phi_{\beta''}^{\max(N,N'), t_0}(\Lag)$ and $\rho'\in \Phi_{\beta''}^{\max(N,N'), t_0}(\Lag)$ for some sequence $\beta''$ built by possibly adding some $0$'s in front of the sequences $\beta$ and $\beta'$. This would contradict the statement of Corollary \ref{marianne}.

Thanks to Lemma \ref{vittel}, we deduce from this that there exists $0\leq k \leq \sigma (\beta,\beta')$ such that $d(\Phi^{-kt_0}(\rho),\Phi^{-k t_0} (\rho'))\geq c_0$. 
Combining this fact with equation (\ref{flowers}), we get
\begin{equation*}
d(\rho,\rho')\geq c_0 e^{-C \sigma(\beta,\beta')}
\end{equation*}
Using the fact that all metrics are equivalent on a compact set, we may compare $d(\rho,\rho')$ with $|\partial \tilde{\phi}_{N,\beta,b} (y^{\rho_b})- \partial \tilde{\phi}_{N',\beta',b}(y^{\rho_b})|$ and we deduce from this the proposition.
\end{proof}

Using the definition of $\tilde{\mathcal{B}}_n$, we deduce the following result about the functions $\phi_{n,\beta,b}$ in the statement of Theorem \ref{ibrahim}.
\begin{corolaire}\label{petrole}
There exist constants $C'_1,C'_2\bel 0$ such that for any $n,n'\in \mathbb{N}$, for any $\beta\in \tilde{\mathcal{B}}_n, \beta'\in \tilde{\mathcal{B}}_{n'}$, for any $b\in B_1$ and for any $y^{\rho_b}$,
we have either $\partial \phi_{n,\beta,b}(y^{\rho_b}) = \partial \phi_{n',\beta',b}(y^{\rho_b}) $ or 
\begin{equation*}
|\partial \phi_{n,\beta,b} (y^{\rho_b})- \partial \phi_{n',\beta',b}(y^{\rho_b})|\geq C'_1 e^{C'_2 \min(n,n')}. 
\end{equation*}
\end{corolaire}
\subsection{Proof of Corollary \ref{blacksabbath}}
We shall now prove Corollary \ref{blacksabbath}, which we recall.
 \begin{corolaire}
There exists a constant $0< c \leq 1$ and functions $e_{n,\beta,b}$ for $n\in \mathbb{N}$, $\beta\in \tilde{\mathcal{B}}_n$ and $b\in B_1$ such that for any $a\in C_c^\infty(T^*X)$ and for any $\chi\in C_c^\infty(X)$, we have
\begin{equation*}\langle Op_h(\pi_b^2 a) \chi E_h,  \chi E_h\rangle = \int_{T^*X} a(x,v)
\mathrm{d}\mu_{b,\chi}(x,v) +
O(h^c),
\end{equation*}
with \begin{equation*}\mathrm{d}\mu_{b,\chi}(\kappa_b^{-1}(y^{\rho_b},\eta^{\rho_b})) = \sum_{n=0}^{\infty} \sum_{\beta\in
\tilde{\mathcal{B}}_n}  e_{n,\beta,b}(y^{\rho_b}) \delta_{\{\eta^{\rho_b}=\partial
\phi_{j,n}(y^{\rho_b})\}} d y^{\rho_b},
\end{equation*}
The functions $e_{n,\beta,b}$ satisfy the estimate (\ref{sheriff}).
\end{corolaire}
\begin{proof}
Take any small $\epsilon \bel 0$, and set 
\begin{equation*}
\begin{aligned}
M&:= \frac{1}{2 C_2'}-\epsilon,\\
c&:= (M - \epsilon) \mathcal{P}(1/2) = \frac{\mathcal{P}(1/2)}{2C'_2}-\epsilon
\end{aligned}
\end{equation*}
where $C_2'$ comes from Corollary \ref{petrole}.

Let $a\in C_c^\infty(T^*X)$, $\chi\in C_c^\infty(X)$ and $b\in B_1$. Using the fact that $Op_h(\mathrm{ab})= Op_h(\mathrm{a}) Op_h(\mathrm{b}) +O_{L^2\rightarrow L^2}(h)$ for any $\mathrm{a},\mathrm{b}\in S^{comp}(X)$, the self-adjointness of $\Pi_b$, and the unitarity of $\mathcal{U}_b$ on the micro-support of $\Pi_b$, we see that we have
\begin{equation*}
\begin{aligned}
\langle Op_h(\pi_b^2 a) \chi E_h, \chi E_h \rangle_{L^2(X)} &=
\langle Op_h(a)\Pi_b \chi E_h, \Pi_b \chi E_h \rangle_{L^2(X)} + O(h)\\
 &= \langle \mathcal{U}_b Op_h(a) \mathcal{U}^*_b \mathcal{U}_b\Pi_b E_h, \mathcal{U}_b \Pi_b \chi E_h \rangle_{L^2(X)} + O(h).
\end{aligned}
\end{equation*}

Now, using Egorov's Theorem (\cite[Theorem 11.1]{Zworski_2012}), we know that
\begin{equation*}
\mathcal{U}_b Op_h(a) \mathcal{U}^*_b \mathcal{U}_b \Pi_b= Op_h (a_b) \mathcal{U}_b \Pi_b + O_{L^2(X)\rightarrow L^2(\mathbb{R}^d)} (h^\infty),
\end{equation*}
where $a_b= a\circ \kappa_b +O_{L^2}(h)$. Using decomposition (\ref{greenday}), we have
\begin{equation}\label{groove2}
\begin{aligned}
&\langle Op_h(\pi_b^2 a) \chi E_h, \chi E_h \rangle_{L^2(X)}\\
&= \sum_{n=0}^{\lfloor
M_c |\log h|\rfloor}
\sum_{\beta\in \tilde{\mathcal{B}}_n} \Big{\langle} Op_h(a_b) \big{[}e^{i \phi_{n,\beta,b}/h}
a_{n,\beta,b}\big{]}, \sum_{n'=0}^{\lfloor
M_c |\log h|\rfloor}
\sum_{\beta'\in \tilde{\mathcal{B}}_{n'}} e^{i \phi_{n',\beta',b}/h}
a_{n',\beta',b} \Big{\rangle} + O(h^c),
\end{aligned}
\end{equation}
But thanks to estimate (\ref{sheriff}), 
$$\sum_{n=0}^{\lfloor
M_c |\log h|\rfloor}
\sum_{\beta\in \tilde{\mathcal{B}}_{n}} e^{i \phi_{n,\beta,b}/h}
a_{n,\beta,b}=\sum_{n=0}^{\lfloor
M |\log h|\rfloor}
\sum_{\beta\in \tilde{\mathcal{B}}_{n}} e^{i \phi_{n,\beta,b}/h}
a_{n,\beta,b} +O_{L^2}(h^c),$$
so that
\begin{equation}\label{groove}
\begin{aligned}
&\langle Op_h(\pi_b^2 a) \chi E_h, \chi E_h \rangle_{L^2(X)}\\
&= \sum_{n=0}^{\lfloor
M |\log h|\rfloor}
\sum_{\beta\in \tilde{\mathcal{B}}_n} \Big{\langle} Op_h(a_b) \big{[}e^{i \phi_{n,\beta,b}/h}
a_{n,\beta,b}\big{]}, \sum_{n'=0}^{\lfloor
M |\log h|\rfloor}
\sum_{\beta'\in \tilde{\mathcal{B}}_{n'}} e^{i \phi_{n',\beta',b}/h}
a_{n',\beta',b} \Big{\rangle} + O(h^c),
\end{aligned}
\end{equation}

We now want to fix a $n\leq M |log h|$ and a $\beta\in \tilde{\mathcal{B}}_n$, and to analyse the behaviour of $$ \Big{\langle} Op_h(a_b) \big{[}e^{i \phi_{n,\beta,b}/h}
a_{n,\beta,b}\big{]}, \sum_{n'=0}^{\lfloor
M |\log h|\rfloor}
\sum_{\beta'\in \tilde{\mathcal{B}}_n} e^{i \phi_{n',\beta',b}/h}
a_{n',\beta',b}\Big{\rangle}.$$

Let us denote by $Y_{n',\beta'}=\{ y^{\rho_b}\in spt(\phi_{n,\beta,b})\cap spt(\phi_{n',\beta',b}); \partial\phi_{n',\beta',b}(y^{\rho_b})= \partial\phi_{n,\beta,b}(y^{\rho_b})\}$. We have

\begin{equation}\label{amsterdam}
\begin{aligned}
& \Big{\langle}  Op_h(a_b) \big{[}e^{i \phi_{n,\beta,b}/h}
a_{n,\beta,b}\big{]}, e^{i \phi_{n',\beta',b}/h}
a_{n',\beta',b}  \Big{\rangle} \\
&= \int_{Y_{n',\beta'}} \Big{(}Op_h(a_b) \big{[}e^{i \phi_{n,\beta,b}/h}
a_{n,\beta,b}\big{]}\Big{)}(y^{\rho_b}) e^{i \phi_{n',\beta',b}(y^{\rho_b})/h}
a_{n',\beta',b} (y^{\rho_b};h) dy^{\rho_b}\\
 &+ \int_{\mathbb{R}^d\backslash Y_{n',\beta'}} \Big{(}Op_h(a_b) \big{[}e^{i \phi_{n,\beta,b}/h}
a_{n,\beta,b}\big{]}\Big{)}(y^{\rho_b}) e^{i \phi_{n',\beta',b}(y^{\rho_b})/h}
a_{n',\beta',b} (y^{\rho_b};h) dy^{\rho_b}.
\end{aligned}
\end{equation}

Recall that the integrals are well defined, because the phase functions are well-defined in a neighbourhood of the functions $a_{n,\beta,b}$.

The second term on the right hand side of (\ref{amsterdam}) is a $O(h^\infty)$. Indeed, the image of a Lagrangian state by a pseudo-differential operator is still a Lagrangian state with the same phase. Therefore, we are computing scalar products between Lagrangian states with respective phases $\phi_{n,\beta,b}$ and $\phi_{n',\beta',b}$.

Now, by the choice of $M$, and by Corollary \ref{petrole}, we know that for each $y^{\rho_b}\in \mathbb{R}^d\backslash Y_{n',\beta'}$ we have $|\partial \phi_{n,\beta,b}(y^{\rho_b})- \partial \phi_{n',\beta',b}(y^{\rho_b})|\geq C h^{1/2+\epsilon}$ for some $C,\epsilon \bel 0$. Hence by Proposition \ref{nonstat}, we deduce that the second term on the right hand side of (\ref{amsterdam}) is a $O(h^\infty)$.

We should now try to understand the properties of the set $Y_{n',\beta'}$.

First of all, $Y_{n',\beta'}$ is an open set. Indeed, if $y^{\rho_b}\in Y_{n',\beta'}$, this means that the point $\rho =(y^{\rho_b}, \partial \phi_{n,\beta,b}(y^{\rho_b}))$ (in the coordinates centred at $\rho_b$) belongs to $\Phi_\beta^{n,t_0}(\Lag)$ as well as to $\Phi_{\beta'}^{n',t_0}(\Lag)$ in the notations of Proposition \ref{velvetunderground}. Suppose for simplicity that $n=n'$ (the general case works the same). Then the condition $y^{\rho_b}\in Y_{n',\beta'}$ simply means that at each intermediate time $k$, $\Phi^{n-k}(\rho)$ was both in $V_{\beta_k}$ and in $V_{\beta'_{k}}$. This clearly an open condition.

On the other hand, by continuity of the phase functions, $Y_{n',\beta'}$ is a closed set. Therefore, $Y_{n',\beta'}$ consists of a certain number of connected components of the support of $\phi_{n',\beta'}$.

We know that the support of 
$a_{n',\beta',b}$ is included in the domain of definition of $\phi_{n',\beta',b}$. Therefore, some of the connected components of $spt(a_{n',\beta',b})$ may be included in $Y_{n',\beta'}$, while others are included in $\mathbb{R}^d\backslash Y_{n',\beta'}$, but none of them may intersect both sets. Therefore, if we set $a^{n,\beta}_{n',\beta',b} (y^{\rho_b}) = a_{n',\beta', b}(y^{\rho_b})$ if $y^{\rho_b}\in Y_{n',\beta'}$, 0 otherwise, then $a^{n,\beta}_{n',\beta',b}\in S$, and we have

\begin{equation*}
\begin{aligned}
&\Big{\langle} Op_h(a_b) \big{[}e^{i \phi_{n,\beta,b}/h}
a_{n,\beta,b}\big{]}, \sum_{n'=0}^{\lfloor
M |\log h|\rfloor}
\sum_{\beta'\in \mathcal{B}_n} e^{i \phi_{n',\beta',b}/h}
a_{n',\beta',b} \Big{\rangle} \\
&= \int_{\mathbb{R}^d} \Big{(} Op_h(a_b) \big{[}e^{i \phi_{n,\beta,b}/h}
a_{n,\beta,b}\big{]} \Big{)} (y^{\rho_b}) e^{-i \phi_{n,\beta,b}(y^{\rho_b})/h} \Big{(} \sum_{n'=0}^{\lfloor
M |\log h|\rfloor}
\sum_{\beta'\in \mathcal{B}_n}
a^{n,\beta}_{n',\beta',b} \Big{)} (y^{\rho_b}) \mathrm{d} y^{\rho_b}.
\end{aligned}
\end{equation*}

Let us write 
$$\tilde{a}_{n,\beta,b}:= \sum_{n'=0}^{\lfloor
M |\log h|\rfloor}
\sum_{\beta'\in \mathcal{B}_n}
a^{n,\beta}_{n',\beta',b}.$$ 
$\tilde{a}_{n,\beta,b}(y^{\rho_b})$ is the sum of all the symbols in the expansion (\ref{greenday}) having phase $\phi_{n,\beta,b}(y^{\rho_b})$. We see by the estimates (\ref{sheriff}) that $\tilde{a}_{n,\beta,b}$ satisfies the estimates (\ref{sheriff}) itself, and that
\begin{equation*}
\begin{aligned}
&\Big{\langle} Op_h(a_b) \big{[}e^{i \phi_{n,\beta,b}/h}
a_{n,\beta,b}\big{]}, \sum_{n'=0}^{\lfloor
M |\log h|\rfloor}
\sum_{\beta'\in \mathcal{B}_n} e^{i \phi_{n',\beta',b}/h}
a_{n',\beta',b} \Big{\rangle} \\
&= \int_{\mathbb{R}^d} \Big{(} Op_h(a_b) \big{[}e^{i \phi_{n,\beta,b}/h}
a_{n,\beta,b}\big{]} \Big{)} (y^{\rho_b}) e^{-i \phi_{n,\beta,b}(y^{\rho_b})/h} \tilde{a}_{n,\beta,b} (y^{\rho_b}) \mathrm{d} y^{\rho_b} +O(h^\infty).
\end{aligned}
\end{equation*}
We may then compute this expression using stationary phase, just as to compute the semiclassical measure of a Lagrangian state (see \cite[\S 5.1]{Zworski_2012}). We obtain
\begin{equation*}
\begin{aligned}
&\Big{\langle} Op_h(a_b) \big{[}e^{i \phi_{n,\beta,b}/h}
a_{n,\beta,b}\big{]}, \sum_{n'=0}^{\lfloor
M |\log h|\rfloor}
\sum_{\beta'\in \mathcal{B}_n} e^{i \phi_{n',\beta',b}/h}
a_{n',\beta',b} \Big{\rangle} \\
&= \int_{\mathbb{R}^{2n}} a_b \mathrm{d}\mu_{n,\beta,b}, 
\end{aligned}
\end{equation*}
where
\begin{equation*}
d\mu_{n,\beta,b} = a_{n,\beta,b}(y^{\rho_b}) \overline{\tilde{a}_{n,\beta,b}(y^{\rho_b})} \delta_{\{\eta^{\rho_b}=\partial \phi_{n,\beta,b}(y^{\rho_b})\}} \mathrm{d} y^{\rho_b}.
\end{equation*}

Summing over all $n,\beta$ and using equation (\ref{groove}), we obtain indeed that

\begin{equation*}\langle Op_h(\pi_b^2 a) E_h,  E_h\rangle = \int_{T^*X} a(x,\xi)
\mathrm{d}\mu_{b,\chi}(x,\xi) +
O(h^c),
\end{equation*}
with $(\kappa_b)^* \mu_{b,\chi} = \sum_{n=0}^\infty \sum_{\beta\in \mathcal{B}_n} \mu_{n,\beta,b}$, that is to say

\begin{equation*}\mathrm{d}\mu_{b,\chi}(\kappa_b^{-1}(y^{\rho_b},\eta^{\rho_b})) = \sum_{n=0}^{\infty} \sum_{\beta\in
\mathcal{B}_n}  e_{n,\beta,b}(y^{\rho_b}) 
\delta_{\{\eta^{\rho_b}=\partial
\phi_{n,\beta,b}(y^{\rho_b})\}} d y^{\rho_b},
\end{equation*}
where $e_{n,\beta,b}(y^{\rho_b}):=\lim_{h\rightarrow 0}\big{(}a_{n,\beta,b} \overline{\tilde{a}_{n,\beta,b}}\big{)}(y^{\rho_b})$.
This concludes the proof of Corollary \ref{blacksabbath}.
\end{proof}

\subsection{Construction of the measure $\mu^\xi$}\label{appendice}
In the introduction we defined the measure $\mu^\xi$ by
\begin{equation*}\int_{T^*\R^d} a \mathrm{d}\mu^\xi := \lim \limits_{t\rightarrow \infty} \int_{T^*\R^d} a\circ\Phi^{t}
\mathrm{d}\mu^\xi_0,
\end{equation*}
for any $a\in C_c^0(T^*\R^d)$. We will now give a sketch of the proof of why the hyperbolicity and transversality hypotheses, along with the assumption that $\mathcal{P}(1)<0$, imply that the above limit exists. 

Note that the assumption $\mathcal{P}(1)<0$ is really less restrictive than $\mathcal{P}(1/2)<0$. For instance, if we assume that the flow $(\Phi^t)$ is \emph{Axiom A}, that is to say, that the periodic orbits are dense in $K$, then \cite[\S 4.C]{bowen1975equilibrium} guaranties us that $\mathcal{P}(1)<0$.

Note that, if $a$ is non-negative, then $t\mapsto  \int_{T^*\R^d} a\circ\Phi^{t}
\mathrm{d}\mu^\xi_0$ is non-decreasing, so that we only have to show that this quantity is bounded.

If $\mu$ is a measure, we define $\Phi^t_* \mu$ by
$$\int_{T^*\R^d} a \mathrm{d}(\Phi^t_*\mu) :=  \int_{T^*\R^d} a\circ\Phi^{t}
\mathrm{d}\mu.$$
If $\pi\in C^\infty(T^*\R^d; [0,1])$ we define the measure $\pi\mu$ by
$$\int_{T^*\R^d} a \mathrm{d}(\pi\mu) :=  \int_{T^*X} a \pi
\mathrm{d}\mu.$$
\begin{remarque}
Note that if $\mu$ is the semiclassical measure associated to a Lagrangian state $\phi_h$, then $\pi \mu$ is the semiclassical measure associated to $ \sqrt{pi} \phi_h$, and, by Egorov's theorem, $\Phi^{t}_* \mu$ is the semiclassical measure associated to $U(t) \phi_h$. 
\end{remarque}

We shall use the functions $\pi_b$ from section \ref{partition}. If $\beta\in B^n$, we set 
$$\Phi_\beta\mu := \pi_{\beta_n}\Phi^{t_0}_*\Big{(}...\pi_{\beta_2}\Phi^{t_0}_*\Big{(}\pi_{\beta_1} \Phi^{t_0}_*\mu \Big{)}\Big{)}.$$

Let $\phi_h$ be a Lagrangian state associated to a Lagrangian manifold which is $\gamma$-unstable in the coordinates $(y^\rho,\eta^\rho)$, and let 
 $\mu$ be the semiclassical measure associated to $\phi_h$. The propagation $U_\beta \phi_h$ can be described using the methods of section \ref{bismark} along with the results of section \ref{houx}. In particular, we obtain, like in \cite[(7.12)]{NZ} that we may find $C,\epsilon>0$ such that for all $N\in \mathbb{N}$ and all $\beta\in B_1^N$, we have
$$\|U_\beta \phi_h\|_{L^2}\leq C(1+C\epsilon)^N \prod_{j=1}^{N} \exp \big{[} \frac{1}{2} S_{t_0} (V_{\beta_j})\big{]}.$$

We may deduce from this the following bound for the measure $\Phi_\beta \mu$. Note that this could also be deduced directly from the transport equations for measures, without using Schrödinger propagators and Egorov's theorem.

For any $a\in C_c^0(T^*X)$, if $\beta\in B_1^N$, we have that
\begin{equation*}
\langle\Phi_\beta \mu,a\rangle
\leq C(a) (1+C\epsilon)^N
\prod_{j=1}^{N} \exp \big{[} S_{t_0} (V_{\beta_j})\big{]}.
\end{equation*}

By possibly taking the sets $V_b$ smaller, we may ensure just like in section \ref{Benoit} that 
\begin{equation*}\sum_{b\in B_1} \exp \{ S_{t_0} (V_b) \} \leq \exp \{ t_0
(\mathcal{P}(1) + \epsilon) \} .
\end{equation*}

Therefore, we obtain that
\begin{equation}\label{dispersionmesure}
\sum_{\beta\in B_1^N} \langle \Phi_\beta \mu, a \rangle \leq C(a) \exp \big{[}- N t_0 (\mathcal{P}(1)-\epsilon)\big{]}.
\end{equation}

If we assume that the flow $(\Phi^t)$ is \emph{Axiom A}, that is to say, that the periodic orbits are dense in $K$, then \cite[\S 4.C]{bowen1975equilibrium} guaranties us that $\mathcal{P}(1)<0$.

Now, we have that 
$$\Phi^{Nt_0}_* \mu^\xi = \sum_{\beta\in \tilde{\mathcal{B}}^N} \Phi_\beta \mu^\xi,$$
and we may use (\ref{dispersionmesure}) along with the assumption that $\mathcal{P}(1)<0$ to show that, if $a$ is non-negative, $t\mapsto  \int_{T^*\R^d} a\circ\Phi^{t}
\mathrm{d}\mu^\xi_0$ is bounded.

Showing that $\mu^\xi$ is the semiclassical measure associated to $E_h$ follows from \cite[\S 5.1]{DG} (which relies on Egorov's theorem), along with estimate (\ref{japon}).

\bibliographystyle{alpha}
\bibliography{references}

\begin{thebibliography}{CDV85}

\bibitem[BBR10]{bonyburqramondminoration}
J.-F. Bony, N.~Burq, and T.~Ramond.
\newblock Minoration de la r{\'e}solvante dans le cas captif.
\newblock {\em Comptes Rendus Math{\'e}matique}, 348(23):1279--1282, 2010.

\bibitem[Bon14]{Yannick}
Y.~Bonthonneau.
\newblock Long time quantum evolution of observables.
\newblock {\em arXiv preprint arXiv:1411.5083}, 2014.

\bibitem[Bow75]{bowen1975equilibrium}
R.~Bowen.
\newblock {\em Equilibrium states and the ergodic theory of Anosov
  diffeomorphisms}, volume 470.
\newblock Springer, 1975.

\bibitem[CDV85]{CdV}
Y.~Colin De~Verdiere.
\newblock Ergodicit{\'e} et fonctions propres du laplacien.
\newblock {\em Communications in Mathematical Physics}, 102(3):497--502, 1985.

\bibitem[DG14]{DG}
S.~Dyatlov and C.~Guillarmou.
\newblock Microlocal limits of plane waves and {E}isenstein functions.
\newblock {\em Ann. Sci. {\'E}c. Norm. Sup}, 47(2):371--448, 2014.

\bibitem[DS99]{DSj}
M.~Dimassi and J.~Sj{\"o}strand.
\newblock {\em Spectral asymptotics in the semi-classical limit}.
\newblock Number 268. Cambridge university press, 1999.

\bibitem[Dya11]{Dcusp}
S.~Dyatlov.
\newblock Microlocal limits of {E}isenstein functions away from the unitarity
  axis.
\newblock {\em J. Spectr. Theory}, 2:181--202, 2011.

\bibitem[DZ16]{dyatlov2016mathematical}
S.~Dyatlov and M.~Zworski.
\newblock Mathematical theory of scattering resonances.
\newblock {\em Book in progress, http://math. mit. edu/dyatlov/res/(22 December
  2015, date last accessed)}, 2016.

\bibitem[GN14]{GN}
C.~Guillarmou and F.~Naud.
\newblock Equidistribution of {E}isenstein series for convex co-compact
  hyperbolic manifolds.
\newblock {\em Amer. J. Math.}, 136:445--479, 2014.

\bibitem[H{\"o}r]{hormander35001analysis}
L.~H{\"o}rmander.
\newblock {\em The analysis of linear partial differential operators. I, 256
  (1983)}.
\newblock Grundlehren Math. Wiss., Springer, Berlin Heidelberg.

\bibitem[Ing15]{Ing2}
M.~Ingremeau.
\newblock Distorted plane waves on manifolds of nonpositive curvature.
\newblock {\em arXiv preprint arXiv:1512.06818}, 2015.

\bibitem[KH95]{KH}
A.~Katok and B.~Hasselblatt.
\newblock {\em Introduction to the modern theory of dynamical systems}.
\newblock 1995.

\bibitem[Mel95]{Mel}
R.B. Melrose.
\newblock {\em Geometric {S}cattering {T}heory}.
\newblock Cambridge University Press, 0995.

\bibitem[NZ09]{NZ}
S.~Nonnenmacher and M.~Zworski.
\newblock Quantum decay rates in chaotic scattering.
\newblock {\em Acta Math.}, 203:149--233, 2009.

\bibitem[Shn74]{Shn}
A.I. Shnirelman.
\newblock Ergodic properties of eigenfunctions.
\newblock {\em Usp. Mat. Nauk.}, 29:181--182, 1974.

\bibitem[Sj{\"o}90]{Sj90}
J.~Sj{\"o}strand.
\newblock Geometric bounds on the density of resonances for semiclassical
  problems.
\newblock {\em Duke mathematical journal}, 60:1--57, 1990.

\bibitem[Zel87]{Zel}
S.~Zelditch.
\newblock Uniform distribution of eigenfunctions on compact hyperbolic
  surfaces.
\newblock {\em Duke mathematical journal}, 55(4):919--941, 1987.

\bibitem[Zwo12]{Zworski_2012}
M.~Zworski.
\newblock {\em Semiclassical Analysis}.
\newblock AMS, 2012.

\end{thebibliography}
\end{document}